\documentclass[11pt]{article}

\usepackage{amsmath,amssymb,amsthm,fullpage,mathrsfs,pgf,tikz,caption,subcaption,mathtools,mathabx}
\usepackage{comment}
\usepackage[ruled,vlined,noresetcount]{algorithm2e}
\usepackage{amsmath,amssymb,amsthm,mathtools}
\usepackage{thmtools}
\usepackage[utf8]{inputenc}
\usepackage[T1]{fontenc}
\usepackage{hyperref}
\usepackage{url}
\usepackage{booktabs}
\usepackage{amsfonts}
\usepackage{nicefrac}
\usepackage{microtype}
\usepackage{times}
\usepackage{bbm}
\usepackage{enumitem}
\usepackage{xcolor}
\usepackage{cleveref}
\usepackage{bm}
\usepackage{float}
\usepackage{tikz-cd}
\usepackage[margin=1in]{geometry}
\usetikzlibrary{trees}
\usepackage{braket}

\newtheorem{theorem}{Theorem}[section]

\newtheorem{lemma}[theorem]{Lemma}
\newtheorem{definition}[theorem]{Definition}

\newtheorem{remark}[theorem]{Remark}

\newtheorem{claim}[theorem]{Claim}

\newcommand{\FF}{\mathbb{F}}
\newcommand{\NN}{\mathbb{N}}
\newcommand{\RR}{\mathbb{R}}
\newcommand{\ZZ}{\mathbb{Z}}

\newcommand{\cB}{\mathcal{B}}
\newcommand{\cC}{\mathcal{C}}

\newcommand{\cS}{\mathcal{S}}

\newcommand{\cX}{\mathcal{X}}
\newcommand{\wt}{\widetilde}

\newcommand{\interior}[1]{%
  {\kern0pt#1}^{\mathrm{o}}%
}
\newcommand{\CNOT}{\mathrm{CNOT}}

\let\ker\relax
\let\im\relax
\DeclareMathOperator{\ker}{ker}
\DeclareMathOperator{\im}{im}

\DeclareMathOperator{\rk}{rk}

\DeclareMathOperator{\cl}{cl}

\DeclareMathOperator{\poly}{poly}

\DeclareMathOperator{\cone}{cone}

\newcommand{\red}[1]{{\color{red} #1}}

\newcommand{\blue}[1]{{\color{blue} #1}}

\def\F2{\mathbb{F}_2}
\def\a{\alpha}
\def\b{\beta}

\title{Simplified Quantum Weight Reduction with Optimal Bounds}
\author{ Min-Hsiu Hsieh\thanks{Foxconn Research, Taipei, Taiwan. Email: \texttt{min-hsiu.hsieh@foxconn.com}.} \and Xingjian Li\thanks{ Department of Computer Science and Technology, Tsinghua University, Beijing, China. Email: \texttt{lxj22@mails.tsinghua.edu.cn}.} \and Ting-Chun Lin\thanks{Department of Physics, University of California San Diego, CA, and Foxconn Research, Taipei, Taiwan. Email: \texttt{til022@ucsd.edu}.} }

\begin{document}

\sloppy

\maketitle

\begin{abstract}
Quantum weight reduction is the task of transforming a quantum code with large check weight into one with small check weight. This problem is important in practice, since low-weight measurements are necessary for reliable implementations of quantum error correction on physical hardware. It is also important theoretically, as it can be used to construct constant-locality versions of quantum locally testable codes, which may be relevant to the quantum PCP conjecture.

We give a streamlined geometric procedure for quantum weight reduction based entirely on coning and treating X and Z checks symmetrically, which simplifies Hastings' previous approach. In particular, given an arbitrary $[[n,k,d]]$ quantum code with weight $w$, our method produces a code with parameters $[[O(n w^2 \log w), k, \Omega(d w)]]$, check weight $5$, and qubit weight $6$; these bounds are optimal or close to optimal within the current geometric framework.

Applied to random dense CSS codes, our procedure yields explicit quantum codes surpassing the square-root distance barrier, with parameters $[[n, \tilde \Omega(n^{1/3}), \tilde \Omega(n^{2/3})]]$. 
These codes also admit a three-dimensional embedding that saturates the Bravyi-Poulin-Terhal (BPT) bound, recovering the layer-code result. Our weight reduction technique also improves fault-tolerant logical-operator measurements by reducing the required number of ancilla qubits.

\end{abstract}

\section{Introduction}

Low-weight quantum code constructions are important throughout quantum computing, not only for syndrome extraction on noisy hardware but also for complexity theory and coding theory. On physical hardware, syndrome extraction is implemented through local operations, so very high-weight checks are inherently difficult to measure reliably. From the perspective of complexity theory, low-weight stabilizer checks naturally correspond to commuting local Hamiltonians, linking code constructions to the study of local Hamiltonians and to questions such as the No Low-Energy Trivial States (NLTS) theorem~\cite{anshuNLTSHamiltoniansGood2023} and potentially the quantum PCP (qPCP) conjecture~\cite{aharonovQuantumPCPConjecture2013}.
In coding theory, low weight is also a central feature in the study of qLDPC codes and quantum locally testable codes~\cite{panteleevAsymptoticallyGoodQuantum2022,willsTradeoffConstructionsQuantum2024,dinurExpansionHighDimensionalCubical2024}.

One particularly important low-weight framework is that of quantum low-density parity-check (qLDPC) codes,  namely quantum stabilizer codes where every check acts on $O(1)$ qubits and every qubit participates in $O(1)$ checks. The problem of finding qLDPC codes with both large rate and large distance was an important open question,  recently resolved in \cite{panteleevAsymptoticallyGoodQuantum2022}. The main strategy behind these asymptotic families of qLDPC codes is a one-shot construction, which combines a bounded-degree complex and suitable local codes. All known constructions of quantum LDPC codes follow this approach~\cite{panteleevAsymptoticallyGoodQuantum2022,leverrierQuantumTannerCodes2022,dinurGoodQuantumLDPC2023}. A key advantage of this approach is that it yields explicit constructions of qLDPC codes.

However, there is another route to constructing qLDPC codes through an iterative strategy. The idea is to alternate between two complementary procedures: one that increases the distance at the cost of higher check weight, and another that reduces the check weight at the cost of lower (relative) distance. By iterating these two procedures, one may hope to obtain a qLDPC code with good parameters. This viewpoint has been discussed in~\cite{bravyiHomologicalProductCodes2014} and was recently applied to code construction in~\cite{golowichQuantumLDPCCodes2025}, suggesting that the tension between distance growth and weight reduction may be managed through repeated rebalancing rather than a single one-shot construction.

One important theoretical motivation for studying the iterative approach is that closely related paradigms have already played a central role in several landmark classical results, such as the construction of spectral expanders via the zig-zag product~\cite{reingold2002entropy} and Dinur's proof of the PCP theorem~\cite{dinurPCPTheoremGap2007}. In these settings, one alternates between complementary transformations, each improving one feature while sacrificing another, and the power of the method comes from showing that the cumulative effect of the iteration is strongly positive. 
From this perspective, the iterative approach to qLDPC codes is appealing not only as a route to better parameters, but also as a possible pathway to amplify other properties of the code, including its soundness. A better understanding of this interaction may lead to structural insights relevant to quantum local testability and to the qPCP conjecture.

These techniques may also be relevant for practical error correction. For example, some constructions, including asymptotically good qLDPC codes, can have check weights on the order of 100, and could benefit from weight reduction to bring these weights down to a more practical regime. More broadly, weight reduction has already found applications in the construction of quantum locally testable codes~\cite{willsTradeoffConstructionsQuantum2024} and in protocols for measuring logical operators~\cite{williamson2024low,ide2025fault,cowtanParallelLogicalMeasurements2025}. The latter application follows from the observation that a logical operator can be treated as a high-weight stabilizer, to which the same reduction techniques may be applied in order to realize fault-tolerant measurements with fewer ancilla qubits.

Taken together, these applications suggest that weight reduction plays the role of a fundamental primitive, rather than simply serving as a post-processing step. At a conceptual level, the central issue is whether one can enforce locality while preserving the global quality of the code. This naturally leads to the following structural question:

\begin{quote}
\emph{Can one transform a quantum code with strong global parameters into a low-weight code without sacrificing too much distance, rate, or geometry?}
\end{quote}

This question was first studied systematically by Hastings in 2021~\cite{hastingsQuantumWeightReduction2023}, who introduced a four-step framework for reducing the weight of a quantum code: copying, gauging, thickening and choosing heights, and coning. Although this framework successfully achieves weight reduction, its multi-step nature makes the overall construction difficult to analyze and tends to obscure the core mechanism. Moreover, Hastings' approach treats X and Z checks asymmetrically. These features suggest that the structural principles underlying weight reduction are still not fully understood, and that there remains room for a more unified and transparent theory.

In this work, we streamline the weight reduction procedure by identifying coning as the essential mechanism behind it, and showing that weight reduction can be carried out entirely through a symmetric coning framework. The construction is guided by a geometric interpretation of quantum codes, which clarifies the structure behind Hastings' multi-step approach and extends naturally to higher-dimensional settings. Beyond the present application, this geometric perspective provides a general framework for sparsifying complexes while preserving their homological properties, and may prove useful in other settings, including high-dimensional expansion and geometric embedding.

\subsection{Our results}

Our first main contribution is a coning-based framework for weight reduction that both isolates the core mechanism and achieves nearly optimal parameters.

\begin{theorem}[Main theorem]\label{thm:main}
Given a $[[n,k,d]]$ quantum code of weight at most $w$, there is a weight reduction procedure that generates a $[[O(nw^2\log w),k,\Omega(dw)]]$ quantum code
  with check weight $\le 5$ and qubit weight $\le 6$.
\end{theorem}

The significance of \Cref{thm:main} is both conceptual and quantitative. Unlike Hastings' construction, our method relies only on coning and treats X and Z checks symmetrically, which makes the weight reduction procedure simpler to describe and analyze. It also improves the resulting parameters to check weight $5$, qubit weight $6$, qubit blowup $O(nw^2\log w)$, and distance $\Omega(dw)$. Recent results~\cite{wangCheckweightconstrainedQuantumCodes2026,weiTheoryLowweightQuantum2026} suggest that check weight $5$ is already necessary for asymptotically good qLDPC codes. More broadly, our geometric viewpoint suggests that the other bounds are also close to the limits of this framework: sparsifying a 2D complex while preserving enough homological structure to route logical operators appears to force qubit blowup $O(nw^2\log w)$, and distance $\Omega(dw)$ up to constants; see~\Cref{sec:optimal} for further discussion.

\medskip

Our second result is a variant of the weight reduction procedure tailored to dense quantum codes, drawing inspiration from the layer code construction in~\cite{williamsonLayerCodes2023}. Compared to~\Cref{thm:main}, it provides a more efficient weight reduction procedure for dense codes, with a smaller qubit blowup of $O(n^3)$ instead of $O(n^3\log n)$:
\begin{theorem}\label{thm:dense}
  Given a $[[n,k,d]]$ quantum code with weight $w=\Omega(n)$,
  there exists a weight reduction procedure that generates
  a $[[O(n^3),k,\Omega(n)]]$ quantum code with check weight $\leq 6$
  and qubit weight $\leq 6$.
  Moreover, the new code has a geometrically local embedding in $\mathbb{R}^3$.

  If we start from a random dense CSS code with parameters $[[n,\Omega(n),\Omega(n)]]$, we will obtain a code with parameters $[[O(n^3), \Omega(n), \Omega(n^{2}/\log n)]]$. As the code has a 3D embedding, it saturates the Bravyi-Poulin-Terhal (BPT) bounds in~\cite{bravyiNogoTheoremTwodimensional2009,bravyiTradeoffsReliableQuantum2010} up to polylogarithmic factors.
\end{theorem}

The first part of the theorem follows from the results of the original layer code construction~\cite{williamsonLayerCodes2023}. We reconsider it from a geometric perspective and incorporate it into our unified framework.
The second part of the theorem states that applying~\Cref{thm:dense} to a random dense CSS code with linear dimension and distance yields a 3D geometrically local code with parameters $[[N, \Omega(N^{1/3}), \Omega(N^{2/3}/\log N)]]$. In particular, this construction breaks the $O(\sqrt{N})$ distance barrier for quantum codes~\cite{hastingsFiberBundleCodes2021} and saturates the Bravyi-Poulin-Terhal upper bounds up to polylogarithmic factors. Compared to previous constructions with these two features~\cite{williamsonLayerCodes2023,liAlmostOptimalGeometrically2026}, it does not require high-dimensional expanders or a pre-existing good qLDPC code, but only random codes together with simple geometric reduction.

Furthermore, to rigorously prove the $\Omega(n^2/\log n)$ distance of this new code, we introduce a new analysis technique in~\Cref{sec:dense-reduce-layer-code}. The standard ``cleaning lemma'' paradigm~\cite{yuan2025unified} used in Theorem 1.1 (and prior works) would only yield a distance lower bound of $\Omega(n)$ for this construction. We show that performing a contraction operation on the layer code can provide an alternative way to analyze the distance. More details are discussed in~\Cref{sec:intro-distance}.

Finally, we note that our construction is built on results concerning the weight reduction of simplicial complexes. For the applications to quantum codes in this paper, only the two-dimensional complex is needed, but we also develop a corresponding procedure for general $t$-dimensional complexes. We expect this more general construction to have additional applications, for example to the weight reduction of quantum locally testable codes, which may be viewed as five-term chain complexes~\cite{dinurExpansionHighDimensionalCubical2024}.

\subsection{History of weight reduction}
\label{sec:history}

The primary prior work on weight reduction is due to Hastings~\cite{hastingsQuantumWeightReduction2023}, although several of the core ideas can already be found in~\cite[Section VI]{bravyiHomologicalProductCodes2014}.\footnote{Although most of the underlying ideas were introduced in 2016~\cite{hastingsWeightReductionQuantum2017}, that version contained an error, so we formally cite the corrected 2021 version instead.}
The paper presents a sequence of procedures that transforms an $[[n,k,d]]$ quantum code of weight $w$ into a constant-weight quantum code, while incurring a $\poly(w)$ overhead in the number of physical qubits and a $\poly(w)$ loss in the distance $d$.
Following this line of work, a series of subsequent studies have investigated various aspects of Hastings' weight reduction procedure, including its impact on code soundness~\cite{willsTradeoffConstructionsQuantum2024}, its behavior for explicit qLDPC codes~\cite{saboWeightReducedStabilizerCodes2024}, and the fault-tolerance properties of the weight reduction process~\cite{tanEffectiveDistanceHigher2024}. 

The main inspiration for Hastings' weight reduction routine is the code-to-manifold construction developed in~\cite{freedmanBuildingManifoldsQuantum2021}. At a high level, this construction takes the chain complex associated with the code and reconstructs, in reverse, a cell complex specified by its boundary operators.
The resulting cell complex may be regarded as a cellulation of a manifold. For non-LDPC codes, some cells can have $\Omega(n)$ neighbors. One therefore seeks to refine the cellulation via local operations to obtain a homotopy equivalent cellulation in which every cell has only $O(1)$ neighbors. The refined cell complex then naturally gives rise to the desired bounded-weight quantum code.

However, several obstacles remain in turning the above idea into a rigorous proof. The code-to-manifold construction of~\cite{freedmanBuildingManifoldsQuantum2021} requires a lifting step that passes from the code to a chain complex over $\ZZ$, and it is unclear how this operation affects the weights of the boundary operators. Moreover, in order to keep the weight reduction procedure as efficient as possible, we would prefer to work directly at the level of codes, rather than introduce an additional translation through manifolds.

To overcome the barriers above, Hastings introduced in~\cite{hastingsQuantumWeightReduction2023} a four-step procedure that implements weight reduction directly at the level of codes. Writing $w_X$ and $w_Z$ for the weights of the X and Z stabilizers, and $q_X$ and $q_Z$ for the numbers of X and Z stabilizers acting on each qubit, the effect of each step can be summarized as follows:

 \begin{enumerate}
     \item Copying: reduces $q_X$ to a constant.
     \item Gauging: reduces $w_X$ to a constant while preserving $q_X$ as a constant.
     \item Thickening and choosing heights: reduces $q_Z$ to a constant while preserving $w_X,q_X$ as constants.
     \item Coning: reduces $w_Z$ to a constant while preserving $w_X,q_X,q_Z$ as constants.
 \end{enumerate}

Taken together, these four steps yield a complete weight reduction procedure. Detailed accounts of the individual steps can be found in~\cite{willsTradeoffConstructionsQuantum2024,saboWeightReducedStabilizerCodes2024}. Moreover, in line with the intuition described above, each operation establishes a homotopy equivalence between the two resulting chain complexes.

For applications of the weight reduction process,
Hastings has indeed applied this method to construct an LDPC code with parameters $[[N,\tilde{\Omega}(N^{1/3}),\tilde{\Omega}(N^{2/3})]]$,
by applying the weight reduction process only to a random CSS code with low-weight $X$ stabilizers and linear distance.
This construction breaks the famous square-root distance barrier for quantum LDPC codes using weight reduction techniques alone.
Combined with distance balancing techniques from~\cite{evraDecodableQuantumLDPC2024}, the result can be improved to $[[N,O(N^{\alpha}),O(N^{1-\alpha/2})]]$ for $\alpha\in [2/3,1]$. In~\cite{saboWeightReducedStabilizerCodes2024}, the authors studied the behaviour of Hastings' weight reduction on explicit small codes, and showed how to combine weight reduction for classical codes to optimize the overhead of Hastings' process on product codes.

Hastings' result provides a general method for weight reduction in quantum codes, but the procedure has several limitations. The individual steps are technically intricate, and the overall effect of the full procedure on code size and distance is difficult to analyze precisely. Moreover, according to~\cite{saboWeightReducedStabilizerCodes2024}, the constant overhead incurred by Hastings' scheme is substantial, which may limit its practicality in near-term implementations. In addition, the procedure is inherently asymmetric in its treatment of X and Z stabilizers, whereas many CSS code constructions are symmetric. These considerations suggest that there remains room for more efficient weight reduction methods.

\subsection{Technical Overview: Construction of Weight-Reduced Codes}

Although our construction is inspired by the intuition above, it proceeds somewhat differently. Among the four steps in Hastings' construction, we identify the coning step as the essential one. Our construction can be viewed as replacing each qubit and check with a homotopy equivalent cone of bounded degree, with carefully designed boundary conditions to connect the cones together\footnote{The cone of a topological space $X$ is defined as $\cone(X)=(X\times [0,1])/(X\times \{1\})$, and $X$ is called the base of the cone. The intuition is to create a contractible space with boundary $X$.}. The resulting weight-reduced code remains homotopy equivalent to the original code, since the construction consists essentially of homotopy equivalence operations on the underlying chain complex. However, because the new X and Z stabilizers must still commute, not every choice of cone is suitable for our purposes.

Our weight reduction procedure for general quantum codes consists of the following steps: We first construct a square complex from the quantum code, and then we subdivide the complex and attach dummy faces to make the link graph an expander. Next, we sparsify the boundary of each local complex around a vertex, and we construct a cone at each vertex over the corresponding sparsified boundary. Finally, we assign checks and qubits given the local cone structure, and connect the local cones together according to the original complex to obtain our weight-reduced code. We now give a more detailed overview of each step.

\emph{Code to Complexes}: We construct a square complex $X=(V,E,F)$ from a quantum code $C\colon C_2\to C_1\to C_0$ using techniques from~\cite{liTransformArbitraryGood2024}. The square complex is obtained from the Tanner graph of the quantum code, with vertices $V=V_2\sqcup V_1\sqcup V_0$ representing X checks, qubits, and Z checks, and edges $E=E_{21}\sqcup E_{10}$ created by the action of the checks. The squares $F$ are created by pairing the qubits shared by intersecting X and Z checks. Therefore, the vertices in the square complex have labels of qubits and checks inherited from the Tanner graph. The following simple observation turns out to be useful for our construction.
\begin{claim}\label{clm:star}
For each vertex $v\in V_1$ in the square complex $X$ (i.e.\ a qubit), every pair of neighboring vertices labeled by an X check and a Z check uniquely determines a square incident to $v$.
\end{claim}

\emph{Subdivision and Expanderization}: We subdivide every square into four subsquares,  each attached to one vertex, as shown in~\Cref{fig:prepare-intro}. This step isolates the local complex around each vertex, thereby enabling subsequent local operations. We additionally attach a collection of dummy faces to each local complex so that the resulting link graph becomes an expander.\footnote{The expanderization step is needed only for the distance analysis, and may be ignored on a first reading.} We use $X^x,X^q,X^z$ to denote the local complexes around the corresponding vertices, and $X^{xq},X^{qz}$ for the intersections of the complexes. Moreover, we define $\partial X^x=\bigcup_{q\sim x}X^{xq}$ to be the boundary of $X^x$, and similarly for $\partial X^z$ and $\partial X^q$.

\begin{figure}[H]
\centering
\begin{subfigure}[b]{0.4\linewidth}
        \centering
\includegraphics[width=0.9\linewidth]{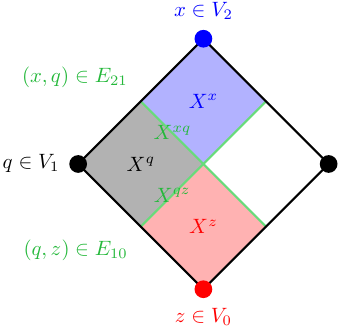}
        \caption{}
    \end{subfigure}
    \hfill
    \begin{subfigure}[b]{0.4\linewidth}
        \centering
\includegraphics[width=0.7\linewidth]{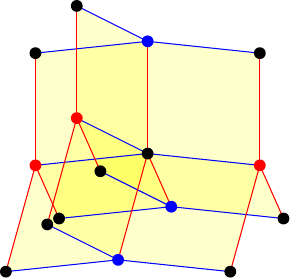}
        \caption{}
    \end{subfigure}
    \caption{We subdivide the square complex \(X\) into local cell complexes \(X^x\), \(X^q\), \(X^z\), \(X^{xq}\), and \(X^{qz}\). For simplicity, the left panel illustrates the construction for a single square; in general, the full construction involves many such squares, as shown in the right panel.}
    \label{fig:prepare-intro}
\end{figure}

\emph{Sparsification}: The goal of this step is to construct, for every $t\in V$, a constant-degree graph $\partial \cX^t$ that is homotopy equivalent to the boundary $\partial X^t$ and serves as the boundary of the weight-reduced complex $\cX^t$. Note that for fixed labels $q$ and $x$, the local complex $X^{xq}$ is a star graph, and the same holds for $X^{qz}$. We sparsify the star graph $S_n$ into a comb-like graph $\cS_n$, as shown in~\Cref{fig:intro-star-graph}. Applying this construction to every star graph appearing in $\partial X^t$, we obtain the new boundary $\partial \cX^t$.
    
\begin{figure}[H]
\centering
\begin{tikzpicture}
      \begin{scope}
        \filldraw (0,0) -- (2,2) circle (2pt);
        \filldraw (0,0) -- (1,2) circle (2pt);
        \filldraw (0,0) -- (0,2) circle (2pt);
        \filldraw (0,0) -- (-1,2) circle (2pt);
        \filldraw (0,0) -- (-2,2) circle (2pt);

        \fill (0,0) circle (2pt);
      \end{scope}

      \draw[->] (2.5,1) -- (5,1);

      \begin{scope}[shift={(8,0.5)}]
        \filldraw (2,0) circle (2pt) -- (2,1) circle (2pt);
        \filldraw (1,0) circle (2pt) -- (1,1) circle (2pt);
        \filldraw (0,0) circle (2pt) -- (0,1) circle (2pt);
        \filldraw (-1,0) circle (2pt) -- (-1,1) circle (2pt);
        \filldraw (-2,0) circle (2pt) -- (-2,1) circle (2pt);

        \draw (2,0) -- (-2,0);
      \end{scope}
    \end{tikzpicture}
    \caption{Sparsifying the star graph $S_n$ into $\cS_n$ with $2n$ vertices.}
    \label{fig:intro-star-graph}
 \end{figure}

\emph{Cone Construction}: This is the main technical contribution of our paper. To construct the weight-reduced complex \(\cX^t\), we build a constant-degree cone at each vertex over the corresponding sparsified base \(\partial \cX^t\). We first construct the qubit cone $\cX^q$. From~\Cref{clm:star}, the 1-skeleton of the local complex $X^q$ around the qubit $q$ is isomorphic to the product of two star graphs $S_{n_x} \times S_{n_z}$,  where $n_x$ and $n_z$ are the numbers of adjacent X-checks and Z-checks, respectively. 
By our sparsification procedure, a natural candidate for the skeleton of the qubit cone $\cX^q$ is the product of two comb-like graphs $\cS_{n_x} \times \cS_{n_z}$. The product of the comb backbones forms a grid structure, while the product of the comb teeth forms a structure that looks like a set of parallel lines.  The product of the tooth vertices and another comb forms a natural embedding of $\partial \cX^q$.
An example of the qubit cone is shown in~\Cref{fig:qubit-cone}.
    \begin{figure}[H]
      \centering
      \begin{tikzpicture}[
            vertex/.style={circle,fill=black,inner sep=1.2pt},
            edgeA/.style={thin},
            edgeB/.style={thin}
        ]

        \foreach \i in {1,2,3}{
          \foreach \a in {0,1}{
            \foreach \j in {1,2,3}{
              \foreach \b in {0,1}{
                \pgfmathsetmacro{\x}{1.6*(\i-1) + 1.2*(\j-1) + 0.8*\a + 0.45*\b}
                \pgfmathsetmacro{\y}{-0.7*(\i-1) + 0.5*(\j-1) + 1.0*\a - 0.95*\b}
                \coordinate (v\i\a\j\b) at (\x,\y);
                \node[vertex] at (v\i\a\j\b) {};
              }
            }
          }
        }
        \fill[violet,opacity=0.3] (v1010)--(v1030)--(v3030)--(v3010)--cycle;

        \foreach \i/\ip in {1/2,2/3}{
          \foreach \j in {1,2,3}{ 
              \draw[edgeA] (v\i0\j0) -- (v\ip0\j0);
              \draw[edgeA,thick,red] (v\i0\j1) -- (v\ip0\j1);
          }
        }

        \foreach \i in {1,2,3}{
          \foreach \j in {1,2,3}{
              \draw[edgeA] (v\i0\j0) -- (v\i1\j0);
              \draw[edgeA,thick,red] (v\i0\j1) -- (v\i1\j1);
          }
        }

        \foreach \j/\jp in {1/2,2/3}{
          \foreach \i in {1,2,3}{
              \draw[edgeB] (v\i0\j0) -- (v\i0\jp0);
              \draw[edgeB,thick,blue] (v\i1\j0) -- (v\i1\jp0);
          }
        }

        \foreach \j in {1,2,3}{
          \foreach \i in {1,2,3}{
              \draw[edgeB] (v\i0\j0) -- (v\i0\j1);
              \draw[edgeB,thick,blue] (v\i1\j0) -- (v\i1\j1);
          }
        }

        \end{tikzpicture}
        \caption{An example of the qubit cone $\cX^q$ with skeleton $\cS_3 \times \cS_3$. The blue and red edges correspond to the boundary $ \cX^{xq}$ and $ \cX^{qz}$ respectively, and the violet faces form the interior grid structure from the backbone product.}
        \label{fig:qubit-cone}
    \end{figure}
For the check cone $\cX^x/\cX^z$, its skeleton may have more general structure, so we choose to build the cone from the base $\partial \cX^x/\partial \cX^z$ directly. A bounded-weight cone construction for general base graphs was given in~\cite{berdnikov2022parsimonious}. We refine that construction and provide a more detailed analysis, obtaining improved degree bounds that meet our target weight constraints. This also reduces the overhead of the original coning construction in~\cite{hastingsQuantumWeightReduction2023} from $O(w^2\log^3 w)$ to $O(w^2\log w)$.

\emph{Chain Embedding}: In this step, we construct a chain complex from our weight reduced complexes, thus obtaining our weight reduced code. We utilize the quantum code embedding framework introduced in~\cite{yuan2025unified}. Each local cone can be described by a three-term chain complex $\cC^{t}:\cC^{t}_2\to\cC^{t}_1\to\cC^{t}_0, t\in V_2\sqcup V_1\sqcup V_0$, where the boundary maps are naturally induced by the cone structure as follows:
     \begin{itemize}
  \item For $\cC^x$, we place X-checks on vertices, qubits on edges, and Z-checks on faces of $\cX^x$.
  \item For $\cC^q$, we place qubits on vertices and faces,
    X-checks on ``vertical'' edges,
    and Z-checks on ``horizontal'' edges
     of the interior grid in $\cX^q$, as shown in~\Cref{fig:qubit-layer}.
  \item For $\cC^z$, we place Z-checks on vertices, qubits on edges, and X-checks on faces of $\cX^z$.
\end{itemize}
 Since every cone is homotopy equivalent to a point, it is easy to check the following properties of the homology groups:
\begin{itemize}
    \item $H_2(\cC^x) = \FF_2, H_1(\cC^x) = 0, H_0(\cC^x) = 0$.
    \item $H_2(\cC^q) = 0, H_1(\cC^q) = \FF_2, H_0(\cC^q) = 0$.
    \item $H_2(\cC^z) = 0, H_1(\cC^z) = 0, H_0(\cC^z) = \FF_2$.
  \end{itemize}
\begin{figure}[H]
  \centering
  \begin{tikzpicture}
\begin{scope}[shift={(0,2)}]
    \node at (0,0) [circle,fill,inner sep=1.5pt] (q){};
        \node at (0,-2) [circle, draw,red] (z2){};
        \node at (-1,-2) [circle, draw,red] (z1) {};
        \node at (1,-2) [circle, draw,red] (z3) {};
        \node at (-1,2) [circle, draw,blue] (x1) {};
        \node at (0,2) [circle, draw,blue] (x2) {};
        \node at (1,2) [circle, draw,blue] (x3) {};
        \draw (q)--(z1);
        \draw (q)--(z2);
        \draw (q)--(z3);
        \draw (q)--(x1);
        \draw (q)--(x2);
        \draw (q)--(x3);
        \node at (0,-2.5) {$Z_2$};
        \node at (-1,-2.5) {$Z_1$};
        \node at (1,-2.5) {$Z_3$};
        \node at (0,2.5) {$X_2$};
        \node at (-1,2.5) {$X_1$};
        \node at (1,2.5) {$X_3$};
        \node at (0.5,0) {$q$};
\end{scope}
        
\draw[ultra thick,->] (2,2) -- (5,2);
  \begin{scope}[shift={(6,0)}]
        \foreach \row in {0,1,2}{
          \foreach \column in {0,1,2}{
            \pgfmathtruncatemacro{\rw}{2*\row}
            \pgfmathtruncatemacro{\cl}{2*\column}
            \node at (2*\row,2*\column)[circle,fill,inner sep=1.5pt] (v\rw\cl){};
            \pgfmathtruncatemacro{\cl}{2*\column+1}
            \node at (2*\row,2*\column+1)[circle, draw,blue] (v\rw\cl) {};
            \pgfmathtruncatemacro{\rw}{2*\row+1}
            \node at (2*\row+1,2*\column+1)[circle,fill,inner sep=1.5pt] (v\rw\cl){};
            \pgfmathtruncatemacro{\cl}{2*\column}
            \node at (2*\row+1,2*\column) [circle, draw,red](v\rw\cl) {};
          };
        };

        \foreach \row in {0,1,...,4}{
          \foreach \column in {0,1,...,3}{
            \pgfmathtruncatemacro{\cl}{\column+1}
            \draw (v\row\column)--(v\row\cl);
          }
        }

        \foreach \column in {0,1,...,4}{
          \foreach \row in {0,1,...,3}{
            \pgfmathtruncatemacro{\rw}{\row+1}
            \draw (v\row\column)--(v\rw\column);
          }
        }

        \filldraw [white] (-0.5,4.5) rectangle (5.5,5.5);
        \filldraw [white] (4.5,-0.5) rectangle (5.5,5.5);
        \node at (5,0) {$X_1$};
        \node at (5,2) {$X_2$};
        \node at (5,4) {$X_3$};
        \node at (0,5) {$Z_1$};
        \node at (2,5) {$Z_2$};
        \node at (4,5) {$Z_3$};

        \filldraw [blue,nearly transparent] (-0.3,-0.3) rectangle (4.3,0.3);
        \filldraw [blue,nearly transparent] (-0.3,1.7) rectangle (4.3,2.3);
        \filldraw [blue,nearly transparent] (-0.3,3.7) rectangle (4.3,4.3);
        \filldraw [red, nearly transparent] (-0.3,-0.3) rectangle (0.3,4.3);
        \filldraw [red, nearly transparent] (1.7,-0.3) rectangle (2.3,4.3);
        \filldraw [red, nearly transparent] (3.7,-0.3) rectangle (4.3,4.3);
        \end{scope}
      \end{tikzpicture}\caption{The left figure shows a portion of the original chain complex $C$ around a qubit $q$. The right figure shows the corresponding interior region in the new chain complex $\cC^q$, which coincides with the structure of a surface code. For clarity, we omit the off-plane edges in the qubit cone, as shown in~\Cref{fig:qubit-cone}. The red regions are connected to $\cC^z$,  while the blue regions are connected to $\cC^x$.}
  \label{fig:qubit-layer}
\end{figure}
After constructing the local chain complexes, we define the connection boundary maps between different local complexes,
\begin{equation}
  g^{xq}\colon \cC^x \to \cC^q,\quad
  g^{qz}\colon \cC^q \to \cC^z,\quad
  p^{xz}\colon \cC^x \to \cC^z.
\end{equation}
In our construction, we directly set $p^{xz}=0$.
From the Sparsification step, the qubit cone and the check cone share the same comb-like boundary structure, we can naturally define the connection maps $g^{xq}$ and $g^{qz}$ by identifying the boundary of the qubit cone with the corresponding boundary of the check cone, as shown in~\Cref{fig:connection}.

\begin{figure}[H]
    \centering
        \begin{tikzpicture}
        \begin{scope}
           \node at (0,0) [circle,fill,inner sep=1.5pt] (q){};
            \node at (-3,0) [circle,fill,inner sep=1.5pt] (q1){};
            \node at (-2,0) [circle,fill,inner sep=1.5pt] (q2){};
            \node at (-1,0) [circle,fill,inner sep=1.5pt] (q3){};
            \node at (-1,-2) [circle, draw,blue] (z1) {};
            \node at (-2,2) [circle, draw,red] (x1) {};
            \node at (-1,2) [circle, draw,red] (x2) {};
            \node at (0,2) [circle, draw,red] (x3) {};
            \draw (q)--(x1);
            \draw (q)--(x2);
            \draw (q)--(x3);
            \draw (z1)--(q);
            \draw (z1)--(q1);
            \draw (z1)--(q2);
            \draw (z1)--(q3);
            \draw (q1)--(x1);
            \draw (q2)--(x2);
            \draw (q3)--(x3);
            \fill[brown,nearly transparent] (-1,-2)--(0,0)--(0,2)--(-1,0)--cycle;
            \fill[brown,nearly transparent] (-1,-2)--(0,0)--(-1,2)--(-2,0)--cycle;
            \fill[brown,nearly transparent] (-1,-2)--(0,0)--(-2,2)--(-3,0)--cycle;
            \node at (-1,-2.5) {$X_1$};
            \node at (-3.3,0) {$q_1$};
            \node at (-2.3,0) {$q_2$};
            \node at (-1.3,0) {$q_3$};
            \node at (0.3,0) {$q$};
            \node at (-2,2.5) {$Z_1$};
            \node at (-1,2.5) {$Z_2$};
            \node at (0,2.5) {$Z_3$};
        \end{scope}
 \draw[ultra thick,->] (1,0) --  (3,0);
        \begin{scope}[shift={(6,0)}]
          \node at (-2,-2) [circle, draw,blue] (z11){};
            \node at (0,-2) [circle, draw,blue] (z12){};
            \node at (2,-2) [circle, draw,blue] (z13) {};
            \node at (-2,0) [circle,fill,inner sep=1.5pt] (qz1){};
            \node at (0,0) [circle,fill,inner sep=1.5pt] (qz2){};
            \node at (2,0) [circle,fill,inner sep=1.5pt] (qz3){};
            \node at (-2.5,2) [circle, draw,red] (x1){};
            \node at (-0.5,2) [circle, draw,red] (x2){};
            \node at (1.5,2) [circle, draw,red] (x3){};
            \node at (-2.5,0.4)  [circle,fill,inner sep=1.5pt] (q1){};
            \node at (-0.5,0.4)  [circle,fill,inner sep=1.5pt] (q2){};
            \node at (1.5,0.4)  [circle,fill,inner sep=1.5pt] (q3){};
           \node at (-1,0) [circle, draw,red] (xz1){};
           \node at (1,0) [circle, draw,red] (xz2){};
           \node at (-1,-2) [circle,fill,inner sep=1.5pt] (qc1){};
           \node at (1,-2) [circle,fill,inner sep=1.5pt] (qc2){};
           \draw (z11)--(qz1);
           \draw (qz1)--(x1);
           \draw (x1)--(q1);
           \draw (q1)--(z11);
           \draw (z12)--(qz2);
           \draw (qz2)--(x2);
           \draw (x2)--(q2);
           \draw (q2)--(z12);
           \draw (z13)--(qz3);
           \draw (qz3)--(x3);
           \draw (x3)--(q3);
           \draw (q3)--(z13);
           \draw (qz1)--(xz1);
           \draw (xz1)--(qz2);
           \draw (z11)--(qc1);
           \draw (qc1)--(z12);
           \draw (xz1)--(qc1);
           \draw (qz2)--(xz2);
           \draw (xz2)--(qz3);
           \draw (z12)--(qc2);
           \draw (qc2)--(z13);
           \draw (xz2)--(qc2);
           \fill[brown, nearly transparent] (-2,-2)--(-2,0)--(-2.5,2)--(-2.5,0.4)--cycle;
           \fill[brown, nearly transparent] (0,-2)--(0,0)--(-0.5,2)--(-0.5,0.4)--cycle;
           \fill[brown, nearly transparent] (2,-2)--(2,0)--(1.5,2)--(1.5,0.4)--cycle;
           \fill[brown, nearly transparent] (-2,-2)--(-2,0)--(2,0)--(2,-2)--cycle;
           \filldraw[black,nearly transparent] (-2.2,-0.2) rectangle (2.2,0.2);
           \fill[black,nearly transparent] (-2.2,-0.2)--(-1.8,-0.2)--(-2.3,2.2)--(-2.7,2.2)--cycle;
           \fill[black,nearly transparent] (-0.2,-0.2)--(0.2,-0.2)--(-0.3,2.2)--(-0.7,2.2)--cycle;
           \fill[black,nearly transparent] (1.8,-0.2)--(2.2,-0.2)--(1.7,2.2)--(1.3,2.2)--cycle;
            \filldraw[blue,nearly transparent] (-2.2,-2.2) rectangle (2.2,-1.8);
           \fill[blue,nearly transparent] (-2.2,-2.2)--(-1.8,-2.2)--(-2.3,0.6)--(-2.7,0.6)--cycle;
           \fill[blue,nearly transparent] (-0.2,-2.2)--(0.2,-2.2)--(-0.3,0.6)--(-0.7,0.6)--cycle;
           \fill[blue,nearly transparent] (1.8,-2.2)--(2.2,-2.2)--(1.7,0.6)--(1.3,0.6)--cycle;
           \node at (-2.5,2.5) {$Z_1$};
           \node at (-0.5,2.5) {$Z_2$};
           \node at (1.5,2.5) {$Z_3$};
           \node at (-2.8,0.4) {$q_1$};
           \node at (-0.8,0.4) {$q_2$};
           \node at (1.2,0.4) {$q_3$};
           \end{scope}
        \end{tikzpicture}
    \caption{The left shows a portion of the original chain complex $C$
              around a $qx$ edge.
              The right shows the corresponding region in the new chain complex $\cC$,
                where two structures identical to $\cX^{qx}$ are matched.
              The blue region is a part of $\cC^x$ by identifying the blue boundary in~\Cref{fig:qubit-cone} with $\cX^x$,
                while the grey region is a part of $\cC^q$.
            }
    \label{fig:connection}
\end{figure}
These maps define the connections between checks and qubits in different cones.
The chain complex of the weight-reduced code is determined by the local chain complexes, as we take the direct sum of $\cC^t_i,i\in\{0,1,2\}$ as the corresponding space of checks/qubits, and the boundary maps are naturally induced by the local boundary maps, as shown in the following diagram.

\begin{equation} \label{eq:height-2-diagram-intro}
  \begin{tikzpicture}[baseline]
    \matrix(a)[matrix of math nodes, nodes in empty cells, nodes={minimum size=25pt},
    row sep=2em, column sep=2em,
    text height=1.25ex, text depth=0.25ex]
    {\blue{\cC^{X}_{2}}  & \blue{\cC^{X}_{1}} & \blue{\cC^{X}_{0}} \\
    {\cC^{Q}_{2}}  & {\cC^{Q}_{1}}  & {\cC^{Q}_{0}} \\
    \red{\cC^{Z}_{2}} & \red{\cC^{Z}_{1}} & \red{\cC^{Z}_{0}} \\};
    \path[->,blue,font=\scriptsize]
    (a-1-1) edge node[above]{$\partial^{X}_2$} (a-1-2)
    (a-1-2) edge node[above]{$\partial^{X}_1$} (a-1-3);
    \path[->,font=\scriptsize]
    (a-2-1) edge node[above]{$\partial^{Q}_2$} (a-2-2)
    (a-2-2) edge node[above]{$\partial^{Q}_1$} (a-2-3);
    \path[->,red,font=\scriptsize]
    (a-3-1) edge node[above]{$\partial^{Z}_2$} (a-3-2)
    (a-3-2) edge node[above]{$\partial^{Z}_1$} (a-3-3);
    \path[->,font=\scriptsize]
    (a-1-1) edge[bend right=65, dashed] node[left]{$p^{XZ}_2$} (a-3-2)
    (a-1-2) edge[bend left=65, dashed] node[right]{$p^{XZ}_1$} (a-3-3);
    \path[->,font=\scriptsize]
    (a-1-1) edge node[above right]{$g^{XQ}_2$} (a-2-2)
    (a-1-2) edge node[above right]{$g^{XQ}_1$} (a-2-3)
    (a-2-1) edge node[above right]{$g^{QZ}_2$} (a-3-2)
    (a-2-2) edge node[above right]{$g^{QZ}_1$} (a-3-3);
  \end{tikzpicture}
\end{equation}
Here the chain complexes $\cC^{X},\cC^Q,\cC^Z$ are defined as the direct sums of the local chain complexes, and the boundary operators $\partial^X,\partial^Q,\partial^Z$ are defined as the direct sums of the local boundary operators, and the connection maps $g^{XQ},g^{QZ},p^{XZ}$ are defined by the local connection maps.

Using techniques from the paper~\cite{yuan2025unified}, we can show that our weight-reduced chain complex $\cC$ defined by the above diagram is homotopy equivalent to the original chain complex $C$. Moreover, there is a natural correspondence between the homology groups of $\cC$ and the basis of the original code $C$, and the chain maps between homology groups $[g^{XQ}]$ and $[g^{QZ}]$ also correspond to the boundary maps $\partial_2$ and $\partial_1$ of the original code. 

Our second construction for dense quantum codes is mostly inspired by the layer code construction~\cite{williamsonLayerCodes2023}. The major advantage of the construction is that it has a lower qubit overhead compared to the first when $w=\Theta(n)$, as it requires $O(n^3)$ qubits instead of $O(n^3\log n)$ qubits.

\subsection{Analysis of distance}
\label{sec:intro-distance}

As shown in the construction section, our code consists of cones of the checks and qubits. For a codeword $c$ in the original code $C$, we can define a canonical codeword $\tilde{c}$ in the new code $\cC$, and the canonical codeword will have size $|\tilde{c}|\geq w|c|$.

The bound follows from the canonical codeword being mostly supported on the qubit cones. Since the qubit cone has the structure of a surface code, its canonical codeword will have distance $w$.
For other codewords $c'$ in the new code, we will follow a cleaning process to find an equivalent canonical codeword $\tilde{c}'$.
For example, if the codeword $c'$ is a Z codeword, we will first clean its support on the Z-check cones, by flipping the Z stabilizers inside the cone.
By our construction, the Z-check cones have trivial homology, thus we are able to find some boundary $\partial z$ to remove the support of $c'$ in the Z-check cones.
Furthermore, we will clean the codeword on the qubit and X-check cones to obtain the canonical codeword $\tilde{c}'$.
As we will see, the distance lower bound will be dominated by the cleaning process on the Z-check cones.
If the underlying graph for Z-checks has constant expansion, we can show that the distance will be lower bounded by $\Omega(dw)$.
The main reason we introduced expander gadgets in the check cone construction is to satisfy this constraint.

We also analyze a special case of the layer code. When applying the layer code construction to a random dense CSS code with parameters $[[n,\Theta(n),\Theta(n)]]$, we improve the distance analysis, giving an $\Omega(n^2/\log n)$ lower bound. The analysis takes a different route beyond the cleaning paradigm. By stepping outside the paradigm, we demonstrate how to perform geometric contractions on the 3D layer code. The resulting contracted code can alternatively be viewed as arising from performing thickening and choosing heights from~\cite{hastingsQuantumWeightReduction2023} on a sparsified original code with distance $\Omega(n/\log n)$. Thickening provides an additional $\Omega(n)$ scaling factor to the distance, thus yielding the $\Omega(n^2/\log n)$ distance lower bound.

Note that we learned of this result~\cite{yuan2025unified} during the preparation of our work. For simplicity and to maintain a unified exposition in the main text, we therefore refer directly to the corresponding lemma in~\cite{yuan2025unified}. Our analysis of the distance is thus based on the standard cleaning approach, as nicely summarized there. Nevertheless, it is still helpful to give some intuition for the proof of the distance bound.

\subsection{Future directions}

\paragraph{Weight reduction of codes with transversal gates.}
Although our method reduces the weight of arbitrary quantum codes, it may fail to preserve additional structural properties of the original code that are important for fault-tolerant computation. A particularly notable example is the existence of transversal logical gates. Since transversal gates provide a highly robust mechanism for implementing logical operations, it is natural to ask whether one can develop weight reduction procedures that preserve them. More generally, one may seek weight reduction transformations that maintain the algebraic or geometric structure responsible for transversal gate implementations, rather than merely preserving the homological parameters of the code. Understanding this compatibility would be of both conceptual and practical interest, and could lead to weight reduction methods better suited for realistic fault-tolerant architectures.

\paragraph{Soundness analysis of our procedure.}
Previous works have analyzed the soundness change of the weight reduction process in Hastings' construction~\cite{willsTradeoffConstructionsQuantum2024}, which is an important feature of quantum locally testable codes.
An open question is how the soundness of our procedure compares to that of Hastings' approach. We conjecture that to obtain a weight reduction scheme with a good soundness parameter, one should design a weight reduction scheme for five-term chain complexes, as discussed in~\cite{dinurExpansionHighDimensionalCubical2024}.

\paragraph{Weight reduction with better parameters.}
Our weight reduction procedure constructs homotopy equivalent chains through a series of local geometric operations.
As discussed in \Cref{sec:optimal}, several pieces of evidence suggest that
  our weight reduction procedure is optimal within the current framework of geometric reduction. 
It remains an important direction to develop new techniques
  that achieve better parameters beyond this framework.

\paragraph{Possible relation to the quantum PCP conjecture.}
Mimicking the classical approach to the PCP theorem, one route toward understanding the quantum PCP conjecture is to study proof systems, while another is to follow Dinur's gap-amplification approach~\cite{dinurPCPTheoremGap2007}. There have been some results on the gap amplification of quantum Hamiltonians~\cite{aharonovDetectabilityLemmaQuantum2009,bergamaschiDerandomisedTensorProduct2025}, but after gap amplification, the resulting Hamiltonian has large locality, so a gap-preserving locality-reduction procedure is required. Our weight reduction procedure cannot be directly applied to reduce the locality of the resulting Hamiltonian, but we hope that a better understanding of the weight reduction process may lead to new insights toward locality reduction for proving the conjecture, since it can be viewed as a locality reduction for certain commuting Hamiltonians.
\subsection{Related works}
The concurrent work by Gu et al.~\cite{guLayerCodesPartially2025} also studies the application of the layer code construction to random dense CSS codes, and they obtain an $\Omega(n^2/\log n)$ bound on the distance. Our second weight reduction procedure, applied to a random dense CSS code, obtains the same distance bound, as shown in~\Cref{sec:dense-reduce-layer-code}. Our analysis is different from theirs, as they mainly study the propagation of excitation in the layer code, while we use the geometric observation that after compression, the layer code can be obtained by applying the thickening and choosing heights procedure in~\cite{hastingsQuantumWeightReduction2023}. 
Apart from that result, our work and theirs have different focuses.
 Their main focus is on the partial self-correction properties of the layer code, while our main focus is on the weight reduction process and its application in constructing good qLDPC codes. Our first result on weight reduction for general codes is unique.

We are also aware of the recent follow-up work~\cite{yuanParsimoniousQuantumLowDensity2026}, which provides a detailed exposition of our coning process in~\Cref{sec:sparse-simplex-reduce}. The paper reintroduces our result in~\Cref{sec:sparse-simplex-reduce} and clarifies several technical obscurities in the first version of our paper. The coning construction is applied there, but the result is not fully optimized in terms of the weight constants, as they may have a check acting on 9 qubits. The paper also summarizes the construction of our sparse cone and optimizes several technical aspects for fault-tolerant logical-operator measurements, as briefly mentioned in~\Cref{sec:logical-measurement}.

On the optimality of our construction, recent work by Wang et al.~\cite{wangCheckweightconstrainedQuantumCodes2026} and Wei et al.~\cite{weiTheoryLowweightQuantum2026} has shown that for any quantum LDPC code with distance $d\geq 3$, the check weight must be at least 4. Moreover, Wang et al. showed that for quantum codes with check weight 4, the codes are essentially generalized surface codes and satisfy the following dimension-distance tradeoff: $kd^2=O(n)$. This implies that if we want to obtain asymptotically good quantum LDPC codes, we need check weight at least 5, which is consistent with our result. Their results provide strong evidence for the weight optimality of our construction, in addition to our discussion in~\Cref{sec:optimal-weight}.

\section{Preliminaries}
In this section, we will introduce the necessary preliminaries.

\subsection{Notation and useful lemmas}

We use standard asymptotic notation $O, \Omega, \Theta$.
We use subscripts to indicate which parameters the hidden constants may depend on.
  Specifically, $f(n) = O_t(g(n))$
  if $f(n) \le c(t) g(n)$
  for a constant $c(t)$ that depends only on $t$.
We use ``soft-O'' notation to hide polylogarithmic factors.
  Specifically, $f(n) = \tilde O(g(n))$
  if $f(n) = O(g(n) \log^k n)$ for some constant $k$.

We use $[n]$ to denote the set $\{1, 2, 3, ..., n\}$. We use $\Delta_n$ to denote the standard $n$-simplex,
and superscript $X^t$ to indicate the dimension of the simplicial complex.

The cone of a topological space $X$ is defined as $\cone(X)=(X\times [0,1])/(X\times \{1\})$, and $X$ is called the base of the cone. 
We will also introduce some standard terminology from topology, and readers are encouraged to refer to~\cite{hatcher2002algebraic} if they are not familiar with the concepts.

\subsection{Chain complex and CSS codes}
\label{sec:css-codes}
Chain complexes offer an intuitive structure for studying quantum CSS codes. Within this framework, we can express the properties of the CSS code using the language of chain complexes, covering aspects such as dimension, distance, and energy barrier of a given code. We mainly consider chain complexes over the finite field $\F2$.

\begin{definition}[Chain complex]
  A chain complex $C$ consists of a sequence of vector spaces $C_i$ along with linear maps $\partial_i\colon C_i\to C_{i-1}$,
  known as boundary operators,
  where the boundary operators satisfy
  \begin{align*}
    \partial_{i-1}\partial_i=0.
  \end{align*}
\end{definition}

In our setting, each vector space is equipped with a chosen basis,
  which allows us to discuss Hamming weights.
With this basis,
  we can define the dual chain complex,
  whose coboundary operators
  $\delta_i\colon C_i\to C_{i+1}$,
  are given by the transposes of the boundary operators, $\delta_i=\partial_{i+1}^T$.
The coboundary operators satisfy
\begin{align*}
  \delta_{i+1}\delta_i=0.
\end{align*}

We introduce some standard definitions. Elements in the kernel of the (co)boundary operators are called (co)cycles:
\begin{align*}
    Z_i \coloneq \ker \partial_i=\{c_i\in C_i:\partial_ic_i=0\},\quad
    Z^i\coloneq\ker \delta_i=\{c_i\in C_i:\delta_ic_i=0\}.
\end{align*}
Elements in the image of the (co)boundary operators are called (co)boundaries:
\begin{align*}
    B_i\coloneq \im \partial_{i+1}=\{\partial_{i+1}c_{i+1}:c_{i+1}\in C_{i+1}\},\quad   B^i\coloneq \im \delta_{i-1}=\{\delta_{i-1}c_{i-1}:c_{i-1}\in C_{i-1}\}.
\end{align*}

A chain is called exact if $Z_i=B_i$ for all $i$. We can also define an exact cochain similarly.

A quantum CSS code $Q$ is defined by two classical codes $C_z, C_x$ represented by their parity check matrices $H_x\colon \F2^n\to \F2^{m_x}$ and $H_z\colon \F2^n\to\F2^{m_z}$ that satisfy $H_z H_x^T =0$. Here $n,m_x,m_z$ correspond to the number of qubits, X checks, and Z checks respectively. It is well known that the CSS code naturally corresponds to a chain complex as follows:
\[
\begin{tikzcd}
    \F2^{m_x}\arrow[r,"\partial_2=H_x^T"]&\F2^n\arrow[r,"\partial_1=H_z"]&\F2^{m_z}.
\end{tikzcd}
\]

The $X$ and Z logical operators correspond to the code $C_x$ and $C_z$, and X and Z stabilizers correspond to the code $C_z^{\perp}$ and $C_x^{\perp}$. The code dimension is defined by $k=\dim C_x-\dim C_z^{\perp}=\dim C_z-\dim C_x^{\perp}$. The code distance $d=\min(d_x,d_z)$ where
\begin{align*}
    d_x=\min_{c_x\in C_x-C_z^\perp}|c_x|,\quad d_z=\min_{c_z\in C_z-C_x^\perp}|c_z|.
\end{align*}

We say a quantum code is a low-density parity-check (LDPC) code if each check acts on a constant number of qubits, and each qubit is acted on by a constant number of checks. We call a quantum LDPC code good if it has asymptotically linear dimension and distance.

We denote the given code by the 3-term chain complex
\begin{equation}
  C: C_2 \xrightarrow{\partial_2} C_1 \xrightarrow{\partial_1} C_0.
\end{equation}
The corresponding weight reduced code is denoted using calligraphic letters
\begin{equation}
  \cC: \cC_2 \xrightarrow{\partial_2} \cC_1 \xrightarrow{\partial_1} \cC_0.
\end{equation}

\subsection{Square complex from quantum code}
\label{sec:square-complex}

One of the key insights from \cite{liTransformArbitraryGood2024},
  is that a CSS code naturally carries the structure of a 2D square complex.
Intuitively,
  the squares encode the fact that the X and Z checks commute.
It was further observed that several geometric operations on the 2D structure, such as subdivision,
  have direct counterparts for quantum codes.
As we will explain later,
  we build on this observation to perform weight reduction.

We review the construction of the 2D square complex $(V, E, F)$.
The vertices $V = V_2 \cup V_1 \cup V_0$
  correspond to the X checks, qubits, and Z checks.
The edges $E = E_{21} \cup E_{10}$
  correspond to the non-zero entries in the parity check matrices $H_x$ and $H_z$.
In particular, the graph $(V, E)$ is precisely the Tanner graph of the CSS code.

The main observation in \cite{liTransformArbitraryGood2024}
  is the additional face structure $F$.
Because the X and Z checks commute,
  each such pair shares an even number of qubits.
This allows the shared qubits to be paired up
  and form squares.
In particular,
  for a given $x \in V_2$ and $z \in V_0$,
  if two shared qubits $q, q' \in V_1$ form a pair,
  we introduce a square with vertices $(x, q, q', z)$.
In general, this pairing is not unique,
  and different pairing choices lead to different square complexes.
Nevertheless,
  the existence of any square complex
  guarantees that the corresponding code
  satisfies the commutation relation. 
  The following claim turns out to be useful in our construction.
  \begin{claim}[Restatement of~\Cref{clm:star}]
    \label{clm:unique}
    Given a vertex $q \in V_1$, for every pair of neighboring vertices $x \in V_2$ and $z \in V_0$,
      there is a unique $q'$ such that $(x, q, q', z)$ forms a square in the complex. This implies that the local square complex of $q$ is a product of two star graphs. 
  \end{claim}

\subsection{Framework of quantum code embedding}
\label{sec:quantum-code-embedding}

To relate the weight-reduced quantum codes back to the original codes,
  we utilize the framework of quantum code embedding introduced in~\cite{yuan2025unified},
  which can also be understood in terms of iterated mapping cones.
We have slightly modified the notation in~\cite{yuan2025unified}
  to align with the other notation used in this work.
We have also specialized the framework to our setting.
The framework provides a systematic way of describing the coning structure, and relates the chain complex of the weight reduced code to the original code.
This simplifies the analysis of dimension and distance of the weight reduced code.
We will give a brief introduction to the framework, and also include the necessary lemmas that we will use in our construction and analysis.

Given the chain complex of the weight reduced code
  $\cC: \cC_2 \xrightarrow{\partial_2} \cC_1 \xrightarrow{\partial_1} \cC_0$,
  constructed from a direct sum of the chain complexes of local cones
    $\cC^x: \cC^x_2 \to \cC^x_1 \to \cC^x_0$,
    $\cC^q: \cC^q_2 \to \cC^q_1 \to \cC^q_0$,
    $\cC^z: \cC^z_2 \to \cC^z_1 \to \cC^z_0$,
    for $x \in V_2$, $q \in V_1$, $z \in V_0$,
  we have
  \begin{equation}
    \cC: \bigoplus_{x \in V_2} \cC^x_2 \oplus \bigoplus_{q \in V_1} \cC^q_2 \oplus \bigoplus_{z \in V_0} \cC^z_2
      \xrightarrow{\partial_2} \bigoplus_{x \in V_2} \cC^x_1 \oplus \bigoplus_{q \in V_1} \cC^q_1 \oplus \bigoplus_{z \in V_0} \cC^z_1
      \xrightarrow{\partial_1} \bigoplus_{x \in V_2} \cC^x_0 \oplus \bigoplus_{q \in V_1} \cC^q_0 \oplus \bigoplus_{z \in V_0} \cC^z_0.
  \end{equation}
We denote by
\begin{equation}
  \cC^X = \bigoplus_{x \in V_2} \cC^x,\quad
  \cC^Q = \bigoplus_{q \in V_1} \cC^q,\quad
  \cC^Z = \bigoplus_{z \in V_0} \cC^z.
\end{equation}
  $\cC^X$ the direct sum of $\cC^x$,
  $\cC^Q$ the direct sum of $\cC^q$,
  $\cC^Z$ the direct sum of $\cC^z$.

The boundary maps of $\cC$ are constructed from
  the boundary maps of $\cC^x, \cC^q, \cC^z$ for the individual local complexes,
  together with the connecting maps between the cones $g^{xq}$, $g^{qz}$, and $p^{xz}$,
\begin{equation}
  g^{xq}\colon \cC^x \to \cC^q,\quad
  g^{qz}\colon \cC^q \to \cC^z,\quad
  p^{xz}\colon \cC^x \to \cC^z.
\end{equation}
We assume there are no direct maps from $\cC^z$ to $\cC^q$ or $\cC^x$,
  nor from $\cC^q$ to $\cC^x$.
This hierarchical structure commonly appears
  in lattice surgery constructions
  and in the geometric methods used in this work.
In fact, our construction is even more restrictive:
  there is no direct map from $\cC^x$ to $\cC^z$,
  i.e. $p^{xz} = 0$.

\begin{equation} \label{eq:height-2-diagram}
  \begin{tikzpicture}[baseline]
    \matrix(a)[matrix of math nodes, nodes in empty cells, nodes={minimum size=25pt},
    row sep=2em, column sep=2em,
    text height=1.25ex, text depth=0.25ex]
    {\blue{\cC^{X}_{2}}  & \blue{\cC^{X}_{1}} & \blue{\cC^{X}_{0}} \\
    {\cC^{Q}_{2}}  & {\cC^{Q}_{1}}  & {\cC^{Q}_{0}} \\
    \red{\cC^{Z}_{2}} & \red{\cC^{Z}_{1}} & \red{\cC^{Z}_{0}} \\};
    \path[->,blue,font=\scriptsize]
    (a-1-1) edge node[above]{$\partial^{X}_2$} (a-1-2)
    (a-1-2) edge node[above]{$\partial^{X}_1$} (a-1-3);
    \path[->,font=\scriptsize]
    (a-2-1) edge node[above]{$\partial^{Q}_2$} (a-2-2)
    (a-2-2) edge node[above]{$\partial^{Q}_1$} (a-2-3);
    \path[->,red,font=\scriptsize]
    (a-3-1) edge node[above]{$\partial^{Z}_2$} (a-3-2)
    (a-3-2) edge node[above]{$\partial^{Z}_1$} (a-3-3);
    \path[->,font=\scriptsize]
    (a-1-1) edge[bend right=65, dashed] node[left]{$p^{XZ}_2$} (a-3-2)
    (a-1-2) edge[bend left=65, dashed] node[right]{$p^{XZ}_1$} (a-3-3);
    \path[->,font=\scriptsize]
    (a-1-1) edge node[above right]{$g^{XQ}_2$} (a-2-2)
    (a-1-2) edge node[above right]{$g^{XQ}_1$} (a-2-3)
    (a-2-1) edge node[above right]{$g^{QZ}_2$} (a-3-2)
    (a-2-2) edge node[above right]{$g^{QZ}_1$} (a-3-3);
  \end{tikzpicture}
\end{equation}

Because $\partial^2 = 0$,
  the maps $g^{xq}, g^{qz}, p^{xz}$ must satisfy the compatibility conditions,
  such as $\partial^q_1 \circ g^{xq}_2 = g^{xq}_1 \circ \partial^x_2$
  for any $x \in V_2, q \in V_1$.
These conditions imply that
  $g^{xq}$, $g^{qz}$ are chain maps,
  which induce maps on homology
\begin{equation}
  [g^{xq}]: H_i(\cC^x) \to H_{i-1}(\cC^q),\quad
  [g^{qz}]: H_i(\cC^q) \to H_{i-1}(\cC^z).
\end{equation}

Our construction further has the property that each local complex
  $\cC^x, \cC^q, \cC^z$
  has nontrivial homology in exactly one dimension; that is to say, each local complex is homologically a cone of the corresponding vertex.
\begin{itemize}
  \item Each $\cC^x$ has homology $H_2(\cC^x) = \FF_2, H_1(\cC^x) = 0, H_0(\cC^x) = 0$.
  \item Each $\cC^q$ has homology $H_2(\cC^q) = 0, H_1(\cC^q) = \FF_2, H_0(\cC^q) = 0$.
  \item Each $\cC^z$ has homology $H_2(\cC^z) = 0, H_1(\cC^z) = 0, H_0(\cC^z) = \FF_2$.
\end{itemize}

These properties allow us to define the induced maps on homology:
\begin{equation}
  C^h: \bigoplus_{x \in V_2} H_2(\cC^x) \cong \FF_2^{V_2}
    \xrightarrow{[g^{XQ}]} \bigoplus_{q \in V_1} H_1(\cC^q) \cong \FF_2^{V_1}
    \xrightarrow{[g^{QZ}]} \bigoplus_{z \in V_0} H_0(\cC^z) \cong \FF_2^{V_0}.
\end{equation}

One can verify that $\partial^2 = 0$ implies that $C^h$ forms a chain complex. Intuitively, $C^h$ can be obtained from contracting each local cone $\cC^x, \cC^q, \cC^z$ to a single vertex $x, q, z$ respectively. Since the contraction map is a homotopy equivalence, one can show the following lemma:

\begin{lemma}
  \label{lem:hom-equiv}
    The chain complex $C^h$ is homotopy equivalent to the original chain complex $\cC$.
\end{lemma}
In other words, $C^h$ serves as an effective chain complex representing $\cC$. To show that the chain complex $C^h$ is indeed the chain complex of the original code $C$, we need to carefully design the maps $g^{xq}$ and $g^{qz}$, which is inherent in our construction.

This is precisely the relation between the original quantum code $C$
  and the weight-reduced quantum code $\cC$.
All these homological features are telling us that
  the two codes are closely related.
It is also discussed in~\cite{yuan2025unified} how to relate the distance of the original code to the distance of the weight-reduced code,
  which we will further explore in~\Cref{sec:distance}.

\section{Weight reduction of a simplicial complex}
\label{sec:weight reduction-simplicial-complex}

In this section, we review some known results on reducing the weight of a simplicial complex,
  as studied in depth by~\cite{berdnikov2022parsimonious}.
We also provide various improvements on the theorem. In~\Cref{sec:sparse-simplex-reduce}, we will review the results in~\cite{berdnikov2022parsimonious} and improve the constants of the cone construction. This will be applied for weight reduction of general quantum codes in~\Cref{sec:general-code-reduce}.
In~\Cref{sec:dense-complex}, we provide an alternative construction with better asymptotic bounds when the simplicial complex is dense. The construction will be used for weight reduction of dense quantum codes in~\Cref{sec:dense-code-reduce}. The two sections are independent, and readers can skip one of them if they are only interested in one of the two settings.

Although we only need to consider 2-dimensional complexes, we will also discuss the weight reduction of higher-dimensional simplicial complexes, which may be of independent interest.

\subsection{Weight reduction of a sparse simplicial complex}\label{sec:sparse-simplex-reduce}

We are interested in reducing the weight of general quantum codes, especially when the weight is significantly smaller than the code size.
In such cases, we rely on the parsimonious cone result from~\cite{berdnikov2022parsimonious},
  from which we extract only the part relevant to our setting.
\begin{lemma}[{\cite[A variant of Theorem 0.1]{berdnikov2022parsimonious}}]
  \label{lemma:cone-simplicial-complex}
  Let $X^1$ be a graph with $N$ edges with degree $\le 3$.
  Then, there exists a 2-simplicial complex $Y^2 \supseteq X^1$
    with bounded degrees so that
    $(Y, X)$ is homotopy equivalent to $(\cone(X), X)$
    and $|Y| = O(N \log N)$.
\end{lemma}

Recall that $\cone(X)$ denotes the cone over a space $X$,
  an operation in algebraic topology
  that stretches $X$ into a cylinder
  and then collapses one of its end faces to a point.
Saying that the pair $(Y, X)$ is homotopy equivalent to $(Z, X)$
  means that the spaces $Y$ and Z are homotopy equivalent in a way that keeps the subspace $X$ fixed.

To minimize the weight as much as possible,
  it is necessary to work with general cell complexes
  rather than restricting to simplicial complexes.
In a simplicial complex,
  the maximum vertex degree is typically large,
  leading to higher check weights.
By contrast,
  allowing non-triangular faces in a more general cell structure
  can significantly reduce the maximum vertex degree and thereby lower the overall weight.

In particular, we introduce the following variant.
\begin{lemma}
  \label{lemma:cone-cell-complex}
  Let $X$ be a graph with $N$ edges with degree $\le 3$.
  Then, there exists a pair of cell complexes $(Y, X)$
    such that
    $(Y, X)$ is homotopy equivalent to $(\cone(X), X)$
    and $|Y| = O(N \log N)$.
  Additionally,
  \begin{itemize}
    \item Each vertex is incident to $\le 5$ edges.
    \item Each edge is incident to $\le 4$ faces.
    \item Each face has $\le 5$ edges.
  \end{itemize}
  Furthermore, for cells contained in $X$,
  \begin{itemize}
    \item Each vertex is incident to $\le 4$ edges ($\le 3$ edges within $X$ and $\le 1$ edge in $Y$).
    \item Each edge is incident to $\le 1$ face.
  \end{itemize}
\end{lemma}

We first briefly outline the proof strategy of~\Cref{lemma:cone-simplicial-complex}, which will be the skeleton of our construction.

We will first construct a cone $(Y',X')$ where $X'$ is homotopy equivalent to $X$. To be specific, $X'$ can be contracted to $X$ by merging superedges formed by degree-2 vertex chains. We will discuss how to obtain a bounded-degree cone with boundary $X$ in our formal proof.

For any graph $X$, we can create a homotopy-equivalent bipartite graph $X'$ by adding a vertex in the middle of each edge, and $X'$ also has degree $\leq 3$.
  Thus we can assume that the graph $X$ is bipartite. 

To begin our construction, let $X=(V_0\sqcup V_1,E)$ be our bipartite graph; we will first build a cone for $V_0$ and $V_1$ respectively. For the 0-simplex $V_0$, its cone can be constructed by building a binary tree with $V_0$ as leaves. Similarly, we can build a cone for $V_1$.  We call the resulting cones the vertex cone trees $T_0^v, T_1^v$.

Now we turn to building a cone for $X$. Note that if we cut each edge at the midpoint, then for every $v\in V_0\sqcup V_1$ the resulting local piece is a star graph with center $v$ and leaf set labeled by the neighbors of $v$, denoted as $N(v)$.
We can replace the star graph by building a binary tree of neighboring edges at each vertex.
Combining the vertex cone tree of $V_0/V_1$ with the edge trees at each vertex, we obtain a tree whose leaves represent edges of $X^1$. We will use $T_0$ and $T_1$ to denote the edge tree for $V_0$ and $V_1$ respectively. Note that $T_0^v$ and $T_1^v$ are subtrees of $T_0$ and $T_1$.

  To construct a cone on $X$, we view the passage from $T_0$ to $T_1$ as a time-evolution built from local modifications of the tree. 
  Informally, we place each intermediate tree in a horizontal plane, stack these planes vertically according to the time order in which the local operations are performed, and use the vertical direction to record the time parameter $t$. 
  We denote the tree at time $t$ as $T(t)$. The sequence takes at most $\tau=\lceil\log_2|V_0|\rceil+\lceil\log_2 |V_1|\rceil$ time steps, and we index its layers by $t\in\{0,\ldots,\tau\}$, with $T(0)=T_0$ and $T(\tau)=T_1$.

Each intermediate tree has leaves labeled by the edges of $X$. Thus, across all time slices, leaves with the same label represent the same edge of the original graph, and in this way we recover a homotopy equivalent $X'$ inside the layered construction. For the rest of the cone, we connect consecutive layers in the most direct way: vertices and edges that are unchanged by the local operation are joined straight across the two time slices. In the small region where a local operation occurs, we insert a local gadget that connects the two adjacent trees. After assembling all layers and gadgets, we attach 2-cells along the resulting loops so that the whole complex becomes contractible.

\begin{proof}[Proof of~\Cref{lemma:cone-cell-complex}]

  With the above intuition in mind, we proceed to describe our construction. We begin by defining the binary representation of a tree.
  \begin{definition}
    The binary representation $B_T=\{(v,s(v))\}$ of a tree $T$   consists of pairs $(v,s(v))$, where $v$ is a leaf of $T$ and $s(v)$ is the binary string indicating the path to $v$ from the root of the tree, under the convention that 0 is the left child and 1 is the right child. We can also use a binary string to represent a node inside the tree following the same convention. 
  \end{definition}

 We can label each vertex $v_0\in V_0$ with a binary string $s(v_0)\in \{0,1\}^m$, which indicates the construction of the vertex cone tree $T^v_0$.
 Similarly, we can also construct the vertex cone tree $T^v_1$ using labels $s(v_1)\in \{0,1\}^n$, where $m=\lceil\log |V_0|\rceil$ and $n=\lceil\log |V_1|\rceil$.
 A natural way to construct the edge tree $T_0$ and $T_1$ is through concatenation of labels.  For an edge $e=(v_0,v_1)\in E$, the corresponding leaf in $T_0$ can be labeled by the concatenated string $s(v_0)s(v_1)$, while the corresponding leaf in $T_1$ is labeled by $s(v_1)s(v_0)$. The labels directly provide a construction for the trees $T_0, T_1$.

  The local move in~\Cref{fig:tree-swap} exchanges two adjacent coordinates in the label. Indeed, if the four leaves are labeled by $00,01,10,11$, then after the swap they appear in the order $00,10,01,11$, which is exactly the effect of exchanging the two coordinates. Therefore, transforming $T_0$ into $T_1$ reduces to moving the string $s(v_1)$ to the left of the string $s(v_0)$ by oblivious adjacent coordinate swaps. We formally summarize the above discussion in the following definition.

  \begin{definition}
    \label{def:leaf-swap}
    Given a tree $T$ with its binary representation $B_T$, a local leaf swap operation at depth $k$ with root $s$ is defined as follows: consider the subtree rooted at $s$, with $|s|=k$, the swap operation is performed as shown in~\Cref{fig:tree-swap} on the subtree, swapping the nodes and their corresponding subtrees. On the binary representation, for the leaves $v$ with $s$ as the prefix, say $s(v)=sb_{k+1}b_{k+2}t$, we obtain a new binary representation $s'(v)=sb_{k+2}b_{k+1}t$, thus obtaining a new tree $T'$ with binary representation $B_{T'}$. 
    
    A global leaf swap operation at depth $k$ is defined as performing the local leaf swap operation on all subtrees rooted at depth $k$ in parallel, which effectively exchanges the $(k+1)$-th and $(k+2)$-th coordinates in the binary representation of all leaves. Missing branches are treated as empty subtrees. Note that global leaf swaps at depths $k$ and $k-2$ can be performed in parallel.
  \end{definition}

  We define the intermediate trees $T(t)$ for $t=0,\ldots,\tau$ through their labels.

  \begin{definition}[Labels of the intermediate trees]
    \label{def:intermediate-tree-labels}
    Set $\tau=m+n$. For each edge $e=(v_0,v_1)\in E$, define $\ell_0(e)=s(v_0)s(v_1)$. Inductively, for $1\le t\le\tau$, obtain $\ell_t(e)$ from $\ell_{t-1}(e)$ by simultaneously exchanging coordinates $i+j-1$ and $i+j$ for all pairs
    \[
      (i,j)\in\{1,\ldots,m\}\times\{1,\ldots,n\}
      \quad\text{such that}\quad m-i+j=t,
    \]
    leaving all other coordinates unchanged. Define $T(t)$ as the tree of all prefixes of the labels $\ell_t(e)$, with leaf $\ell_t(e)$ identified with $e$.
  \end{definition}

Note that by our construction, each $a_i$ crosses each $b_j$ exactly once, while the order within each block is preserved. Thus $\ell_\tau(e)=s(v_1)s(v_0)$, giving $T(0)=T_0$ and $T(\tau)=T_1$ in $O(\log N)$ parallel steps.

  \begin{figure}[H]
    \centering
    \begin{tikzpicture}[x=1cm,y=1cm,line cap=round,line join=round,>=Latex]

    \node[above] (a1) at (-2,2) {$a$};
    \node[above] (b1) at (-1,2) {$b$};
    \node[above] (c1) at ( 0,2) {$c$};
    \node[above] (d1) at ( 1,2) {$d$};

    \coordinate (rootL) at (-0.5,0);
    \coordinate (int1) at (-1.5,1);
    \coordinate (int2) at (0.5,1);

    \draw (rootL) -- (int1) -- (a1.south);
    \draw (int1) -- (b1.south);
    \draw (rootL) -- (int2) -- (c1.south);
    \draw (int2) -- (d1.south);

    \draw[->,thick] (1.8,1) -- (3.2,1);

    \node[above] (a2) at (4,2) {$a$};
    \node[above] (c2) at (5,2) {$c$};
    \node[above] (b2) at (6,2) {$b$};
    \node[above] (d2) at (7,2) {$d$};

    \coordinate (rootR) at (5.5,0);
    \coordinate (int3) at (4.5,1);
    \coordinate (int4) at (6.5,1);

    \draw (rootR) -- (int3) -- (a2.south);
    \draw (int3) -- (c2.south);
    \draw (rootR) -- (int4) -- (b2.south);
    \draw (int4) -- (d2.south);

    \end{tikzpicture}

    \caption{The leaf swap operation on trees. Viewing $a,b,c,d$ as corresponding to $00,01,10,11$ gives us the local leaf swap in~\Cref{def:leaf-swap}.}
    \label{fig:tree-swap}
  \end{figure}

  Now we describe the local cone with its left and right boundary as the two trees before and after a local leaf swap operation.
  The implementation of the operation in~\cite{berdnikov2022parsimonious} proceeds by contracting the tree root and creating a node with four leaves, as shown in~\Cref{fig:parsi-tree-rotate}.
 \begin{figure}[H]
    \centering
    \begin{tikzpicture}[x=1cm,y=1cm,line cap=round,line join=round,>=Latex]

    \node[above] (a1) at (-4,2) {$a$};
    \node[above] (b1) at (-3,2) {$b$};
    \node[above] (c1) at (-2,2) {$c$};
    \node[above] (d1) at (-1,2) {$d$};

    \coordinate (rootL) at (-2.5,0);
    \coordinate (int1) at (-3.5,1);
    \coordinate (int2) at (-1.5,1);

    \draw (rootL) -- (int1) -- (a1.south);
    \draw (int1) -- (b1.south);
    \draw (rootL) -- (int2) -- (c1.south);
    \draw (int2) -- (d1.south);

    \node[above] (aM) at (1,2) {$a$};
    \node[above] (bM) at (2,2) {$b$};
    \node[above] (cM) at (3,2) {$c$};
    \node[above] (dM) at (4,2) {$d$};

    \coordinate (rootM) at (2.5,0);
    \coordinate (intM) at (2.5,1);

    \draw (rootM) -- (intM);
    \draw (intM) -- (aM.south);
    \draw (intM) -- (bM.south);
    \draw (intM) -- (cM.south);
    \draw (intM) -- (dM.south);

    \node[above] (a2) at (6,2) {$a$};
    \node[above] (c2) at (7,2) {$c$};
    \node[above] (b2) at (8,2) {$b$};
    \node[above] (d2) at (9,2) {$d$};

    \coordinate (rootR) at (7.5,0);
    \coordinate (int3) at (6.5,1);
    \coordinate (int4) at (8.5,1);

    \draw (rootR) -- (int3) -- (a2.south);
    \draw (int3) -- (c2.south);
    \draw (rootR) -- (int4) -- (b2.south);
    \draw (int4) -- (d2.south);

    \draw[dashed] (rootL) -- (rootM) -- (rootR);
    \draw[dashed] (int1) .. controls (0,1.5) .. (intM) .. controls (5,1.5) .. (int4);
    \draw[loosely dashed] (int2) -- (intM) -- (int3);
    \foreach \L/\M/\R in {a1/aM/a2, b1/bM/b2, c1/cM/c2, d1/dM/d2}{
        \draw[dashed]
            (\L.south) .. controls + (0,1) and +(0,1)  .. (\M.south)
                        .. controls + (0,1) and +(0,1) .. (\R.south);
    }

    \fill[blue, opacity=0.2] (rootL) -- (int2) -- (intM) -- (rootM) -- cycle;
    \fill[brown, opacity=0.2] (rootL) -- (int1) .. controls (0,1.5) .. (intM) -- (rootM) -- cycle;
    \fill[red, opacity=0.2] (intM) -- (aM.south) .. controls + (0,1) and +(0,1) .. (a2.south) -- (int3) -- cycle;
    \fill[green, opacity=0.2] (intM) -- (bM.south) .. controls + (0,1) and +(0,1) .. (b2.south) -- (int4) .. controls (5,1.5) .. cycle;
    \end{tikzpicture}

    \caption{The local cone structure in the construction of~\cite{berdnikov2022parsimonious}.
            Due to overlaps, only a few representative faces are shown to illustrate the idea.}
    \label{fig:parsi-tree-rotate}
  \end{figure}
  
  This gadget produces a complex $Y$ with the following properties:
  \begin{itemize}
    \item Each vertex is incident to $\le 9$ edges.
    \item Each edge is incident to $\le 4$ faces.
    \item Each face has $\le 5$ edges.
  \end{itemize}
  We can construct our global cone by composing these local cones in one layer to implement the global leaf swap operation, and by $O(\log N)$ such layers we can transform $T_0$ to $T_1$. 
If the binary tree is complete\footnote{A binary tree is complete if it has every possible leaf at every level.} and of size $N$, it is immediate to see that the number of local leaf swap operations is $O(N\log N)$, as the tree is of size $O(N)$ and we have $O(\log N)$ layers, and each gadget is of bounded degree, giving our cone size. 

   However, if we construct the tree from its binary representation, the tree may not be complete, and it can have long chains of degree-2 internal nodes. We contract these chains and define the resulting tree as the effective tree, which is of size $O(N)$ as we have $N$ leaves.
   We call a local leaf swap operation on the binary representation effective if it changes the structure of the effective tree. By enumerating all the possibilities in~\Cref{fig:tree-swap} for incomplete binary trees, we have the following claim:
  
\begin{claim}
  A local leaf swap operation is effective only if the root has valid left and right children. Moreover, on the effective tree, the operation in~\Cref{fig:tree-swap} can be characterized by one of the three cases:
  \begin{itemize}
    \item swapping the left and right children of the root.
    \item tree rotation as shown in~\Cref{fig:tree-rotation}.
    \item a constant number (at most 6) of combinations of swap and tree rotation.
  \end{itemize}
\end{claim}
Moreover, one can observe that parallelizable swap operations on the binary representation can still be parallelized on the effective tree. 
Since the effective operation requires the root to have valid left and right child, it implies that the root is a degree-3 vertex and cannot be contracted.
Therefore, after contracting the vertices with a unique child, original local leaf swap operations on depth $k$ and $k+2$ still act on
disjoint subtrees, and thus can be performed in parallel.

  We provide the spatial gadget for implementing the swap and tree rotation operation as follows.
  The swap gadget is obtained directly by cross-connecting the two leaves.
  The gadgets shown in~\Cref{fig:tree-rotation-1} implement tree rotations explicitly.
  These gadgets now produce a complex $Y$ with the following properties:
  \begin{itemize}
    \item Each vertex is incident to $\le 5$ edges.
    \item Each edge is incident to $\le 4$ faces.
    \item Each face has $\le 5$ edges.
  \end{itemize}

  \begin{figure}[H]
    \centering
    \begin{tikzpicture}[x=1cm,y=1cm,line cap=round,line join=round,>=Latex]

      \node[above] (a0) at (-1,2) {$a$};
      \node[above] (b0) at (0,2) {$b$};
      \node[above] (c0) at (1,2) {$c$};

      \coordinate (root0) at (0,0);
      \coordinate (int0) at (-0.5,1);

      \draw (root0) -- (int0) -- (a0.south);
      \draw (int0) -- (b0.south);
      \draw (root0) -- (c0.south);

      \draw[->,thick] (1.8,1) -- (3.2,1);

      \node[above] (a3) at (4,2) {$a$};
      \node[above] (b3) at (5,2) {$b$};
      \node[above] (c3) at (6,2) {$c$};

      \coordinate (root3) at (5,0);
      \coordinate (int5) at (5.5,1);

      \draw (root3) -- (a3.south);
      \draw (root3) -- (int5) -- (b3.south);
      \draw (int5) -- (c3.south);
    \end{tikzpicture}

    \caption{Tree rotation can be used to swap the leaves. }
    \label{fig:tree-rotation}
  \end{figure}

  \newcommand{\insetdraw}[3]{\begin{scope}
      \clip #1;
      \draw[draw=#2, line width=2*#3, line join=round, line cap=round, opacity=0.5] #1;
    \end{scope}
  }
  \begin{figure}[H]
    \centering
    \begin{tikzpicture}[x=1cm,y=1cm,line cap=round,line join=round,>=Latex]

      \node[above] (a1) at (-2,2) {$a$};
      \node[above] (b1) at (-1,2) {$b$};
      \node[above] (c1) at (0.5,2) {$c$};

      \coordinate (rootL) at (-0.5,0);
      \coordinate (downL) at (-0.5,-1);
      \coordinate (int1) at (-1.5,1);

      \draw (rootL) -- (int1) -- (a1.south);
      \draw (int1) -- (b1.south);
      \draw (rootL) -- (c1.south);

      \node[above] (a2) at (4.5,2) {$a$};
      \node[above] (b2) at (6,2) {$b$};
      \node[above] (c2) at (7,2) {$c$};

      \coordinate (rootR) at (5.5,0);
      \coordinate (downR) at (5.5,-1);
      \coordinate (int4) at (6.5,1);

      \draw (rootR) -- (a2.south);
      \draw (rootR) -- (int4) -- (b2.south);
      \draw (int4) -- (c2.south);
      \draw (rootL) -- (downL);
      \draw (rootR) -- (downR);

      \draw[dashed] (rootL)--(rootR);
      \draw[dashed] (int1)--(int4);
      \draw[dashed] (downL)--(downR);

      \draw[dashed] (a1.south) .. controls + (0,1.5) and +(0,1.5) .. (a2.south);
      \draw[dashed] (b1.south) .. controls + (0,1.5) and +(0,1.5) .. (b2.south);
      \draw[dashed] (c1.south) .. controls + (0,1.5) and +(0,1.5) .. (c2.south);

      \def\BlueRegionPath{(rootL) -- (int1) -- (a1.south)
        .. controls +(0,1.5) and +(0,1.5) .. (a2.south)
        -- (rootR) -- cycle}

      \def\RedRegionPath{(int1) -- (b1.south)
        .. controls +(0,1.5) and +(0,1.5) .. (b2.south)
        -- (int4) -- cycle}

      \def\GreenRegionPath{(rootL) -- (c1.south)
        .. controls +(0,1.5) and +(0,1.5) .. (c2.south)
        -- (int4) -- (rootR) -- cycle}

      \def\BrownRegionPath{(rootL) -- (int1) -- (int4) -- (rootR) -- cycle}

      \def\OrangeRegionPath{(downL) -- (rootL) -- (rootR) -- (downR) -- cycle}

      \fill[blue,  opacity=0.2] \BlueRegionPath;
      \fill[red,   opacity=0.2] \RedRegionPath;
      \fill[green, opacity=0.2] \GreenRegionPath;
      \fill[brown, opacity=0.4] \BrownRegionPath;
      \fill[orange,opacity=0.2] \OrangeRegionPath;

      \insetdraw{\BlueRegionPath}{blue}{2pt}
      \insetdraw{\RedRegionPath}{red}{2pt}
      \insetdraw{\GreenRegionPath}{green}{2pt}
      \insetdraw{\BrownRegionPath}{brown}{2pt}
      \insetdraw{\OrangeRegionPath}{orange}{2pt}
    \end{tikzpicture}

    \caption{The tree rotation operation and its corresponding cone gadgets. We have omitted the edge connecting to other parts of the tree for clarity.}
    \label{fig:tree-rotation-1}
  \end{figure}

  For one layer of parallelized local leaf swaps, the total size of the local cones created is on the order of the size of the effective tree $O(N)$.
  As we can parallelize the local leaf swap operations in $O(\log N)$ layers, the new cone is still of size $O(N\log N)$.
  We now define the boundary of the cone. In order to obtain a boundary $X'$ that is homotopy equivalent to the original graph $X$, we observe
  the following facts:
  \begin{itemize}
    \item The leaves of the trees $T(t)$ correspond to the edges of the original graph $X$, and the vertices of the same edge label $e$ are connected as a
    chain following the time order, forming the chain $P_e$.
    \item For each leaf of the vertex cone tree $T_0^v$, it corresponds to a vertex $v_0$ in $V_0$. In $T_0$, the subtree with root $v_0$ corresponds to an incidence tree with at most three leaves,
    each corresponding to a neighboring edge of $v_0$. A similar argument holds for vertices in $V_1$.
  \end{itemize}
  The boundary $X'$ consists of the leaf chains of $T(t)$, together with the vertex subtrees of $T_0$ and $T_1$. The reader can find a visualization of the cone $Y$ and its boundary $X'$ in~\Cref{fig:cone-boundary}.
  It is easy to verify that the boundary $X'$ is contractible to $X$.

\begin{figure}[p]
\centering
\begin{tikzpicture}[x=1.06cm,y=1.06cm,
  every node/.style={font=\small},
  coneedge/.style={draw=black!32,line width=.9pt},
  leftinc/.style={draw=orange!85!black,line width=1.35pt},
  rightinc/.style={draw=teal!80!black,line width=1.35pt},
  leafchain/.style={draw=blue!65!black,line width=1.1pt},
  dot/.style={circle,fill=black,inner sep=1.5pt}]
\path[fill=black!3] (-.8,-.55) -- (2.6,.1) -- (2.8,4.7)
  -- (-.6,4.05) -- cycle;
\path[fill=black!3] (7.2,.45) -- (10.6,1.1) -- (10.8,5.7)
  -- (7.4,5.05) -- cycle;
\coordinate (root0) at (.9,-.05);
\coordinate (root1) at (8.9,.95);
\coordinate (v1) at (.25,1.7);
\coordinate (v2) at (1.9,2.15);
\coordinate (bv1) at (.78,2.67);
\coordinate (bv2) at (2.02,3.48);
\coordinate (b1) at (9.65,2.05);
\coordinate (w1) at (8.15,2.8);
\coordinate (w2) at (9.05,3.12);
\coordinate (w3) at (9.95,3.45);
\foreach \i in {1,...,6} {
  \coordinate (l\i) at ({.45*(\i-1)},{3.2+.28*(\i-1)});
  \coordinate (r\i) at ({8+.45*(\i-1)},{4.2+.28*(\i-1)});
}
\path[fill=blue!13] (root0) -- (v2) -- (bv2) -- (l6)
  -- (r6) -- (w3) -- (b1) -- (root1) -- cycle;
\path[fill=blue!10] (root0) -- (v1) -- (bv1) -- (l3)
  -- (r5) -- (w3) -- (b1) -- (root1) -- cycle;
\path[fill=blue!8] (root0) -- (v2) -- (bv2) -- (l5)
  -- (r4) -- (w2) -- (b1) -- (root1) -- cycle;
\path[fill=blue!6] (root0) -- (v1) -- (bv1) -- (l2)
  -- (r3) -- (w2) -- (b1) -- (root1) -- cycle;
\path[fill=blue!5] (root0) -- (v2) -- (l4)
  -- (r2) -- (w1) -- (root1) -- cycle;
\path[fill=blue!3] (root0) -- (v1) -- (l1)
  -- (r1) -- (w1) -- (root1) -- cycle;
\foreach \i/\j in {6/6,5/4,4/2,3/5,2/3,1/1} {
  \coordinate (m\i) at
    ({.45*(\i-1)+.16*(8+.45*(\j-\i))},
     {3.2+.28*(\i-1)+.16*(1+.28*(\j-\i))});
  \coordinate (n\j) at
    ({.45*(\i-1)+.84*(8+.45*(\j-\i))},
     {3.2+.28*(\i-1)+.84*(1+.28*(\j-\i))});
}
\coordinate (mr) at (2.18,.11);
\coordinate (mv1) at (1.53,1.86);
\coordinate (mv2) at (3.18,2.31);
\coordinate (mbv1) at (2.06,2.83);
\coordinate (mbv2) at (3.30,3.64);
\draw[black!22,line width=.65pt] (mv1) -- (mr) -- (mv2);
\draw[black!22,line width=.65pt] (m1) -- (mv1) -- (mbv1);
\draw[black!22,line width=.65pt] (m2) -- (mbv1) -- (m3);
\draw[black!22,line width=.65pt] (m4) -- (mv2) -- (mbv2);
\draw[black!22,line width=.65pt] (m5) -- (mbv2) -- (m6);
\coordinate (nr) at (7.62,.79);
\coordinate (nb1) at (8.37,1.89);
\coordinate (nw1) at (6.87,2.64);
\coordinate (nw2) at (7.77,2.96);
\coordinate (nw3) at (8.67,3.29);
\draw[black!22,line width=.65pt] (nw1) -- (nr) -- (nb1);
\draw[black!22,line width=.65pt] (nw2) -- (nb1) -- (nw3);
\draw[black!22,line width=.65pt] (n1) -- (nw1) -- (n2);
\draw[black!22,line width=.65pt] (n3) -- (nw2) -- (n4);
\draw[black!22,line width=.65pt] (n5) -- (nw3) -- (n6);
\draw[black!65,line width=1.2pt] (root0) -- (root1);
\foreach \i/\j in {6/6,5/4,4/2,3/5,2/3,1/1} {
  \draw[white,line width=3.6pt] (l\i) -- (r\j);
  \draw[leafchain] (l\i) -- (r\j);
}
\foreach \i/\j in {6/6,5/4,4/2,3/5,2/3,1/1} {
  \foreach \t in {.23,.47,.73}
    \fill[blue!65!black]
      ({.45*(\i-1)+\t*(8+.45*(\j-\i))},
       {3.2+.28*(\i-1)+\t*(1+.28*(\j-\i))}) circle (1.5pt);
}
\node[blue!65!black,fill=white,inner sep=2pt] at (4.35,3.55) {$P_{e_1}$};
\node[fill=white,inner sep=3pt,align=center] at (5.05,2.7)
  {leaf chains $P_e$};
\draw[coneedge] (v1) -- (root0) -- (v2);
\draw[coneedge] (w1) -- (root1) -- (b1);
\draw[coneedge] (w2) -- (b1) -- (w3);
\foreach \p in {root0,root1,b1}
  \node[circle,fill=black!35,inner sep=1.5pt] at (\p) {};
\draw[leftinc] (l1) -- (v1) -- (bv1);
\draw[leftinc] (l2) -- (bv1) -- (l3);
\draw[leftinc] (l4) -- (v2) -- (bv2);
\draw[leftinc] (l5) -- (bv2) -- (l6);
\draw[rightinc] (r1) -- (w1) -- (r2);
\draw[rightinc] (r3) -- (w2) -- (r4);
\draw[rightinc] (r5) -- (w3) -- (r6);
\foreach \p in {v1,v2,bv1,bv2,w1,w2,w3}
  \node[dot] at (\p) {};
\foreach \i in {1,...,6} {
  \node[dot] at (l\i) {};
  \node[above=3pt,font=\footnotesize] at (l\i) {$e_{\i}$};
}
\foreach \j/\i in {1/1,2/4,3/2,4/5,5/3,6/6} {
  \node[dot] at (r\j) {};
  \node[above=3pt,font=\footnotesize] at (r\j) {$e_{\i}$};
}
\node[left=3pt] at (v1) {$u_{v_1}$};
\node[right=3pt] at (v2) {$u_{v_2}$};
\node[left=3pt] at (w1) {$u_{w_1}$};
\node[left=3pt] at (w2) {$u_{w_2}$};
\node[right=3pt] at (w3) {$u_{w_3}$};
\node[black!55] at (.25,.5) {$T_0^v$};
\node[black!55] at (9.8,1.65) {$T_1^v$};
\node[below=8pt] at (root0) {$T(0)=T_0$};
\node[below=8pt] at (root1) {$T(\tau)=T_1$};
\draw[->,black!65,line width=.8pt] (2.65,-.65) -- (6.9,-.12)
  node[midway,below=5pt] { $t:0\longrightarrow\tau$};
\end{tikzpicture}
\caption{An example of the cone $Y$ and its boundary $X'$, where the original graph $X$ is $K_{2,3}$. We use the orange and teal edges to represent the incidence trees of the vertices in $V_0$ and $V_1$, respectively. The blue edges represent
 the leaf chains $P_e$ for each edge $e\in E(X)$, and the gray region represents the interior of the cone $Y$. The reader can verify that the boundary $X'$ is contractible to the original graph $X$.}
\label{fig:cone-boundary}
\end{figure}

  To obtain a cone with the boundary exactly isomorphic to $X$, we introduce a final gadget to contract the graph $X'$.
  We will introduce the following additional vertices:
  \begin{itemize}
    \item For each vertex $v\in V$, we introduce $v$ and one copy of $v$ labeled as $s_v$.
    \item For each edge $e\in E$, we introduce a vertex $r_e$.
  \end{itemize}
  We now introduce the following edges:
  \begin{itemize}
    \item We introduce every edge $e=(v,w)$ in $X$, providing the new boundary of our cone.
    \item For each $v\in V$, we introduce an edge between $v$ and $s_v$, and an edge between $s_v$ and the node representing $v$ in $T_0/T_1$, labeled as $u_v$.\footnote{For a general non-bipartite graph, since we first made the graph bipartite, the corresponding vertices of $V$ are all in $T_0$.}
    \item For each $r_e$ for $e=(v,w)$, we connect $r_e$ with $s_v$, $s_w$, and the two leaves in $T_0/T_1$ that represent the edge $e$, we denote as $u_{v,e}$ and $u_{w,e}$.
  \end{itemize}
  Finally we introduce the following faces:
  \begin{itemize}
    \item The face formed by the loop $vws_wr_es_vv$.
    \item The face formed by the loop $s_v r_e [u_{v,e}\to u_v] s_v$. The path from $u_{v,e}$ to $u_v$ has length at most 2.
  \item A large disk $D_e$ formed by the loop $u_{v,e} r_e u_{w,e} P_e u_{v,e}$. We can subdivide the disk by introducing edges between $e_i$ and $e_{\tau-i}$.
  \end{itemize}
  We can visualize the construction in~\Cref{fig:boundary-contract}. The readers can check that with the new appended vertices, edges, and faces, the final cone $X$ still satisfies the following bounds:
   \begin{itemize}
    \item Each vertex is incident to $\le 5$ edges.
    \item Each edge is incident to $\le 4$ faces.
    \item Each face has $\le 5$ edges.
  \end{itemize}
  We will not explicitly count the incidence relation of every cell, but we emphasize several crucial points:
  \begin{itemize}
    \item Each vertex $s_v$ is incident to 3 vertices of $r_e$, and one $v$, and one $u_v$, thus incident to 5 edges.
    \item The edge $s_vv$ and $s_vu_v$ are both incident to 3 faces, and the edges on the path from $u_{v,e}$ to $u_v$ are also incident to at most 4 faces, including at most two faces from 
          our gadget and at most two faces from the original cone, as the subtree with root $u_v$ has at most three leaves.
    \item Each vertex on the path $P_e$ has a new neighboring edge, as they originally have degree 3 from our construction, and thus are incident to at most 4 edges. 
    \item The boundary of our new cone also satisfies that each vertex is incident to $\leq 4$ edges, and each edge is incident to $1$ face, as required in~\Cref{lemma:cone-cell-complex}.
  \end{itemize}
  The other degree bounds are direct from our figures, and we achieve the desired bounds on the number of vertices, edges, and faces in our final cone.
  
\begin{figure}[p]
\centering
\begin{minipage}[t]{0.62\textwidth}
\centering
\textbf{(a) Boundary gadget at $e=(v,w)$}\par\medskip
\begin{tikzpicture}[x=.83cm,y=.83cm,
  every node/.style={font=\small},
  edge/.style={draw=black,line width=.8pt},
  trajectory/.style={draw=blue!65!black,line width=1.2pt},
  vertex/.style={circle,fill=black,inner sep=1.6pt}]
\coordinate (v) at (-2.5,3.1);
\coordinate (w) at (2.5,3.1);
\coordinate (sv) at (-3.4,1.65);
\coordinate (sw) at (3.4,1.65);
\coordinate (r) at (0,.35);
\coordinate (uv) at (-3.7,.05);
\coordinate (uw) at (3.7,.05);
\coordinate (bv) at (-3.05,-.65);
\coordinate (bw) at (3.05,-.65);
\coordinate (av) at (-2.15,-1.2);
\coordinate (aw) at (2.15,-1.2);
\path[fill=black!5] (uv) -- (bv) -- (av)
  -- (-1.65,-2.05) -- (-.85,-2.55) -- (0,-2.75)
  -- (.85,-2.55) -- (1.65,-2.05) -- (aw) -- (bw) -- (uw)
  to[out=-20,in=65] (4.25,-2.2)
  to[out=-115,in=0] (0,-3.85)
  to[out=180,in=-65] (-4.25,-2.2)
  to[out=115,in=200] (uv) -- cycle;
\path[fill=orange!13] (v) to[out=25,in=155] (w)
  -- (sw) -- (r) -- (sv) -- cycle;
\path[fill=cyan!14] (sv) -- (r) -- (av) -- (bv) -- (uv) -- cycle;
\path[fill=cyan!14] (sw) -- (r) -- (aw) -- (bw) -- (uw) -- cycle;
\path[fill=green!12] (r) -- (av)
  -- (-1.65,-2.05) -- (-.85,-2.55) -- (0,-2.75)
  -- (.85,-2.55) -- (1.65,-2.05) -- (aw) -- cycle;
\draw[edge,line width=1.3pt] (v) to[out=25,in=155]
  node[above=3pt] {$e\in X$} (w);
\draw[edge] (v) -- (sv) -- (uv) -- (bv) -- (av) -- (r)
  -- (aw) -- (bw) -- (uw) -- (sw) -- (w);
\draw[edge] (sv) -- (r) -- (sw);
\draw[trajectory] (av) -- (-1.65,-2.05) -- (-.85,-2.55)
  -- (0,-2.75) -- (.85,-2.55) -- (1.65,-2.05) -- (aw);
\foreach \p in {v,w,sv,sw,uv,uw,bv,bw,av,aw}
  \node[vertex] at (\p) {};
\node[circle,fill=orange!80!black,inner sep=2.2pt] at (r) {};
\node[above left=2pt] at (v) {$v$};
\node[above right=2pt] at (w) {$w$};
\node[left=3pt] at (sv) {$s_v$};
\node[right=3pt] at (sw) {$s_w$};
\node[left=3pt] at (uv) {$u_v$};
\node[right=3pt] at (uw) {$u_w$};
\node[below left=1pt] at (av) {$u_{v,e}$};
\node[below right=1pt] at (aw) {$u_{w,e}$};
\node[above=4pt] at (r) {$r_e$};
\node at (0,-1.5) {disk $D_e$};
\node[blue!65!black,below=3pt] at (0,-2.75) {$P_e$};
\node[font=\footnotesize,text=black!65] at (0,-3.55) {existing cone $Y$};
\end{tikzpicture}
\end{minipage}\hfill
\begin{minipage}[t]{0.35\textwidth}
\centering
\textbf{(b) Subdivision of $D_e$}\par\medskip
\begin{tikzpicture}[x=.85cm,y=.83cm,
  every node/.style={font=\small},
  edge/.style={draw=black,line width=.8pt},
  vertex/.style={circle,fill=black,inner sep=1.5pt}]
\coordinate (r) at (0,3.0);
\coordinate (x0) at (-1.65,1.3);
\coordinate (x1) at (-1.65,.1);
\coordinate (x2) at (-1.2,-1.05);
\coordinate (x3) at (0,-1.9);
\coordinate (x4) at (1.2,-1.05);
\coordinate (x5) at (1.65,.1);
\coordinate (x6) at (1.65,1.3);
\path[fill=green!12] (r) -- (x0) -- (x1) -- (x2) -- (x3)
  -- (x4) -- (x5) -- (x6) -- cycle;
\draw[edge] (x0) -- (r) -- (x6);
\draw[draw=blue!65!black,line width=1.2pt]
  (x0) -- (x1) -- (x2) -- (x3) -- (x4) -- (x5) -- (x6);
\foreach \i/\j in {1/5,2/4}
  \draw[edge] (x\i) -- (x\j);
\foreach \i in {0,...,6} \node[vertex] at (x\i) {};
\node[circle,fill=orange!80!black,inner sep=2.2pt] at (r) {};
\node[above=3pt] at (r) {$r_e$};
\node[above left=1pt] at (x0) {$e_0$};
\node[left=2pt] at (x1) {$e_1$};
\node[left=2pt] at (x2) {$e_2$};
\node[below=3pt] at (x3) {$e_3$};
\node[right=2pt] at (x4) {$e_4$};
\node[right=2pt] at (x5) {$e_5$};
\node[above right=1pt] at (x6) {$e_6$};
\node[align=center,font=\footnotesize] at (0,-2.95)
  {$e_0=a_{v,e}$, $e_6=a_{w,e}$\\[3pt]
   chords $x_i x_{6-i}$, $i=1,2$};
\end{tikzpicture}
\end{minipage}
\caption{The boundary gadget for the edge $e=(v,w)$. Figure (a) visualizes the appended boundary gadget to the existing cone $Y$, and the new boundary $e=(v,w)$ is drawn with a thicker line. Figure (b) exhibits how we subdivide the
large disk $D_e$ into smaller faces by adding chords.}
\label{fig:boundary-contract}
\end{figure}

\end{proof}

\begin{remark}
  Our construction is in the same spirit as the construction in~\cite{yuanParsimoniousQuantumLowDensity2026}, but we have obtained better constants by using a different local gadget, and our description is different from theirs.
  The authors start from the complete binary tree construction, which gives an overestimation of order $O(N^2\log N)$ for the cone size. They use steps called pruning to remove irrelevant edges and contract long paths leading to the leaves. 
  Effectively, our description is another way to achieve the same result. On the local cone structure, the authors follow the construction in~\cite{berdnikov2022parsimonious}, while we provide an alternative construction using tree rotations, which has better local parameters. 
  We also improve the final boundary gadget to reduce the degree of the cone.
\end{remark}

Utilizing \Cref{lemma:cone-simplicial-complex}, we obtain the following theorem,
  which was also briefly remarked in \cite[Section 3]{berdnikov2022parsimonious}.
\begin{theorem}
  \label{thm:weight reduction-2D-simplicial-complex}
  Let $X^2$ be a 2-dimensional simplicial complex.
  Then there exists a homotopy equivalent simplicial complex $\cX$
    with bounded degree such that
  \begin{equation}
    |\cX| = O\left(\sum_{v \in X(0)} |X_{\ge v}(2)| \log |X_{\ge v}(2)|\right),
  \end{equation}
  where $X(0)$ is the set of vertices in $X$,
    and $|X_{\ge v}(2)|$ is the number of triangles containing the vertex $v$.

  In particular, if each vertex in $X$ is adjacent to $\le w$ many vertices,
    then $|X_{\ge v}(2)| \le w^2$,
    which implies
  \begin{equation}
    |\cX| = O(|X(2)| \cdot \log w).
  \end{equation}
\end{theorem}
Notice that the factor $\log N$ in \Cref{thm:weight reduction-simplicial-complex-dense},
  is replaced by $\log w$ in \Cref{thm:weight reduction-2D-simplicial-complex}.

A similar bound holds in higher dimensions,
  $|\cX| = O_t(\sum_{v \in X(0)} |X_{\ge v}(t)| \log^{t-1} |X_{\ge v}(t)|)$.
However, the proof we are aware of is not a direct black-box application of the preceding theorems.
Instead, it requires the use of the ``priority string'' introduced in \cite{berdnikov2022parsimonious}.
We omit the details here, as this higher-dimensional generalization
  is not needed for our discussion of quantum codes.
The key observation is that the length of the ``priority string'' can be reduced to $\log |X_{\ge v}(t)|$,
  as opposed to the original $\log |X(t)|$.

The purpose of this proof is to draw a parallel
  with the upcoming construction on the weight reduction of quantum codes in \Cref{sec:weight reduction-quantum-code}.

\begin{proof}
  We begin by partitioning the original complex $X$ into subcomplexes,
    as illustrated in \Cref{fig:dual-structure}.
  These subcomplexes are analogous to dual cells in Poincaré duality,
    where each cell in the original complex corresponds bijectively to one of these subcomplexes.

  \begin{figure}[H]
    \centering
    \includegraphics[width=0.8\linewidth]{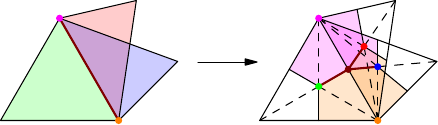}
    \caption{We perform a barycentric subdivision and
              associate each cell in the original complex
              with one of the newly formed subcomplexes.}
    \label{fig:dual-structure}
  \end{figure}

  Specifically, the subcomplexes are derived from the barycentric subdivision of $X$,
    denoted $X'$.
  In the barycentric subdivision,
    we create new vertices at the midpoint of edges
    and the centers of the triangles.
  New edges and faces are then formed as illustrated in \Cref{fig:dual-structure}.

  There is a natural bijection between the cells of $X$ and the vertices of $X'$,
    $\bigsqcup_{i=0}^2 X(i) \cong X'(0)$,
    where $\bigsqcup$ denotes disjoint union.
  Moreover, each cell in $X'$ corresponds to a sequence of nested cells in $X$,
    $(x_0, x_1, ..., x_m)$,
    where $x_0 \prec x_1 \prec ... \prec x_m$ and each $x_i \in X$. Here we use $\prec$ to show that $x_i$ is contained in the cell $x_{i+1}$.

  This nested structure induces a natural partition of $X'$ into subcomplexes indexed by the cells of $X$.
  For each cell $x \in X$,
    we define the associated subcomplex $X'_x$
    as the closure of the set of cells in $X'$ that have $x$ as their first element in the nesting sequence:
    \begin{equation}
      X'_x = \cl(\{(x=x_0, x_1, ..., x_m): x \prec x_1 \prec ... \prec x_m, x_i \in X, m \in \ZZ^{\ge 0}\}).
    \end{equation}
  The dimension of $X'_x$ is $2 - \dim x$.
  That means if $x$ is a vertex, then $X'_x$ is 2-dimensional,
    while if $x$ is a triangle, then $X'_x$ is 0-dimensional.
  We define the boundary of $X'_x$, denoted $\partial X'_x$,
    as the union of cells in its attaching region:
    \begin{equation}
      \partial X'_x = \{(x_1, ..., x_m): x \prec x_1 \prec ... \prec x_m, x_i \in X, m \in \ZZ^{> 0}\}.
    \end{equation}

  Note that each subcomplex $X'_x$ has a cone structure,
    satisfying $X'_x \cong C(\partial X'_x)$,
    where $C(\cdot)$ denotes the cone construction.
  In other words, the complex $X'$ can be viewed as being built inductively through a sequence of cone constructions:
  \begin{itemize}
    \item We begin with the 0-dimensional subcomplex by placing a vertex for each $2$-cell $x \in X(2)$.
    \item Then, for each $1$ cell $y \in X(1)$,
      the subcomplex $X'_y$ is obtained by coning the 0-dimensional complex
      $\{x: y \prec x\} \cong \partial X'_y$.
    \item Finally, for each $0$ cell $z \in X(0)$,
      the subcomplex $X'_z$ is obtained by coning the 1-dimensional complex
      $\overline{\{(y, x): z \prec y \prec x\}} \cong \partial X'_z$.
  \end{itemize}

  With the subcomplexes $X'_x$ defined,
    we now proceed to construct the weight-reduced complex $\cX$.
  This is achieved by applying the weight reduction procedure to each subcomplex $X'_x$,
    in ascending order of their dimension.
  We build a sequence of intermediate complexes $\cX_0, \cX_1, \cX_2$,
    where the final complex is $\cX = \cX_2$.
  At each stage,
    we incrementally attach new structure to the complex constructed in the previous step.

  We begin with the 0-dimensional subcomplexes $X'_x$ for $x \in X(2)$.
  These require no sparsification.
  We simply create one vertex for each such subcomplex.
  Thus, we define $\cX_0 \cong \{*\}^{X(2)}$,
    a disjoint union of $|X(2)|$ vertices.

  Next, we consider the 1-dimensional subcomplexes $X'_y$ for $y \in X(1)$.
  Each such subcomplex is a star graph $S_n$ with one internal node and $n$ leaves,
    where $n$ is the number of 2-cells incident to $y$ in $X$.
  We sparsify $S_n$ into a bounded-degree graph $\cS_n$ on $2n$ vertices,
    as shown in \Cref{fig:star-graph}.
  We define $\cX_y = \cS_n$.
  The endpoints of $\cS_n$ correspond to the leaves of the original star and form the attaching region.
  These endpoints are then attached to the corresponding vertices in $\cX_0$,
    yielding the next complex $\cX_1$.

  \begin{figure}[H]
    \centering
    \begin{tikzpicture}
      \begin{scope}
        \filldraw (0,0) -- (2,2) circle (2pt);
        \filldraw (0,0) -- (1,2) circle (2pt);
        \filldraw (0,0) -- (0,2) circle (2pt);
        \filldraw (0,0) -- (-1,2) circle (2pt);
        \filldraw (0,0) -- (-2,2) circle (2pt);

        \fill (0,0) circle (2pt);
      \end{scope}

      \draw[->] (2.5,1) -- (5,1);

      \begin{scope}[shift={(8,0.5)}]
        \filldraw (2,0) circle (2pt) -- (2,1) circle (2pt);
        \filldraw (1,0) circle (2pt) -- (1,1) circle (2pt);
        \filldraw (0,0) circle (2pt) -- (0,1) circle (2pt);
        \filldraw (-1,0) circle (2pt) -- (-1,1) circle (2pt);
        \filldraw (-2,0) circle (2pt) -- (-2,1) circle (2pt);

        \draw (2,0) -- (-2,0);
      \end{scope}
    \end{tikzpicture}
    \caption{Sparsifying the star graph $S_n$ into $\cS_n$ with $2n$ vertices.}
    \label{fig:star-graph}
  \end{figure}

  Finally, we consider the 2-dimensional subcomplexes $X'_z$ for $z \in X(0)$.
  As discussed earlier, each $X'_z$ has a cone structure, satisfying $X'_z \cong C(\partial X'_z)$.
  In the complex $X'$, $\partial X'_z$ is the union of subcomplexes $X'_y$
    for all $y \in X(1)$ such that $y \succ z$,
  \begin{equation}
    \partial X'_z = \bigcup_{y \in X(1): y \succ z} X'_y.
  \end{equation}
  These subcomplexes $X'_y$ have been replaced by their weight-reduced counterparts $\cX_y$.
  This defines $\partial \cX_z \subset \cX(1)$ as
  \begin{equation}
    \partial \cX_z = \bigcup_{y \in X(1): y \succ z} \cX_y.
  \end{equation}
  We now construct $\cX_z$, which is intended to serve as the cone over $\partial \cX_z$.
  Of course, we do not take the cone directly,
    as this would introduce a high-degree vertex at the tip of the cone.
  Instead, we apply \Cref{lemma:cone-simplicial-complex}
    to obtain a bounded-degree complex $\cX_z$,
    such that the pair $(\cX_z, \partial \cX_z)$ is homotopy equivalent to $(C(\partial \cX_z), \partial \cX_z)$.
  This completes the construction of each $\cX_z$.
  By attaching all such components to $\cX_1$
    via their respective boundaries $\partial \cX_z$,
    we obtain the final complex $\cX = \cX_2$.

  To estimate the size of $\cX$,
    it suffices to count the number of $2$-cells.
  These arise exclusively from the $2$-subcomplexes in $\cX'$,
    which are associated with the vertices of the original complex $X$.
  For each vertex $v \in X(0)$,
    the corresponding $2$-subcomplex
    is a cone over a graph with $|X_{\ge v}(2)|$ edges.
  Applying the bound from \Cref{lemma:cone-simplicial-complex},
    we obtain the estimate
  \begin{equation}
    |\cX| = O\left(\sum_{v \in X(0)} |X_{\ge v}(2)| \log |X_{\ge v}(2)|\right),
  \end{equation}
  as desired.
\end{proof}

\subsection{Weight reduction of a dense simplicial complex}\label{sec:dense-complex}

We now study the regime where the simplicial complex is somewhat dense,
  where the degree of a vertex is large, say polynomial in $N$.

Below is a known result from \cite{berdnikov2022parsimonious}
  that applies to an arbitrary simplicial complex.
\begin{theorem}[{\cite[Theorem 3.1]{berdnikov2022parsimonious}}]
  \label{thm:weight reduction-simplicial-complex-dense}
  Let $X^t$ be a simplicial complex with $N$ $t$-simplices.
  Then there exists a homotopy equivalent simplicial complex $\cX$
    with bounded degrees $O_t(1)$,
    and $|\cX| = O_t(N \log^{t-1} N)$.
\end{theorem}

The new result below, motivated by the layer code construction \cite{williamsonLayerCodes2023},
  removes the logarithmic factors.
The result yields a better bound in the case of a dense simplicial complex,
  where a constant fraction of the $t$-simplices are present.
\begin{theorem}\label{thm:weight reduction-simplicial-complex-dense-new}
  Let $X^t$ be a simplicial complex with $N_v$ vertices.
  Then there exists a homotopy equivalent simplicial complex $\cX$
    with bounded degrees $O_t(1)$,
    and $|\cX| = O_t(N_v^{t+1})$.

  Furthermore, $\cX$ has a geometrically local embedding in
    $(t+1)$ dimensions, $\RR^{t+1}$,
    within $[0,N_v]^{t+1}$.
\end{theorem}
As we will discuss in \Cref{sec:application},
  this result allows us to construct quantum LDPC codes
  with distance $\tilde\Omega(n^{2/3})$, breaking the square-root barrier.
Furthermore, it can be embedded into $\RR^3$, which saturates the BPT bound in 3D up to a logarithmic factor.

The basic idea is to embed each vertex into a $t$-dimensional plane inside $\RR^{t+1}$,
  and to arrange these $N_v$ planes in generic position.
We place these planes so that the gaps between them are at least $O(1)$,
  which accounts for why the embedding has width $N_v$.
Notice that the configuration of these $N_v$ planes
  is homotopically equivalent to the $0$-skeleton of $X^t$.
We then attach additional pieces to glue these planes together,
  thereby forming the higher skeleta: $1$-skeleton, $2$-skeleton, and so on.

\begin{proof}
  We first consider the simpler case where $X^t$ is $(t+1)$-partite.
  This means the vertices are partitioned into $t+1$ subsets,
    which we denote by $V_0, V_1, \ldots, V_t$.

  We label the vertices in each part $V_i$ with integers in $[|V_i|]$,
    so that each vertex in $V_i$ carries a unique integer label
      between $1$ and $|V_i|$.
  This labeling induces a tuple of integers for every simplex,
    determined by the labels of its vertices.
  Specifically, each simplex $\sigma$
    defines a map $f_{\sigma}: \{0, 1, \ldots, t\} \to \NN \cup \{\star\}$,
    where $f_{\sigma}(i) = \star$ if $\sigma$ has no vertex in $V_i$,
      and otherwise $f_{\sigma}(i)$ is the integer label of the unique vertex of $\sigma$ in $V_i$.

  Such a map allows us to associate a $k$-simplex
    $\sigma$ with a $(t-k)$-dimensional plane
    within the integer grid structure of $\RR^{t+1}$.
  Consider the box $B = \prod_{i=0}^t [1, |V_i|] \subset \RR^{t+1}$.
  The structure we will construct fits within this box.
  Using the label of the $k$-simplex $\sigma$,
    we identify it with the plane $P_\sigma$ consisting of points in $B$,
    whose $i$-th coordinate is equal to $f_{\sigma}(i)$ if $f_{\sigma}(i) \ne \star$,
      and can take arbitrary values otherwise.
  For example,
    a vertex $v \in V_0$ corresponds to the codimension-$1$ hyperplane
    whose $0$-th coordinate equals the integer label of the vertex.
  As another example,
    a $t$-simplex corresponds to a single point with coordinate $f_{\sigma}(i)$.
  It is clear that if $\sigma \prec \tau$,
    there is an embedding $\varphi_{\tau, \sigma}: P_\tau \hookrightarrow P_\sigma$.
  See \Cref{fig:dense-simplex-embedding-1} for an illustration.

    \begin{figure}
        \centering
        \begin{tikzpicture}[line cap=round,line join=round,>=Latex,scale=1.1]

        \node[circle,fill=yellow,inner sep=2pt,label=left:{}] (A) at (-3,0) {};
        \node[circle,fill=red,inner sep=2pt,label=below:{}] (B) at (-1.5,-0.3) {};
        \node[circle,fill=blue,inner sep=2pt,label=above:{}] (C) at (-2,1.5) {};

        \draw[orange,thick] (A)--(B);
        \draw[green,thick] (B)--(C);
        \draw[violet,thick] (A)--(C);
         \fill[brown,opacity=0.3] (-3,0)--(-1.5,-0.3)--(-2,1.5)--cycle;

        \draw[->,thick] (0,0.7)--(1.5,0.7);

        \fill[blue,opacity=0.4] (4,0,0)--(4,2,0)--(4,2,2)--(4,0,2)--cycle;
        \fill[yellow,opacity=0.4] (3,1,0)--(5,1,0)--(5,1,2)--(3,1,2)--cycle;
        \fill[red,opacity=0.4] (3,0,1)--(5,0,1)--(5,2,1)--(3,2,1)--cycle;

        \draw[orange,thick] (3,1,1)--(5,1,1);
        \draw[green,thick] (4,0,1)--(4,2,1);
        \draw[purple,thick] (4,1,0)--(4,1,2);

        \node[circle,fill=brown,inner sep=2pt] (O) at (4,1,1) {};
        \end{tikzpicture}
        \caption{An example of our map that takes a 2-simplex and its boundaries in $X$ to planes in $\cX$. The corresponding objects on both sides use the same color.}
        \label{fig:dense-simplex-embedding-1}
    \end{figure}
  With the scaffold in place, we are ready to construct the desired bounded-degree complex.
  We build the structure iteratively,
    starting with the vertices.
  Recall that for each vertex $v$, we have an associated hyperplane $P_v$.
  Let $\cX_0 = \bigsqcup_{v \in X(0)} P_v$ be the disjoint union of these hyperplanes,
    each naturally cellulated into unit hypercubes.

  For each edge $e = (v_1, v_2)$\footnote{
    Since the complex has a partite structure,
      the vertices of a simplex have a natural ordering determined by the parts they belong to.},
    we have the associated plane $P_e$
    and embeddings $\varphi_{e, v_1}: P_e \hookrightarrow P_{v_1}$ and $\varphi_{e, v_2}: P_e \hookrightarrow P_{v_2}$.
  We glue $P_{v_1}$ and $P_{v_2}$ via $P_e$
    by forming $[0,1] \times P_e$,
    identifying $\{0\} \times P_e$ with the image of $\varphi_{e, v_1}$,
    and $\{1\} \times P_e$ with the image of $\varphi_{e, v_2}$.
  (Note that $[0,1]$ is the standard 1-simplex $\Delta_1$.)
  Repeating this for all edges defines $\cX_1$.
  Concretely, $\cX_1 = \left(\cX_0 \sqcup \bigsqcup_{e \in X(1)} ([0,1] \times P_e)\right) \big/ \sim$,
    where the identification is given by $\varphi_{e, v_1}(x) \sim (0,x)$
      and $\varphi_{e, v_2}(x) \sim (1,x)$
      for $x \in P_e$
      and $e = (v_1, v_2)$.
  See \Cref{fig:dense-simplex-embedding-2} for an illustration.

  \begin{figure}
        \centering
        \begin{tikzpicture}[line cap=round,line join=round,>=Latex,scale=1.1]

        \node[circle,fill=red,inner sep=2pt,label=right:{$v_1$}] (B) at (-1.8,-0.3) {};
        \node[circle,fill=blue,inner sep=2pt,label=right:{$v_2$}] (C) at (-2,1.5) {};

        \draw[violet,thick] (B)--(C);
        \node at (-1.5,0.6) {$e$};

        \draw[->,thick] (0,0.7)--(1.5,0.7);

        \draw (4,0,0)--(4,2,0)--(4,2,2)--(4,0,2)--cycle;
        \draw (3,0,1)--(5,0,1)--(5,2,1)--(3,2,1)--cycle;
        \fill[blue,opacity=0.3] (4,0,0)--(4,2,0)--(4,2,2)--(4,0,2)--cycle;
        \fill[red,opacity=0.3] (3,0,1)--(5,0,1)--(5,2,1)--(3,2,1)--cycle;
        \fill[violet,opacity=0.7] (4.5,0,1)--(4,0,1.5)--(4,2,1.5)--(4.5,2,1)--cycle;
        \node[right] at (5,0.5,0) {$P_{v_1}$};
        \node[above] at (4,2,0) {$P_{v_2}$};
        \node[below] at (4,-0.5,1) {$[0,1]\times P_{e}$};

        \end{tikzpicture}
        \caption{Figure showing the gluing procedure between $P_{v_1}$ and $P_{v_2}$. The violet region is $[0,1]\times P_{e}$.}
        \label{fig:dense-simplex-embedding-2}
    \end{figure}

  We repeat this process for higher-dimensional simplices.
  Let $\cX_j$ be obtained from $\cX_{j-1}$
    by attaching $\Delta_{j} \times P_\sigma$
    via the map $\partial \Delta_{j} \times P_\sigma \to \cX_{j-1}$.
  This map is defined as follows:
    Each face $\Delta_i \subset \partial \Delta_{j}$
    corresponds to a face $\tau \prec \sigma$,
    This correspondence induces a map
      $\Delta_{i} \times P_\sigma \hookrightarrow \Delta_i \times P_\tau \subset \cX_{j-1}$,
      via the embedding $\varphi_{\sigma, \tau}$.
    It is straightforward to verify that these maps, for different $\tau \prec \sigma$, are compatible.
  Hence, $\cX_j = \left( \cX_{j-1} \sqcup \bigsqcup_{\sigma \in X(j)} (\Delta_j \times P_\sigma) \right) \big / \sim$,
    where the identification is as described above.

  In this construction, the cells of $\cX_t$
    are products of simplices.
  To obtain a simplicial complex,
    we further subdivide these cells into simplices.
  Let $\cX$ denote the resulting cellulation of $\cX_t$.
  We do not describe this subdivision explicitly,
    but it is well understood and
    increases the degree by at most $t!$.
  This completes the construction of $\cX$.

  We now verify that $\cX$ has the desired properties.
  First, each vertex in $\cX_t$ has degree at most $2t + t=3t$.
  Specifically, $2t$ neighbors come from the structure of $P_v$,
    and $t$ from the structure of $\Delta_t$.
    Thus, the degree of $\cX$ is at most $3t \cdot t! = O_t(1)$.
  Second, we show that $\cX$ is homotopy equivalent to $X$.
  Consider the map $\cX \to X$,
    induced by the projections $\Delta_j \times P_\sigma \to \Delta_j$
      for every $j$-simplex $\sigma$ with $j = 0, ..., t$.
  It is straightforward to check that these projections are compatible,
    and the resulting map induces a homotopy equivalence.
  Third, it is clear that $\cX$ is geometrically local in $\RR^{t+1}$.

  For the more general case where $X^t$ is not $(t+1)$-partite,
    the proof remains largely the same.
  The only difference is that, without the partite structure,
    we must redefine the ``$t$-dimensional planes'' $P_v \subset \RR^{t+1}$ associated with the vertices $v$.
  The structure $P_\sigma$ for a simplex $\sigma$
    is then defined as the intersection of the corresponding $P_v$,
    where $P_\sigma = \bigcap_{v \in \sigma} P_v$.
  For example, for any $t+1$ vertices,
    the intersection of their $P_v$ is a single point,
    corresponding to the unique $t$-simplex they form.

  To define $P_v$ in the absence of a partite structure,
    we again order the vertices and label them with integers from $[|V|]$.
  We may occasionally abuse notation by identifying a vertex $v$ with its corresponding integer label.
  Now the box is $B = \prod_{i=0}^t [1, |V|] \subset \RR^{t+1}$.
  We work within a subregion $T\subset B$ where the coordinates are non-decreasing,
    $T = \{x \in B: x_0 \le x_1 \le ... \le x_t\}$.
  For each vertex $v$, define $P_v$ as the intersection of $T$ with
    the union of the $t+1$ coordinate hyperplanes
    where the $i$-th coordinate is fixed to $v$,
    $P_v = T \cap \bigcup_{i=0}^t \{x \in \RR^{t+1}: x_i = v\}$.
  Using these sets $P_v$,
    we then define $P_\sigma = \bigcap_{v \in \sigma} P_v$.
  See \Cref{fig:dense-simplex-embedding-3} for an illustration.
  These structures admit natural subdivisions into constant-size cells.
  We then construct $\cX$ by repeating the previous construction, replacing the earlier definition of $P_\sigma$ with this one.

  \begin{figure}[H]
    \centering
    \includegraphics[width=0.4\textwidth]{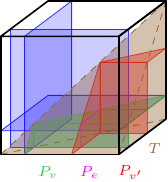}
    \caption{Illustration of the geometric structure of $P_v$ and $P_e$.
              The brown region is $T$.
              The green region $P_v$ is the intersection between the brown $T$ and the blue region formed by the hyperplanes.
              The red region $P_{v'}$ corresponds to another vertex $v'$.
              The intersection between $P_v$ and $P_{v'}$ defines $P_{e}$,
                for the edge $e=\{v,v'\}$.}
    \label{fig:dense-simplex-embedding-3}
  \end{figure}

  Finally, we verify that $\cX$ still satisfies the desired properties.
  First, the degree remains at most $3t \cdot t!$ by the same argument.
  Second, $\cX$ is homotopy equivalent to $X$,
    since each $P_\sigma$ is contractible,
    as in the partite case.
  Third, it is clear that $\cX$ is geometrically local in $B \subset \RR^{t+1}$.
  This completes the proof.
\end{proof}

\section{Weight reduction of a quantum code}
\label{sec:weight reduction-quantum-code}

In this section, we present two approaches for reducing the weight of an arbitrary quantum code
  to at most $6$.
The first method applies to all codes, with particular emphasis on sparse codes.
The second method, specialized for dense codes,
  reduces the qubit blowup by a logarithmic factor.

\subsection{Weight reduction of a general quantum code}
\label{sec:general-code-reduce}
The overall strategy for reducing the weight of a quantum code proceeds as follows.
First, from the discussion in \Cref{sec:square-complex},
  we can associate the quantum code with a square complex.
Next, as described in \Cref{sec:weight reduction-simplicial-complex},
  we weight-reduce the cell complex
  so that each cell is incident to only a small number of other cells.
Finally, we place qubits and checks on the weight reduced cell complex
  to obtain the weight-reduced quantum code.
This is summarized in \Cref{fig:overview}.
\begin{figure}[H]
  \centering
  \begin{tikzpicture}
    \draw (0,0) node(TL) {chain complex};
    \draw (3,-1.2) node(TR) {cell complex};
    \draw (0,-4.5) node(BL)[align=center] {weight-reduced \\ chain complex};
    \draw (3,-3) node(BR)[align=center] {weight-reduced \\ cell complex};
    \draw[->] (TL) -- node[above right]{\Cref{sec:square-complex}} (TR);
    \draw[->] (TR) -- node[right]{geometric reduction} (BR);
    \draw[->] (BR) -- node[below right]{chain embedding} (BL);
    \draw[->] (TL) -- node[left] {goal} (BL);
  \end{tikzpicture}
  \caption{The overall strategy for reducing the weight of a quantum code.}
  \label{fig:overview}
\end{figure}

However, one technical challenge
  is that the ``geometric structure'' near a qubit must exhibit a product structure.
As a result, we cannot apply weight reduction for the cell complex in a purely black-box manner.
Instead, we must carefully account for the differences between the local structure around a qubit and that around a check.

The goal of geometric reduction is to obtain a weight-reduced cell complex $\cX$. We will break down the process into three steps: subdivision, sparsification, and cone construction. The three steps alone are sufficient to provide us with a weight-reduced cell complex as the skeleton of our code.
The goal of chain embedding is to place the qubits and checks on the weight-reduced cell complex $\cX$ appropriately to form the quantum code. We will describe the process in detail in~\Cref{sec:chain-embedding}.

The above process alone is sufficient to obtain a weight-reduced quantum code $\cC$ that is homotopy equivalent to the original code $C$, but the resulting code may have poor distance.
 To obtain a code with good distance, we need to perform additional operations called uniformization and expanderization when constructing the sparsified complex. In order to provide a better exposition of our construction, we will provide a vanilla version of the construction in~\Cref{sec:geometric-reduction} and~\Cref{sec:chain-embedding}, and describe the additional operations in~\Cref{sec:additional-op}.
\subsubsection{Geometric reduction}
\label{sec:geometric-reduction}
In \Cref{sec:square-complex}, we have arrived at a square complex
  induced from the original quantum code.
In this section, we want to reduce its weight through geometric operations
  similar to those described in \Cref{sec:sparse-simplex-reduce}.

  \paragraph*{Subdivision}
As shown in \Cref{fig:prepare}, we begin by partitioning the square complex into regions
  by subdividing each square into 4 smaller squares.
For each vertex $v \in V$,
  we define the local 2-complex $X^v$ as the closure of all subsquares incident to $v$.
For each edge $e = (v,w) \in E$,
  we define the local 1-complex $X^e$ as the intersection of $X^v$ and $X^w$.
For each face $f \in F$,
  we associate a local 0-complex $X^f$, located at the center of the square.
Moreover, for each vertex $v$,
  we define the boundary $\partial X^v$ as the union of $X^e$ for all edges $e$ incident to $v$.
This construction is analogous to the complexes $X'_x$ discussed in \Cref{sec:sparse-simplex-reduce}.

\begin{figure}[H]
  \centering
  \includegraphics[width=0.35\linewidth]{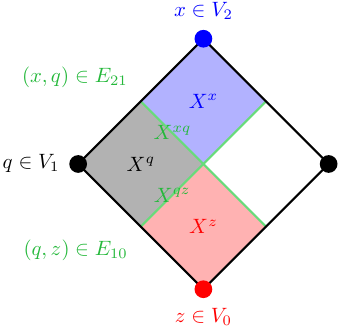}
    \caption{We divide the square complex $X$ into several local cell complexes,
            $X^x, X^q, X^z, X^{xq}, X^{qz}$.
            We only show the structure for a single square;
            in general, the construction involves many such squares.}
  \label{fig:prepare}
\end{figure}
The following simple claim from~\cite{liTransformArbitraryGood2024} provides a bound on the number of faces in a local complex $X^v$.
\begin{claim}[Implicit in Claim 5.2~\cite{liTransformArbitraryGood2024}]
  \label{clm:face-bound}
  For a quantum code $C$ with weight bound $w$, for every vertex $v$ in the associated square complex $X$,
    the number of faces in the local complex $X^v$ is at most $w^2$, which also provides a bound on the number of edges in the boundary $\partial X^v$.
\end{claim} 
We remark that in general, for a check vertex $v$, the boundary $\partial X^v$ may have multiedges, thus the bound is nontrivial.

\paragraph*{Sparsification}
We sparsify each local complex $\partial X^v$ individually. Observe that for every $v$, 
$\partial X^v$ is a bipartite graph, with two types of vertices: the vertices representing an edge have a star graph neighborhood with degree at most $w$, while the vertices representing a face will have degree 2. We will sparsify each edge vertex neighborhood, as shown in \Cref{fig:star-graph},
  each star graph $S_n$ can be sparsified
  into $\cS_n$ with $2n$ vertices,
  where every vertex has degree at most $3$. The resulting graph $\partial \cX^v$ will serve as the base of our new graph. We summarize its properties in the following claim.
  \begin{claim}
    \label{clm:sparse}
  For every boundary $\partial X^v=(V^v,E^v)$, the sparsified graph $\partial \cX^v$ will have at most $w|V^v|$ vertices and $2|E^v|$ edges, while keeping the degree of each vertex at most $3$.  
  \end{claim}

\paragraph*{Cone construction}
We will construct the cones separately for the check vertices and the qubit vertices, as they have different local structures.

We begin with constructing the qubit cones from the subcomplexes $X^q$,
  corresponding to the vertices associated with qubits $q \in V_1$.
Recall that from~\Cref{clm:unique},
  $X^q \cong S_{n_x} \times S_{n_z}$,
  where $S_n$ denotes a star graph with one internal node and $n$ leaves.
  Here, $n_x$ and $n_z$ are the numbers of adjacent X-checks and Z-checks, respectively.
This product structure can be sparsified by treating each factor independently.

We then define $\cX^q$ as the product $\cX^q = \cS_{n_x} \times \cS_{n_z}$, as shown in~\Cref{fig:qubit-cone}, and one can verify that its boundary matches $\partial \cX^q$ we obtained in sparsification, as shown in~\Cref{fig:qubit-cone}.
We have the following claim about the parameters of $\cX^q$.
\begin{claim}
  \label{clm:qubit-cone}
  $\cX^q$ satisfies:
\begin{itemize}
  \item Each vertex is incident to $\le 6$ edges.
  \item Each edge is incident to $\le 3$ faces.
  \item Each face is incident to $\le 4$ edges.
  \item $|\cX^q| = O(n_x n_z) = O(|X^q(2)|)=O(w^2)$.
\end{itemize}
(Recall that $|X^q(2)|$ is the number of faces in $X$ incident to the vertex $q$.)
\end{claim}

\begin{remark}
  One can also view the process as we first construct the cone $\cX^q$ and derive the sparsified boundary $\cX^{xq},\cX^{qz}$ from it. We then construct the cones $\cX^x$ and $\cX^z$ from their sparsified boundaries induced by $\cX^{xq},\cX^{qz}$. The current exposition treats the qubits and checks more symmetrically.
\end{remark}

What remains is to construct the cones $\cX^x$ and $\cX^z$,
  which correspond to the X-check and Z-check vertices in $X$, respectively.
While the overall sparsification strategy follows the approach described in \Cref{thm:weight reduction-2D-simplicial-complex},
  we make use of \Cref{lemma:cone-cell-complex} to obtain tighter bounds on the resulting check weights.
Given the sparsified boundary $\partial \cX^x$, we
 apply \Cref{lemma:cone-cell-complex}
  to construct the corresponding bounded-degree complex $\cX^x$.
An analogous process is used to construct $\cX^z$ for the Z-checks. We have the following claim about $\cX^x$ and $\cX^z$:
\begin{claim}
  \label{clm:check-cone}
  $\cX^x$ ($\cX^z$) satisfies:
\begin{itemize}
  \item Each vertex is incident to $\le 5$ edges.
  \item Each edge is incident to $\le 4$ faces.
  \item Each face is incident to $\le 5$ edges.
  \item $|\cX^x| = O(|X^x(2)| \log |X^x(2)|)=O(w^2\log w)$.
\end{itemize}
(Recall that $X^x(2)$ is the number of faces in $X$ incident to the vertex $x$.)
\end{claim}

The last property follows from the fact that the graph $\partial \cX^x$
  contains $O(|X^x(2)|)$ edges, and the bound $O(w^2\log w)$ follows from~\Cref{clm:face-bound} and~\Cref{clm:sparse}.

Finally,
  by attaching the subcomplexes $\cX^x$ and $\cX^z$ to $\bigcup_{q \in V_1} \cX^q$
  via their respective boundaries $\partial \cX^x$ and $\partial \cX^z$,
  we obtain the weight-reduced cell complex $\cX$.

\subsubsection{Chain embedding}
\label{sec:chain-embedding}
We now place the qubits and checks on the sparsified complexes
  to form the weight-reduced chain complex.
We employ the framework of quantum code embedding in \Cref{sec:quantum-code-embedding}.
In particular,
  the new chain complex $\cC$
  is constructed from a family of local chain complexes $\cC^x, \cC^q, \cC^z$,
  which satisfy
  \begin{itemize}
    \item $H_2(\cC^x) = \FF_2, H_1(\cC^x) = 0, H_0(\cC^x) = 0$.
    \item $H_2(\cC^q) = 0, H_1(\cC^q) = \FF_2, H_0(\cC^q) = 0$.
    \item $H_2(\cC^z) = 0, H_1(\cC^z) = 0, H_0(\cC^z) = \FF_2$.
  \end{itemize}
We further need to specify $g^{xq}$ and $g^{qz}$.
In our case, $p^{xz}$ is taken to be $0$.

We begin by describing the local chain complexes:
\begin{itemize}
  \item $\cC^x = C^T(\cX^x)$,
  \item $\cC^q = C^T(\cX^{xq}) \otimes C(\cX^{zq})$,
  \item $\cC^z = C(\cX^z)$,
\end{itemize}
where $C(\cX)$ is the cellular chain complex of $\cX$, and $C^T(\cX)$ is its dual.
More explicitly,
\begin{gather*}
  C^T(\cX^x): C_0(\cX^x) \xrightarrow{\partial_1^T} C_1(\cX^x) \xrightarrow{\partial_2^T} C_2(\cX^x), \\
  C^T(\cX^{xq}): C_0(\cX^{xq}) \xrightarrow{\partial_1^T} C_1(\cX^{xq}),
\end{gather*}
where the boundary operators are defined by the neighboring relations between cells in the respective complexes. Readers can refer to~\Cref{fig:qubit-layer} for an illustration of $\cC^q$.

We also define
\begin{itemize}
  \item $\cC^{xq} = C^T(\cX^{xq})$,
  \item $\cC^{qz} = C(\cX^{qz})$.
\end{itemize}

Next, we describe the maps $g^{xq}$ and $g^{qz}$.
Recall that $\cX^{xq}$ is a subcomplex of $\cX^x$ and $\cX^q$.
This induces a one-to-one correspondence between certain cells of $\cX^x$ and $\cX^q$.
We define $g^{xq}$ to be the linear map induced by the identity on these corresponding cells.
Because the 1-cells of $\cX^x$ correspond to $\cC^x_1$,
  while the 1-cells of $\cX^q$ contained in the subcomplex $\cX^{qx}$
    arise from the product of
    the 1-cells of $\cX^{xq}$ and the 0-cells of $\cX^{zq}$
  which corresponds to the product $\cC^{xq}_0 \otimes \cC^{zq}_0 = \cC^q_0$.
This implies the map $g^{xq}$ sends $\cC^x_1$ to $\cC^q_0$, as desired.
Similarly, because the 0-cells of $\cX^x$ correspond to $\cC^x_2$,
  while the 0-cells of $\cX^q$
    arise from the product of
    the 0-cells of $\cX^{xq}$ and the 0-cells of $\cX^{zq}$
  which corresponds to the product $\cC^{xq}_1 \otimes \cC^{zq}_0 \subset \cC^q_1$.
This implies the map $g^{xq}$ sends $\cC^x_2$ to $\cC^q_1$, as desired. An example of the map $g^{xq}$ is illustrated in~\Cref{fig:local-chain-complex-e}.

\begin{figure}[H]
  \centering
  \begin{tikzpicture}
    \node at (-2,-2) [circle, draw,blue] (z11){};
            \node at (0,-2) [circle, draw,blue] (z12){};
            \node at (2,-2) [circle, draw,blue] (z13) {};
            \node at (-2,0) [circle,fill,inner sep=1.5pt] (qz1){};
            \node at (0,0) [circle,fill,inner sep=1.5pt] (qz2){};
            \node at (2,0) [circle,fill,inner sep=1.5pt] (qz3){};
            \node at (-2.5,2) [circle, fill,inner sep=1.5pt] (x1){};
            \node at (-0.5,2) [circle, fill,inner sep=1.5pt] (x2){};
            \node at (1.5,2) [circle, fill, inner sep=1.5pt] (x3){};
            \node at (-2.5,0.4)  [circle,draw,blue] (q1){};
            \node at (-0.5,0.4)  [circle,draw,blue] (q2){};
            \node at (1.5,0.4)  [circle,draw,blue] (q3){};
           \node at (-1,0) [circle, draw,red] (xz1){};
           \node at (1,0) [circle, draw,red] (xz2){};
           \node at (-1,-2) [circle,fill,inner sep=1.5pt] (qc1){};
           \node at (1,-2) [circle,fill,inner sep=1.5pt] (qc2){};
           \draw (z11)--(qz1);
           \draw (qz1)--(x1);
           \draw (x1)--(q1);
           \draw (q1)--(z11);
           \draw (z12)--(qz2);
           \draw (qz2)--(x2);
           \draw (x2)--(q2);
           \draw (q2)--(z12);
           \draw (z13)--(qz3);
           \draw (qz3)--(x3);
           \draw (x3)--(q3);
           \draw (q3)--(z13);
           \draw (qz1)--(xz1);
           \draw (xz1)--(qz2);
           \draw (z11)--(qc1);
           \draw (qc1)--(z12);
           \draw (xz1)--(qc1);
           \draw (qz2)--(xz2);
           \draw (xz2)--(qz3);
           \draw (z12)--(qc2);
           \draw (qc2)--(z13);
           \draw (xz2)--(qc2);
           \node at (-2.25,-0.8) [circle,fill,inner sep=1.5pt] (e1){};
           \node at (-0.25,-0.8) [circle,fill,inner sep=1.5pt] (e2){};
           \node at (1.75,-0.8) [circle,fill,inner sep=1.5pt] (e3){};
           \node at (-2.25,1) [circle,draw,red,fill=white](e4){};
           \node at (-0.25,1) [circle,draw,red,fill=white](e5){};
            \node at (1.75,1) [circle,draw,red,fill=white](e6){};
            \draw (e1)--(e4);
            \draw (e2)--(e5);
            \draw (e3)--(e6);
           \fill[brown, nearly transparent] (-2,-2)--(-2,0)--(-2.5,2)--(-2.5,0.4)--cycle;
           \fill[brown, nearly transparent] (0,-2)--(0,0)--(-0.5,2)--(-0.5,0.4)--cycle;
           \fill[brown, nearly transparent] (2,-2)--(2,0)--(1.5,2)--(1.5,0.4)--cycle;
           \fill[brown, nearly transparent] (-2,-2)--(-2,0)--(2,0)--(2,-2)--cycle;
           \filldraw[black,nearly transparent] (-2.2,-0.2) rectangle (2.2,0.2);
           \fill[black,nearly transparent] (-2.2,-0.2)--(-1.8,-0.2)--(-2.3,2.2)--(-2.7,2.2)--cycle;
           \fill[black,nearly transparent] (-0.2,-0.2)--(0.2,-0.2)--(-0.3,2.2)--(-0.7,2.2)--cycle;
           \fill[black,nearly transparent] (1.8,-0.2)--(2.2,-0.2)--(1.7,2.2)--(1.3,2.2)--cycle;
            \filldraw[blue,nearly transparent] (-2.2,-2.2) rectangle (2.2,-1.8);
           \fill[blue,nearly transparent] (-2.2,-2.2)--(-1.8,-2.2)--(-2.3,0.6)--(-2.7,0.6)--cycle;
           \fill[blue,nearly transparent] (-0.2,-2.2)--(0.2,-2.2)--(-0.3,0.6)--(-0.7,0.6)--cycle;
           \fill[blue,nearly transparent] (1.8,-2.2)--(2.2,-2.2)--(1.7,0.6)--(1.3,0.6)--cycle;
        
  \end{tikzpicture}
  \caption{Illustration of the map $g^{xq}$, where the blue region corresponds to $\cC^x$, and the brown region corresponds to $\cC^q$. Note that the two regions are the two copies of $\cC^{xq}$ at different levels in the chain, and the map $g^{xq}$ acts as the identity on these two copies. }
  \label{fig:local-chain-complex-e}
\end{figure}

Similarly, $\cX^{qz}$ is a subcomplex of $\cX^q$ and $\cX^z$.
We define $g^{qz}$ to be the linear map induced by the identity on the corresponding cells.
This completes the definition of $\cC$. In \Cref{fig:local-chain-complex-v}, we illustrate
  the relative positioning of the local chain complexes $\cC^x, \cC^q, \cC^z$
  and how they fit together.

\begin{figure}[H]
  \centering
  \includegraphics[width=0.4\linewidth]{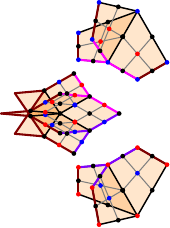}
  \caption{Illustration of three local chain complexes
            $\cC^x$, $\cC^q$, and $\cC^z$,
            arranged from top to bottom.
            The thicker lines indicate regions that will be attached to other pieces, which is shown in the shaded region in~\Cref{fig:local-chain-complex-e}.
            In particular, the pink regions are attached to each other,
              and likewise, the purple regions are attached to each other.
            To reduce visual clutter, not all qubits and checks of $\cC^q$ are shown.}
  \label{fig:local-chain-complex-v}
\end{figure}

To verify that $\cC$ forms a chain complex,
  we check that $g^{xq} \partial^x = \partial^q g^{xq}$
  and $g^{qz} \partial^q = \partial^z g^{qz}$.
These equalities hold because $g^{xq}$ acts as the identity on the corresponding cells of $\cX^{xq}$,
  and the boundary maps $\partial^x$ and $\partial^q$ act identically on those cells.
The same argument applies to $g^{qz}$. 
Combining that with~\Cref{clm:qubit-cone} and~\Cref{clm:check-cone}, we have the following theorem:

\begin{theorem}
  \label{thm:vanilla-construct}
  $\cC$ is a quantum code with check weight at most 5, qubit weight at most 6, and the numbers of qubits and checks are at most $O(w^2 \log w)$ times those of the original code $C$.
\end{theorem}

\begin{remark}
  In our vanilla construction, 
  instead of using a comb $\cS_n$, one may simply take a line,
    which corresponds to choosing $\cC^q$ as an ordinary surface code.
  This structure is closely related to the layer code.
  This choice of $\cC^q$ can be viewed as offsetting the structure
    by shrinking the teeth of the comb down to its backbone.

  The boundary of $\cX^x$ can be constructed in a similar manner
    by taking the union of lines.
  In particular, consider all lines in $\{\cC^q\}$ associated with $x$.
  For each square $(x, q, q', z)$,
    identify the vertex in $\cC^q$ that corresponds to $(x,z)$
    with the vertex in $\cC^{q'}$ that corresponds to $(x,z)$.
  Once we have $\partial \cX^x$,
    the coning operation can then be performed as before.
\end{remark}

\subsubsection{Uniformization and expanderization from dummy faces}
\label{sec:additional-op}
So far, we have obtained a weight-reduced code $\cC$.
However, the distance of the resulting code is not guaranteed to be $\Omega(dw)$.
To obtain the desired distance $\Omega(dw)$,
  we must introduce dummy faces into the complexes $\cX^x, \cX^q, \cX^z$.
This technique has appeared in several prior works.
In \cite{hastingsQuantumWeightReduction2023,liTransformArbitraryGood2024},
  it was used to handle `unreasonable codes'
  by ensuring that the boundary graph becomes connected.
In \cite{williamson2024low},
  it was used to achieve a larger distance,
  by making the boundary graph an expander,
  which is similar to our purpose here.

As we will see in the analysis of distance \Cref{sec:distance},
  we need to extend $\cX^x$ and $\cX^z$,
  so that their boundaries, $\partial \cX^x$ and $\partial \cX^z$,
  viewed as graphs,
  are expanders.
In addition, $\cX^q$ must be extended
  so that its associated local code $\cC^q$ has distance $\Theta(w)$.
  Currently, it is only $\min(n_x, n_z)$,
    where $n_x, n_z$ denote the numbers of X and Z checks incident on the qubit.

We first extend $\cX^q$, then $\cX^x$ and $\cX^z$.
To define the extension of $\cX^q$, we introduce the following structure named a relaxed comb $\cS^r_{w,n}$.
\begin{definition}
  A relaxed comb $\cS^r_{w,n}$ consists of a line with $w$ vertices, and $n$ additional vertices (called teeth) attached to the line as follows: if we label the vertices on the line from $0$ to $w-1$, we attach a tooth to the line at vertex $\lfloor i \cdot w/n \rfloor$ for each $i = 0, 1, \ldots, n-1$. An example can be found in~\Cref{fig:relaxed-comb}.
\end{definition}

\begin{figure}
  \centering
  \begin{tikzpicture}
    \filldraw (3,0) circle (2pt) ;
    \filldraw (2,0) circle (2pt) -- (2,1) circle (2pt);
        \filldraw (1,0) circle (2pt) ;
        \filldraw (0,0) circle (2pt) -- (0,1) circle (2pt);
        \filldraw (-1,0) circle (2pt) ;
        \filldraw (-2,0) circle (2pt) -- (-2,1) circle (2pt);

        \draw (3,0) -- (-2,0);
  \end{tikzpicture}
  \caption{An example of the relaxed comb $\cS^r_{6,3}$}
  \label{fig:relaxed-comb}
\end{figure}
Recall that $\cX^q$ was originally defined as the product $\cS_{n_x} \times \cS_{n_z}$.
We now replace this with $\cS^r_{4w,n_x} \times \cS^r_{4w,n_z}$,
  so that we have a cone with a grid interior of size $4w \times 4w$,
We call this procedure `uniformization'. Note that by uniformization, the boundary $\cX^{qx}$ and $\cX^{zq}$ are also extended to $\cS^r_{4w,n_x}$ and $\cS^r_{4w,n_z}$, respectively.

We now turn to extending $\cX^x$. 
With the modified $\cX^{qx}$, our goal is to construct the boundary $\partial \cX^x$ as an expander that also contains $\cX^{qx}$ as a subgraph.

To construct $\partial \cX^x$ as an expander, we will perform the following steps on $\partial X^x$, and an illustration can be found in~\Cref{fig:dummy-expander}. To simplify the exposition, in this section, we will ignore the face vertices in $\partial X^x$, as they can be counted by the number of edges.

  Recall that in our vanilla sparsification, we use $\cS_n$ to replace a degree $n$ vertex. Now we replace it with $\cS^r_{w,n}$, which has $w$ vertices. The $n$ teeth of $\cS^r_{w,n}$ are connected to the same neighbors as the original vertex, and the remaining vertices are dummy vertices, as shown in the second figure of \Cref{fig:dummy-expander}. This creates a graph with $|\partial X^x(0)|w$ vertices.

Next, we superimpose an $|\partial X^x(0)|$ partite, degree-3 expander $G$ with $w$ vertices for each part onto this graph. Each part of the vertices will map to the $w$ vertices on the spine of $\cS^r_{w,n}$,
  as shown in the third figure of \Cref{fig:dummy-expander}. Such an expander exists, and can be constructed through graph lifting~\cite{mohantyExplicitNearRamanujanGraphs2022}.
This temporarily increases the degree of some vertices beyond $3$ but still constant bounded, and the new graph $G'$ is also an expander.

To control the degree, we perform another step of sparsification, as we expand every vertex into 4 vertices,
  we redistribute the excess edges among the three dummy vertices
  associated with each original vertex,
  as shown in the final figure of \Cref{fig:dummy-expander}.
This operation can alternatively be viewed as replacing each vertex with a 4 vertex line, and by~\Cref{lem:expander-composition}, the resulting graph is still an expander.
\begin{lemma}[Expander composition; {\cite[Lemma 4.4]{sawicki2020superexpanders}}]
  \label{lem:expander-composition}
Let $G=(V_G,E_G)$ be a bounded degree graph with maximum degree $D$,
and $h(G)\ge\alpha>0$,
and let $H=(V_H,E_H)$ be a connected graph on $m$ vertices.
For each $v\in V_G$, choose a map $p_v:E_G(v)\to V_H$,
where $E_G(v)$ is the set of edges incident to $v$.
Define $G\circ_p H$ by taking disjoint copies $\{v\}\times H$
and, for each $e=\{u,v\}\in E_G$, adding the edge
\[
  \{(u,p_u(e)),(v,p_v(e))\}.
\]
Then $h(G\circ_p H)\ge c(\alpha,D,m)>0$, uniformly in $|V_G|$
and the maps $p_v$. In particular, replacing the vertices of a
bounded-degree expander family by copies of a fixed connected graph
preserves expansion.
Here $h(K)$ is the Cheeger constant of the graph $K$.
\end{lemma}
After redistribution, all vertices again have degree $\le 3$.
Moreover, we can identify $\cX^{qx}$ as a subgraph, as we construct the graph $\cS^r_{4w,n}$ in the process, and the connections are preserved.
\footnote{There is a minor difference between $\cS^r_{4w,n}$ and sparsifying every vertex of $\cS^r_{w,n}$ with a 4-vertex line, since $4\lfloor w/n\rfloor$ does not always equal $\lfloor 4w/n\rfloor$.  We will not dwell on the technicality here, as either construction will not affect the final result.}

Combining the above steps, we have the following theorem about the resulting $\partial \cX^x$ and $\partial \cX^z$:
\begin{theorem}
  \label{thm:expander-boundary}
  The resulting $\partial \cX^x$ is an expander with degree $\le 3$, and contains $\cX^{qx}$ as a subgraph, with $\Theta(|\partial X^x(0)|w)=O(w^2)$ vertices and edges. The same holds for $\partial \cX^z$ and $\cX^{qz}$.
\end{theorem}
Finally, we perform the cone construction as before to obtain $\cX^x$ of size $O(w^2\log w)$.
The same procedure applies to $\cX^z$.
 
\begin{figure}[H]
  \centering
  \includegraphics[width=0.99\linewidth]{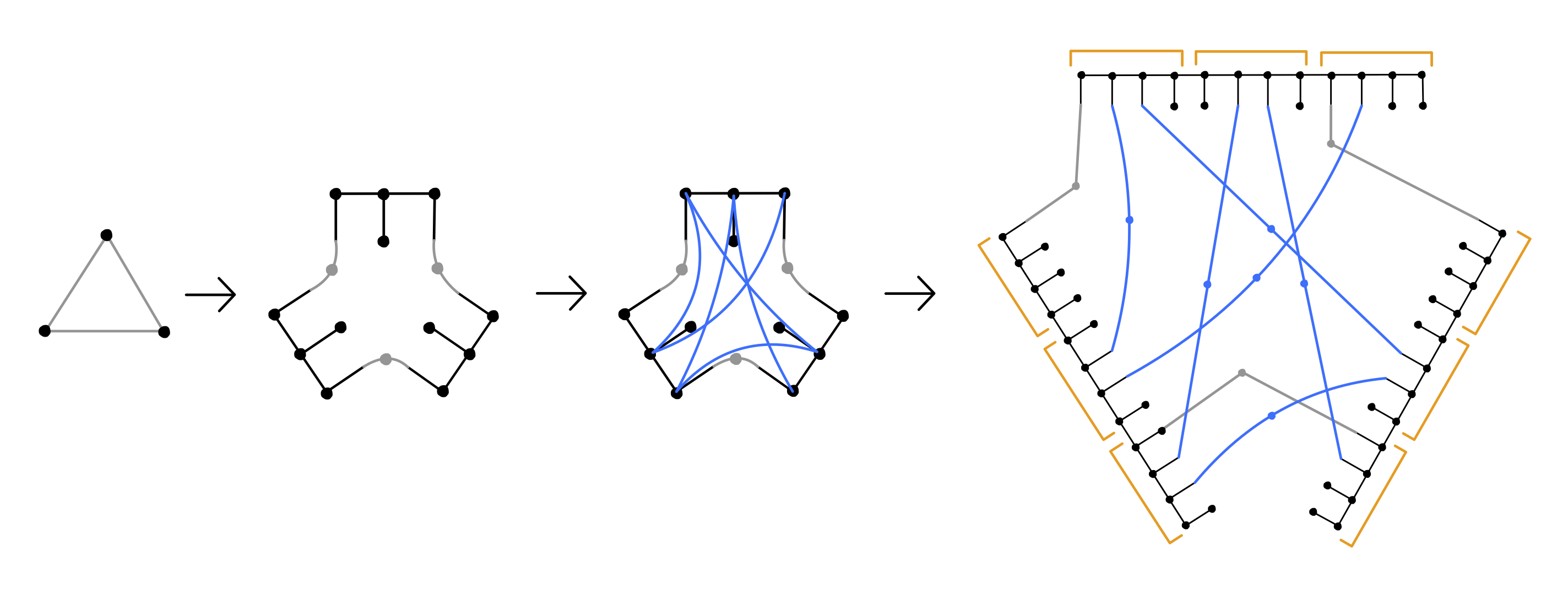}
  \caption{The progression of adding dummy vertices and edges to form an expander.
            The first figure shows the initial link structure around a vertex $x$.
            The second figure shows the structure after uniformization,
              where each vertex corresponding to $q$ is split into $w$ vertices.
            The third figure shows the addition of edges (shown in blue) to form an expander.
            The fourth figure shows the redistribution of edges
              to ensure that all vertices have degree $\le 3$. The vertices in the comb with no outgoing edges will be ignored in our final construction.
            Note that in the final figure, we can identify $\cX^{qx}$ as the subgraph, consisting of the vertices and edges in black and gray.
          }
  \label{fig:dummy-expander}
\end{figure}

Now that we have obtained the desired structure of $\cX^x, \cX^q, \cX^z$,
  we can perform the same chain embedding as before to obtain the final code $\cC$. The construction of $\cC^x$, $\cC^q$, $\cC^z$ remains the same as that induced by the cellular chain complexes of $\cX^x, \cX^q, \cX^z$, while the maps $g^{xq}$ and $g^{qz}$ are also defined by the identity on the corresponding cells of $\cX^{xq}$ and $\cX^{qz}$. There is a minor difference in verifying that $g^{xq} \partial^x = \partial^q g^{xq}$ and $g^{qz} \partial^q = \partial^z g^{qz}$, as $\partial^x$ and $\partial^q$ may act differently due to the introduction of dummy faces, and similarly for $\partial^q$ and $\partial^z$. However, one can observe that for the image of $g^{xq}$, the boundary operator only acts on cells in $\cX^{xq}$, and it acts identically to $\partial^x$ restricted to cells in $\cX^{xq}$. Similar verification applies to $g^{qz}$, so the equalities still hold. 

  \begin{figure}[H]
  \centering
  \begin{tikzpicture}
    \node at (-2,-2) [circle, draw,blue] (z11){};
            \node at (0,-2) [circle, draw,blue] (z12){};
            \node at (2,-2) [circle, draw,blue] (z13) {};
            \node at (-2,0) [circle,fill,inner sep=1.5pt] (qz1){};
            \node at (0,0) [circle,fill,inner sep=1.5pt] (qz2){};
            \node at (2,0) [circle,fill,inner sep=1.5pt] (qz3){};
            \node at (-2.5,2) [circle, fill,inner sep=1.5pt] (x1){};
            \node at (1.5,2) [circle, fill, inner sep=1.5pt] (x3){};
            \node at (-2.5,0.4)  [circle,draw,blue] (q1){};
            \node at (-0.5,0.4)  [circle,draw,blue] (q2){};
            \node at (1.5,0.4)  [circle,draw,blue] (q3){};
           \node at (-1,0) [circle, draw,red] (xz1){};
           \node at (1,0) [circle, draw,red] (xz2){};
           \node at (-1,-2) [circle,fill,inner sep=1.5pt] (qc1){};
           \node at (1,-2) [circle,fill,inner sep=1.5pt] (qc2){};
           \draw (z11)--(qz1);
           \draw (qz1)--(x1);
           \draw (x1)--(q1);
           \draw (q1)--(z11);
           \draw (z12)--(qz2);
           \draw (q2)--(z12);
           \draw (z13)--(qz3);
           \draw (qz3)--(x3);
           \draw (x3)--(q3);
           \draw (q3)--(z13);
           \draw (qz1)--(xz1);
           \draw (xz1)--(qz2);
           \draw (z11)--(qc1);
           \draw (qc1)--(z12);
           \draw (xz1)--(qc1);
           \draw (qz2)--(xz2);
           \draw (xz2)--(qz3);
           \draw (z12)--(qc2);
           \draw (qc2)--(z13);
           \draw (xz2)--(qc2);
           \node at (-2.25,-0.8) [circle,fill,inner sep=1.5pt] (e1){};
           \node at (-0.25,-0.8) [circle,fill,inner sep=1.5pt] (e2){};
           \node at (1.75,-0.8) [circle,fill,inner sep=1.5pt] (e3){};
           \node at (-2.25,1) [circle,draw,red,fill=white](e4){};
            \node at (1.75,1) [circle,draw,red,fill=white](e6){};
            \draw (e1)--(e4);
            \draw (e3)--(e6);
           \fill[brown, nearly transparent] (-2,-2)--(-2,0)--(-2.5,2)--(-2.5,0.4)--cycle;
           \fill[brown, nearly transparent] (2,-2)--(2,0)--(1.5,2)--(1.5,0.4)--cycle;
           \fill[brown, nearly transparent] (-2,-2)--(-2,0)--(2,0)--(2,-2)--cycle;
           \filldraw[black,nearly transparent] (-2.2,-0.2) rectangle (2.2,0.2);
           \fill[black,nearly transparent] (-2.2,-0.2)--(-1.8,-0.2)--(-2.3,2.2)--(-2.7,2.2)--cycle;
           \fill[black,nearly transparent] (1.8,-0.2)--(2.2,-0.2)--(1.7,2.2)--(1.3,2.2)--cycle;
            \filldraw[blue,nearly transparent] (-2.2,-2.2) rectangle (2.2,-1.8);
           \fill[blue,nearly transparent] (-2.2,-2.2)--(-1.8,-2.2)--(-2.3,0.6)--(-2.7,0.6)--cycle;
           \fill[blue,nearly transparent] (-0.2,-2.2)--(0.2,-2.2)--(-0.3,0.6)--(-0.7,0.6)--cycle;
           \fill[blue,nearly transparent] (1.8,-2.2)--(2.2,-2.2)--(1.7,0.6)--(1.3,0.6)--cycle;
        
  \end{tikzpicture}
  \caption{Illustration of the map $g^{xq}$, when $\cX^x$ includes dummy faces, as shown in the middle. The blue region corresponds to $\cC^x$, and the brown region corresponds to $\cC^q$. Note that the connected regions are two copies of $\cC^{xq}$ on different levels, and the map $g^{xq}$ acts as the identity on these two copies. }
  \label{fig:local-chain-complex-e-dummy}
\end{figure}

Combining the above steps, we have the following theorem about the resulting code $\cC$:

\begin{theorem}
  \label{thm:final-construct}
  $\cC$ is a quantum code with check weight at most 5, qubit weight at most 6, and the numbers of qubits and checks are at most $O(w^2 \log w)$ times those of the original code $C$. Moreover, each qubit cone contains a surface code of distance $\Theta(w)$, and the boundary of each X or Z cone is an expander with degree $\le 3$.
\end{theorem}

\subsection{Weight reduction of a dense quantum code using layer code construction}\label{sec:dense-code-reduce}

For a dense code, the weight $w$ is of order $\Theta(n)$.
This means given a code of size $n$,
  the resulting code has size $O(n w^2 \log w) = O(n^3 \log n)$.
In fact, we can do better with size $O(n^3)$.
This follows from the layer code construction \cite{williamsonLayerCodes2023},
  which automatically sparsifies the code.

For quantum codes,
  the 2D simplicial complex is naturally 3-partite, $V_0, V_1, V_2$.
Following the construction in \Cref{thm:weight reduction-simplicial-complex-dense-new}
  (i.e.\ the layer code),
  we place the X checks, qubits and Z checks on $V_0,V_1,V_2$ respectively.
 The edges are naturally induced by the Tanner graph of the quantum code. By our construction in~\Cref{thm:weight reduction-simplicial-complex-dense-new}, after applying the weight reduction process in the theorem, we are able to obtain a new complex $\cX$ that is homotopy equivalent to the original complex $X$. Now we place the checks and qubits as shown in~\Cref{fig:plane-grid}. Taking the X check as an example, on the $V_0$ planes in $\cX$, we form a grid as follows: for the points with integer coordinates, we place an X check; we place qubits in the middle of the grid lines that connect two X checks; finally we use Z checks to fill in the faces in the grid.

\begin{figure}[H]
    \centering
    \begin{subfigure}[b]{0.3\textwidth}
       \resizebox{!}{\linewidth}
       {
       \begin{tikzpicture}
           \foreach \row in {0,1,2}
        {
            \foreach \column in {0,1,2}
            {
                \pgfmathtruncatemacro{\rw}{2*\row}
                \pgfmathtruncatemacro{\cl}{2*\column}
                \node at (2*\row,2*\column)[circle, draw,blue] (v\rw\cl){};
                \pgfmathtruncatemacro{\cl}{2*\column+1}
                \node at (2*\row,2*\column+1)[circle,fill,inner sep=1.5pt] (v\rw\cl) {};
                \pgfmathtruncatemacro{\rw}{2*\row+1}
                \node at (2*\row+1,2*\column+1)[circle,draw,red] (v\rw\cl){};
                \pgfmathtruncatemacro{\cl}{2*\column}
                \node at (2*\row+1,2*\column) [circle,fill,inner sep=1.5pt](v\rw\cl) {};
            };
        };

        \foreach \row in {0,1,...,4}
        {
            \foreach \column in {0,1,...,3}
            {
                \pgfmathtruncatemacro{\cl}{\column+1}
                \draw (v\row\column)--(v\row\cl);
            }
        }

        \foreach \column in {0,1,...,4}
        {
            \foreach \row in {0,1,...,3}
            {
                \pgfmathtruncatemacro{\rw}{\row+1}
                \draw (v\row\column)--(v\rw\column);
            }
        }
        \filldraw [white] (-0.5,4.5) rectangle (5.5,5.5);
        \filldraw [white] (4.5,-0.5) rectangle (5.5,5.5);
       \end{tikzpicture}
       }
    \end{subfigure}
    \begin{subfigure}[b]{0.3\textwidth}
       \resizebox{!}{\linewidth}
       {
       \begin{tikzpicture}
           \foreach \row in {0,1,2}
        {
            \foreach \column in {0,1,2}
            {
                \pgfmathtruncatemacro{\rw}{2*\row}
                \pgfmathtruncatemacro{\cl}{2*\column}
                \node at (2*\row,2*\column)[circle,fill,inner sep=1.5pt] (v\rw\cl){};
                \pgfmathtruncatemacro{\cl}{2*\column+1}
                \node at (2*\row,2*\column+1)[circle, draw,blue] (v\rw\cl) {};
                \pgfmathtruncatemacro{\rw}{2*\row+1}
                \node at (2*\row+1,2*\column+1)[circle,fill,inner sep=1.5pt] (v\rw\cl){};
                \pgfmathtruncatemacro{\cl}{2*\column}
                \node at (2*\row+1,2*\column) [circle, draw,red](v\rw\cl) {};
            };
        };

        \foreach \row in {0,1,...,4}
        {
            \foreach \column in {0,1,...,3}
            {
                \pgfmathtruncatemacro{\cl}{\column+1}
                \draw (v\row\column)--(v\row\cl);
            }
        }

        \foreach \column in {0,1,...,4}
        {
            \foreach \row in {0,1,...,3}
            {
                \pgfmathtruncatemacro{\rw}{\row+1}
                \draw (v\row\column)--(v\rw\column);
            }
        }

        \filldraw [white] (-0.5,4.5) rectangle (5.5,5.5);
        \filldraw [white] (4.5,-0.5) rectangle (5.5,5.5);
       \end{tikzpicture}
       }
    \end{subfigure}
    \begin{subfigure}[b]{0.3\textwidth}
       \resizebox{!}{\linewidth}
       {
       \begin{tikzpicture}
           \foreach \row in {0,1,2}
        {
            \foreach \column in {0,1,2}
            {
                \pgfmathtruncatemacro{\rw}{2*\row}
                \pgfmathtruncatemacro{\cl}{2*\column}
                \node at (2*\row,2*\column)[circle, draw,red] (v\rw\cl){};
                \pgfmathtruncatemacro{\cl}{2*\column+1}
                \node at (2*\row,2*\column+1)[circle,fill,inner sep=1.5pt] (v\rw\cl) {};
                \pgfmathtruncatemacro{\rw}{2*\row+1}
                \node at (2*\row+1,2*\column+1)[circle,draw,blue] (v\rw\cl){};
                \pgfmathtruncatemacro{\cl}{2*\column}
                \node at (2*\row+1,2*\column) [circle,fill,inner sep=1.5pt](v\rw\cl) {};
            };
        };

        \foreach \row in {0,1,...,4}
        {
            \foreach \column in {0,1,...,3}
            {
                \pgfmathtruncatemacro{\cl}{\column+1}
                \draw (v\row\column)--(v\row\cl);
            }
        }

        \foreach \column in {0,1,...,4}
        {
            \foreach \row in {0,1,...,3}
            {
                \pgfmathtruncatemacro{\rw}{\row+1}
                \draw (v\row\column)--(v\rw\column);
            }
        }

        \filldraw [white] (-0.5,4.5) rectangle (5.5,5.5);
        \filldraw [white] (4.5,-0.5) rectangle (5.5,5.5);
       \end{tikzpicture}
       }
    \end{subfigure}
    \caption{From left to right, the graph describes the grid structure we place on the planes corresponding to $V_0,V_1,V_2$ respectively.}
    \label{fig:plane-grid}
\end{figure}

 Now we connect the different layers of our construction as follows: if there is an edge $e=(v_0,v_1)$ between $v_0\in V_0$ and $v_1\in V_1$, we will connect the corresponding two rows in the layer $v_0$ and $v_1$. Similarly, we will connect the layers corresponding to $V_1$ and $V_2$. By our construction, the structure $v_0-v_1-v_2$ uniquely determines a point in $\RR^3$.
 In order to construct a valid quantum code, we need to guarantee that the new code we constructed respects the commutation relation between X and Z stabilizers, while our current construction at $v_0$ and $v_2$ will only share one qubit.
 To fix the issue, we follow the construction of layer codes~\cite{williamsonLayerCodes2023}.
 For every pair of neighboring X and Z stabilizers $(v_0,v_2)$,
 they will share an even number of qubits.
 After creating the layer structure, we start from the left and pair the neighboring qubits together.
 As the plane $v_0$ and $v_2$ will intersect along a line, these pairings can be viewed as line segments on the line, as shown in green in the figures of layer codes~\cite{williamsonLayerCodes2023}, which they call $y$ defect lines.
The explicit construction of the gadget along the pairing segments is shown in~\Cref{fig:defect-struct}. Note that in the construction of the gadget, we've introduced additional structures on $\cX$, but it is not hard to verify that the new gadget is homotopy equivalent to the original structure.

\begin{figure}
    \centering
    \begin{subfigure}[b]{0.4\textwidth}

        \begin{tikzpicture}
            \node at (-2,-2) [circle, draw,blue] (z11){};
            \node at (0,-2) [circle, draw,blue] (z12){};
            \node at (2,-2) [circle, draw,blue] (z13) {};
            \node at (2,0) [circle,fill,inner sep=1.5pt] (qz3){};
            \node at (-2.5,2) [circle, draw,red] (x1){};
            \node at (-0.5,2) [circle, draw,red] (x2){};
            \node at (1.5,2) [circle, draw,red] (x3){};
            \node at (-2.5,0)  [circle,fill,inner sep=1.5pt] (q1){};
           \node at (-1,-2) [circle,fill,inner sep=1.5pt] (qc1){};
           \node at (1,-2) [circle,fill,inner sep=1.5pt] (qc2){};
           \node at (-1.5,2) [circle,fill,inner sep=1.5pt] (qx1){};
           \node at (0.5,2) [circle,fill,inner sep=1.5pt] (qx2){};
          \draw (x1)--(qc1);
           \draw (x1)--(q1);
           \draw (q1)--(z11);
           \draw (qx1)--(z12);
           \draw (x2)--(qc2);
           \draw (qx2)--(z13);
           \draw (z13)--(qz3);
           \draw (qz3)--(x3);
           \draw (z11)--(qc1);
           \draw (qc1)--(z12);
           \draw (z12)--(qc2);
           \draw (qc2)--(z13);
           \draw (x1)--(qx1);
           \draw (x2)--(qx1);
           \draw (x2)--(qx2);
           \draw (x3)--(qx2);
        \end{tikzpicture}

    \end{subfigure}
    \caption{The explicit structure of the ``$y$ defect lines'' colored in green in~\cite{williamsonLayerCodes2023}. We will add the gadget on the pairing lines along the intersection of X layers and Z layers, as the top line is from the Z layer and the bottom line is from the X layer.}
    \label{fig:defect-struct}
\end{figure}

\begin{remark}
  If one attempts to apply this strategy to a sparse CSS code with weight $w$,
    the blowup depends on the number of layers we need.
  If one revisits the construction,
    the place where we can save the blowup is that
  For every X-check and Z-check,
    the incident qubits should be on a different layer.
  We can use a graph to indicate this relation:
    the vertices are the qubits,
    and the two vertices have an edge if they share an X-check or a Z-check.
  We want to color the vertices so that no two adjacent vertices
    have the same color.
  Since the graph has degree $O(w^2)$,
    we need $O(w^2)$ colors to label the qubits,
    i.e.\ the number of layers for qubits is $O(w^2)$.

  What about for the checks?
  Given an X-check,
    all the Z-checks that share a qubit with the X-check
    should be on a different layer.
  Since there are $O(w^2)$ many such Z-checks,
    we need $O(w^2)$ layers for the X-checks.
  Similarly, we need $O(w^2)$ layers for the Z-checks.

  That means overall, the blowup is $w^6$.
  This is larger than the optimal $w^2 \log w$.

  One can apply the same trick in \cite[IIB Choosing Heights]{hastingsQuantumWeightReduction2023}
    to reduce the number of qubit layers to $O(w^{1+\epsilon})$ for any $\epsilon > 0$.
\end{remark}

\section{Code parameters}
\label{sec:code-parameters}

We now analyze the parameters of the weight-reduced quantum code $\cC$,
  thereby establishing the main result stated in~\Cref{thm:main}.
\begin{theorem}[Restatement of~\Cref{thm:main}]\label{thm:main-re}
  Given a $[[n,k,d]]$ quantum code with weight $\leq w$,
  there exists a weight reduction procedure that generates
  a $[[O(nw^2\log w),k,\Omega(dw)]]$ quantum code
  with check weight $\le 5$ and qubit weight $\le 6$.
\end{theorem}

\subsection{Check weight and qubit weight}
\label{sec:weight}

We claim that each qubit in $\cC$ is incident to at most $6$ checks
  and each check is incident to at most $5$ qubits.
This can be checked by examining the local complexes $\cC^x, \cC^q, \cC^z$,
  and the attaching regions $\cC^{xq}, \cC^{qz}$.
We use the term node to refer to either a qubit or a check.
We call a node of $\cC^x$ internal
  if it is not located on the boundary of $\cX^x$.

Each internal node in $\cC^x$ and $\cC^z$
  is incident to at most $5$ nodes.
In particular, in $\cC^x$,
  the X-checks (vertices) are incident to $\le 5$ qubits (edges).
The Z-checks (faces) are incident to $\le 5$ qubits (edges).
The qubits (edges) are incident to $2$ X-checks (vertices) and $\le 4$ Z-checks (faces).
The same holds for $\cC^z$.

For each internal node in $\cC^q$,
  it is incident to at most $6$ nodes.
In particular, in $\cC^q$,
  the checks (edges) are incident to $\le 2 + 3$ qubits
    originating from vertices and faces.
  The face qubits are incident to $2$ X-checks and $2$ Z-checks.
  The vertex qubits are incident to
    $\le 3$ X-checks and $\le 3$ Z-checks.

The qubits and checks in the attaching regions $\cC^{xq}$ and $\cC^{qz}$ can be analyzed analogously.

\subsection{Code dimension}
\label{sec:dimension}

The goal is to show that $C$ is the effective chain complex of $\cC$
  as defined in \Cref{sec:quantum-code-embedding}.
As discussed in that section,
  this relationship implies that $C$ and $\cC$ are chain homotopy equivalent
  and therefore share the same homology groups. 
  From~\Cref{lem:hom-equiv}, we know that we can construct the effective chain complex $C^h$ from $\cC$, and we will show the following theorem.
\begin{theorem}
  The chain complex $C$ is isomorphic to the effective chain complex $C^h$ of $\cC$. 
\end{theorem}
At a fundamental level,
  the reason $C$ serves as the effective chain complex of $\cC$
  is that $X$ can be obtained from $\cX$ by collapsing each of its local complexes.

\begin{proof}
To establish this relationship,
  we need to verify the homology of $\cC_x, \cC_q, \cC_z$,
  as well as the maps $[g_1], [g_2]$.
We first compute the homology groups.
\begin{itemize}
  \item Since $\cX^x$ is contractible,
        $C(\cX^x) \simeq 0 \to 0 \to \FF_2$,
        which implies $\cC^x := C^T(\cX^x) \simeq \FF_2 \to 0 \to 0$.
  \item Since $\cX^{xq}$ and $\cX^{qz}$ are both contractible,
        $\cC^{xq} := C^T(\cX^{xq}) \simeq \FF_2 \to 0$ and $\cC^{qz} := C(\cX^{qz}) \simeq 0 \to \FF_2$.
        Therefore, $\cC^q := C^T(\cX^{xq}) \otimes C(\cX^{qz}) \simeq (\FF_2 \to 0) \otimes (0 \to \FF_2) = 0 \to \FF_2 \to 0$.
  \item Since $\cX^z$ is contractible,
        $\cC^z := C(\cX^z) \simeq 0 \to 0 \to \FF_2$.
\end{itemize}
The construction naturally induces an isomorphism between the chains of $C$ and the homology groups of $\cC^x, \cC^q, \cC^z$,
by mapping each cone to its corresponding chain.

Next, we compute the maps $[g^{xq}], [g^{qz}]$.
The representative of the $\FF_2$ in $\cC^x$
  is the 2-chain that includes all the faces of $\cX^x$.
If $x$ and $q$ are not incident in $C$,
  the complexes $\cX^x$ and $\cX^q$ are disjoint,
  so $[g^{xq}] = 0$.
If $x$ and $q$ are incident in $C$,
  the map $g^{xq}$ sends the 2-chain of $\cX^x$
  to the 1-chain on the shared boundary $\cX^{xq} \subset \cX^q$.
This 1-chain is precisely the representative of the $\FF_2$ in $\cC^q$.
Therefore, the induced map $[g^{xq}]$ reproduces the incidence relation between the X checks and qubits in $C$.
The same argument applies to $[g^{qz}]$.

Therefore, $C$ is the effective chain complex of $\cC$,
  which implies $H_1(\cC) \cong H_1(C)$;
  in other words, $\cC$ has the same code dimension as $C$.
\end{proof}
\subsection{Code distance}
\label{sec:distance}

The strategy for showing distance follows the standard approach,
  which uses variants of the cleaning lemma.
Similar arguments can be found in several prior works,
  including
  \cite[Sec. III A]{hastingsQuantumWeightReduction2023},
  \cite[Lemma 2]{williamson2024low},
  and \cite[Proposition 5.3]{yuan2025unified}.

Intuitively,
  the factor $w$ increase in distance arises
  because each qubit in the original code
  is replaced by a code of distance $\Theta(w)$.
This mechanism is analogous to the factor from thickening (and choosing heights)
  in \cite{hastingsQuantumWeightReduction2023}.
The role of the expander is to ensure that this increase in distance is preserved
  and does not degrade during the gluing of $\cC^x$ and $\cC^z$.

For each $x\in V_2$, let
\[
  g_2^{xQ}\colon\cC_2^x\longrightarrow\cC_1^Q,
  \qquad g_2^{xQ}c_2^x=(g_2^{xq}c_2^x)_{q\in V_1}.
\]

\begin{lemma}[{\cite[Lemma I.3]{yuan2025unified}}]
  \label{lem:cleaning}
  Suppose that there exists $0<h\le1$ such that
    for any $x\in V_2$ and any $c_2^x\in\cC_2^x$,
    there exists $\hat c_2^x\in\cC_2^x$
    such that $\partial_2^x c_2^x=\partial_2^x\hat c_2^x$ and
  \begin{equation}
    \label{eq:isoperimetric}
    |\partial_2^x\hat c_2^x|\ge h\,|g_2^{xQ}\hat c_2^x|.
  \end{equation}
  Then
  \begin{equation}
    \label{eq:distance}
    \min_{c_1\in\ker\partial_1:
      0\ne[c_1]\in H_1(\cC)}|c_1|
    \ge h\min_{c_1^Q\in\ker\partial_1^Q:
      0\ne[[c_1^Q]]\in H_1(C)}|c_1^Q|.
  \end{equation}
  Here $[[c_1^Q]]$ is the class of the local homology vector
  $[c_1^Q]\in H_1(\cC^Q)$ in the effective complex $C$;
  it is defined when $[g^{QZ}][c_1^Q]=0$.
\end{lemma}

The idea of the lemma is that the assumption \eqref{eq:isoperimetric}
  allows us to clean up the support of any X-logical operator on $\cC^x$.
After cleaning,
  the remaining support lies on $\cC^q$ and $\cC^z$.
Since $H_1(\cC^z) = 0$,
  the essential part of the logical operator is its support on $\cC^q$.
This is why the right-hand side of \eqref{eq:distance} focuses only on $\cC^Q$.

\begin{lemma}
  \label{lem:dist}
  Suppose every boundary $\partial\cX^x$, viewed as a graph,
    has Cheeger constant at least $h>0$, i.e.\ for every subset
    $V\subseteq V(\partial\cX^x)$,
  \begin{equation}
    \label{eq:edge-expansion}
    |\partial V|\ge h\min\{|V|,|V(\partial\cX^x)\setminus V|\}.
  \end{equation}
  Suppose also that every local code $\cC^q$ has X-distance
    at least $d_x^{\mathrm{loc}}$.
  Then
  \[
    d_x(\cC)\geq \frac{h}{2}d_x^{\mathrm{loc}}d_x(C).
  \]
\end{lemma}

\begin{proof}
  To apply \Cref{lem:cleaning}, we verify \eqref{eq:isoperimetric}.
  Given $c_2^x$, choose $\hat c_2^x$ as $c_2^x$ or
 its complement, to mininmize its boundary support $S$.
  Then $\partial_2^x\hat c_2^x=\partial_2^x c_2^x$, and the
  restriction of $\partial_2^x\hat c_2^x$ to the boundary edges
  is the edge cut of $S$. Hence
  \[
    |\partial_2^x\hat c_2^x|\ge|\partial S|\ge h|S|.
  \]
  Each boundary vertex belongs to at most two attaching combs:
  a backbone vertex belongs to one, and a shared square vertex
  belongs to the two combs corresponding to its paired qubits.
  Since the attaching maps are restrictions to these combs,
  \[
    |g_2^{xQ}\hat c_2^x|
    =\sum_{q\in V_1}|g_2^{xq}\hat c_2^x|
    \le2|S|.
  \]
  Thus \eqref{eq:isoperimetric} holds with coefficient
    $h/2$.

  Since every local code $\cC^q$ has X-distance
    at least $d_x^{\mathrm{loc}}$, we obtain
  \[
    \min_{c_1^Q\in\ker\partial_1^Q:
      0\ne[[c_1^Q]]\in H_1(C)}|c_1^Q|
    \ge d_x^{\mathrm{loc}}d_x(C).
  \]
  Indeed, its local homology vector is a nontrivial original
  logical and therefore has at least $d_x(C)$ nonzero components;
  each such component has weight at least $d_x^{\mathrm{loc}}$.
  Combining the two inequalities proves the result.
\end{proof}

From~\Cref{thm:expander-boundary}, our construction satisfies the assumption. By applying~\Cref{lem:dist}, since the local code $\cC^q$ has distance $\Theta(w)$, we immediately obtain that the distance of $\cC$ is $\Omega(dw)$.

\section{Application: breaking square root barrier from the weight reduction of a dense quantum code}
\label{sec:application}

In this section, we introduce potential applications of our weight reduction routine to random quantum codes.
These constructions can directly provide quantum codes that break the $O(\sqrt{n})$ distance barrier.
In Hastings's paper on quantum weight reduction~\cite{hastingsQuantumWeightReduction2023}, he discussed a similar application of his routine on random quantum codes that have sparse X stabilizers.
In~\Cref{sec:hastings-reduce}, we apply our routine to the same type of code in~\cite{hastingsQuantumWeightReduction2023}, and discuss the improvements of our construction.
In~\Cref{sec:dense-reduce}, we discuss how to construct quantum LDPC codes by weight reducing random quantum codes that have dense stabilizers.

Besides breaking the $O(\sqrt{n})$ barrier, our second construction is additionally geometrically local in 3D,
and matches the upper bounds of dimension and distance of geometrically local codes~\cite{bravyiNogoTheoremTwodimensional2009,bravyiTradeoffsReliableQuantum2010}.
Breaking down our construction for the dense codes, our results indicate that the layer code construction~\cite{williamsonLayerCodes2023} does not need to start from a quantum LDPC code with linear dimension and distance, as a random quantum code already suffices.

\subsection{Weight reducing the code from Hastings}\label{sec:hastings-reduce}

In Hastings' paper~\cite{hastingsQuantumWeightReduction2023}, he provided a construction of LDPC codes that breaks the distance square-root barrier through weight reduction. Following his routine, we start from the following $n$-qubit code through probabilistic construction: We first pick $n/2$ random X-stabilizers of average weight $\Delta=\beta\log(n)$, for some constant $\beta$. (In fact, $\Delta = \Theta(1)$ is enough. See \cite{hastingsQuantumWeightReduction2023} for more discussion.) Among the codewords that commute with the X stabilizers, we choose $n/4$ of them to be Z-stabilizers. Hastings proved the following properties of the quantum code $Q$:

\begin{lemma}
    For sufficiently large $\beta$, the code $Q$ has dimension $k=\Theta(n)$ and distance $d=\Theta(n)$. The code has weight $w_X,q_X=O(\log n),w_Z,q_Z=\Theta(n)$.
\end{lemma}

Starting from such codes,
  the weight reduction process in his paper
  induces a quantum code with the following parameters.

\begin{theorem}
  There is an $N=\Tilde{\Theta}(n^{2})$ qubit quantum code with dimension $\Theta(n)$ and distance $d_X=\tilde{\Omega}(n^{2}),d_Z=\tilde{\Omega}(n)$. By the distance balancing trick in~\cite{hastingsWeightReductionQuantum2017}, we can obtain an LDPC code with $N=\Tilde{\Theta}(n^3)$, $K=\Theta(n)$, and distance $d=\tilde{\Omega}(n^2)$.
\end{theorem}

\begin{remark}
In the original paper by Hastings, he also showed how to construct an LDPC code family with dimension $k=\Tilde{\Theta}(N^{2/3})$ and distance $d=\Tilde{\Theta}(N^{2/3})$ by applying the distance balancing technique in~\cite{evraDecodableQuantumLDPC2024}. We can also combine the technique with our weight reduction scheme, and obtain similar parameters as Hastings's result.
\end{remark}

It is clear that our main theorem reproduces the result.

\subsection{Weight reducing random dense CSS code}\label{sec:dense-reduce}

Our method can also be directly applied to random dense CSS codes.

\begin{theorem}
  There exists a family of quantum LDPC codes on $\Theta(n^3 \log n)$ qubits
  with dimension $k = \Theta(n)$, distance $d = \Omega(n^2)$,
    qubit weight $\le 6$, and check weight $\le 5$.
\end{theorem}
\begin{proof}
  Sample random dense parity-check matrices $H_x, H_z$
    with $n$ qubits, $0.2n$ X checks, $0.2n$ Z checks,
    satisfying $H_x H_z^T = 0$,
    the resulting CSS code has parameters $[[n, \Theta(n), \Theta(n)]]$ with high probability.
  Applying our weight reduction procedure to this code
    yields a new quantum code with the desired parameters.
\end{proof}

\subsection{Weight reducing random dense CSS code using layer code construction}
\label{sec:dense-reduce-layer-code}

In this section, we analyze the distance of random dense CSS codes using the layer code construction in~\Cref{sec:dense-code-reduce}.
Starting from a random dense CSS code $C$, we will apply the layer code construction \Cref{sec:dense-code-reduce} to obtain a code $\cC$.

The original analysis in~\cite{williamsonLayerCodes2023} on the distance of $\cC$ is mainly based on the cleaning approach. They obtained the distance lower bound of $\Omega(d \min(n_X,n_Z)/w)$.
Applying their lower bound directly to our random dense code will only provide us with a distance of $\Omega(n)$.
We will improve the analysis for layer codes in the random dense code setting and provide a lower bound of order $\Omega(n^2/\log n)$.

The idea is to introduce auxiliary codes $\cC', \cC''$
  where $\cC'$ is a contraction of $\cC$,
    $\cC''$ is a contraction of $\cC'$,
    and $C$ is a contraction of $\cC''$.
We then show that $d \ge d' \ge d'' n$ and $d'' = \Omega(n/\log n)$,
  where $d, d',d''$ are the distances of $\cC, \cC', \cC''$ respectively.
We will focus on analyzing the X-distance,
  as the Z-distance follows by symmetry.

We first define the contraction process from $\cC$ to $\cC'$.
We first contract each Z layer to a point, and each qubit layer into a line that has the structure of a repetition code.
The connections between the new structures will be naturally induced by the original layer code construction.
If there is a connection between the original X layer and a qubit layer at some point in the layer code, the new code will also connect the point in the layer and the corresponding qubit at the same height.
Similarly, the connection between the Z check point and the qubit line follows the same rule.
An example of the contraction process can be seen in~\Cref{fig:layer-code-compress}.

\begin{figure}[H]
    \centering
    \includegraphics[width=0.6\linewidth]{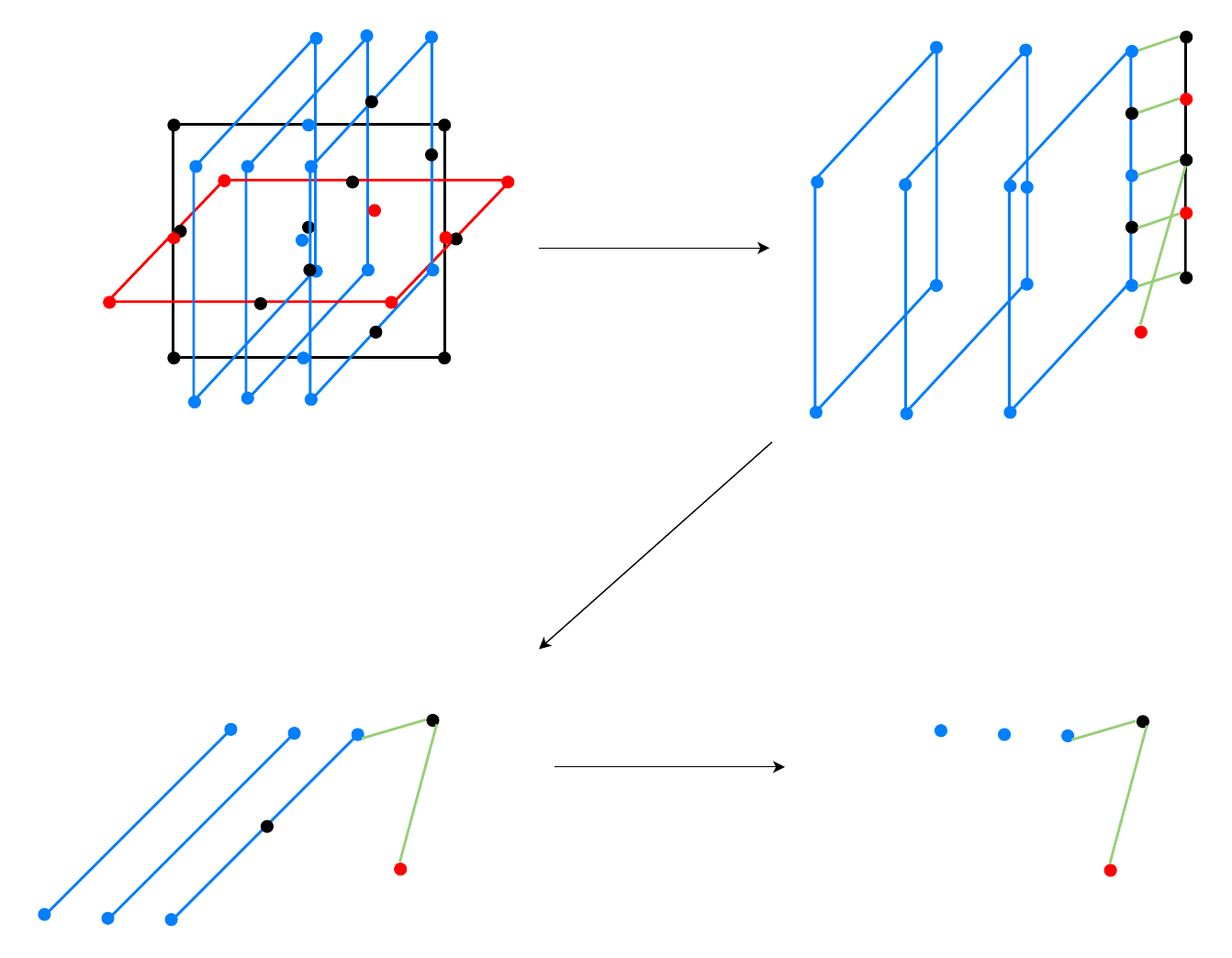}
    \caption{An example of layer code contraction, the green lines represent the connection between different layers of checks and qubits. From the upper left to the lower right, the graph corresponds to the structure of code $\cC,\cC',\cC'', C$. Note that in $\cC'$, each Z check is assigned to a different height on the qubit line.
    }
    \label{fig:layer-code-compress}
\end{figure}

This will provide us with a new code $\cC'$ and we denote its distance by $d'$.
We claim that $d' \leq d$.
At a high level, this can be seen as the contraction process will only compress the strings of the logical operators in the layer code, thus the distance will not increase.
To be more concrete, we will define the following process to map codewords $c\in \cC$ to codewords $c'\in \cC'$.
In order to properly define codewords in the new code, we will perform cleaning in the qubit layer before contraction.
Note that the flipping process will not affect the structure in the Z-check layer.
We will choose one column in each qubit layer as the canonical column, and remove the support of the code on other columns.
As shown in~\Cref{fig:line-contract}, if we denote the support of the codeword $c$ in each column of qubit layer $q$ as $c^{(i)}_l(q)$, we have that $c'(q)=\sum_i c^{(i)}_l(q)$, by the triangle inequality, we can show that $|c'|\leq \sum_i|c^{(i)}_l|\leq |c|$.
Thus we can contract the qubit layer to the corresponding line, and take the resulting vertical string $c'(q)$ as the support.
The contraction in the X layer is more direct, as the codewords do not have any open boundary in the layer, we can directly contract its support into one point.
The distance argument will follow immediately.

\begin{figure}[H]
    \centering
    \begin{tikzpicture}
        \begin{scope}
           \foreach \row in {0,1,2}
        {
            \foreach \column in {0,1,2}
            {
                \pgfmathtruncatemacro{\rw}{2*\row}
                \pgfmathtruncatemacro{\cl}{2*\column}
                \node at (2*\row,2*\column)[circle,fill,inner sep=1.5pt] (v\rw\cl){};
                \pgfmathtruncatemacro{\cl}{2*\column+1}
                \node at (2*\row,2*\column+1)[circle, draw,red] (v\rw\cl) {};
                \pgfmathtruncatemacro{\rw}{2*\row+1}
                \node at (2*\row+1,2*\column+1)[circle,fill,inner sep=1.5pt] (v\rw\cl){};
                \pgfmathtruncatemacro{\cl}{2*\column}
                \node at (2*\row+1,2*\column) [circle, draw,blue](v\rw\cl) {};
            };
        };

        \foreach \row in {0,1,...,4}
        {
            \foreach \column in {0,1,...,3}
            {
                \pgfmathtruncatemacro{\cl}{\column+1}
                \draw (v\row\column)--(v\row\cl);
            }
        }

        \foreach \column in {0,1,...,4}
        {
            \foreach \row in {0,1,...,3}
            {
                \pgfmathtruncatemacro{\rw}{\row+1}
                \draw (v\row\column)--(v\rw\column);
            }
        }

        \filldraw [white] (-0.5,4.5) rectangle (5.5,5.5);
        \filldraw [white] (4.5,-0.5) rectangle (5.5,5.5);
        \draw[blue,very thick] (4.5,1.5)--(4,1)--(4,3)--(4.5,3.5);
        \draw[blue,very thick] (1.5,1.5)--(2,1)--(2,3)--(0,3)--(0,4);
        \end{scope}
        \begin{scope}[shift={(6.5,0)}]
           \foreach \row in {0,1,2}
        {
            \foreach \column in {0,1,2}
            {
                \pgfmathtruncatemacro{\rw}{2*\row}
                \pgfmathtruncatemacro{\cl}{2*\column}
                \node at (2*\row,2*\column)[circle,fill,inner sep=1.5pt] (v\rw\cl){};
                \pgfmathtruncatemacro{\cl}{2*\column+1}
                \node at (2*\row,2*\column+1)[circle, draw,red] (v\rw\cl) {};
                \pgfmathtruncatemacro{\rw}{2*\row+1}
                \node at (2*\row+1,2*\column+1)[circle,fill,inner sep=1.5pt] (v\rw\cl){};
                \pgfmathtruncatemacro{\cl}{2*\column}
                \node at (2*\row+1,2*\column) [circle, draw,blue](v\rw\cl) {};
            };
        };

        \foreach \row in {0,1,...,4}
        {
            \foreach \column in {0,1,...,3}
            {
                \pgfmathtruncatemacro{\cl}{\column+1}
                \draw (v\row\column)--(v\row\cl);
            }
        }

        \foreach \column in {0,1,...,4}
        {
            \foreach \row in {0,1,...,3}
            {
                \pgfmathtruncatemacro{\rw}{\row+1}
                \draw (v\row\column)--(v\rw\column);
            }
        }

        \filldraw [white] (-0.5,4.5) rectangle (5.5,5.5);
        \filldraw [white] (4.5,-0.5) rectangle (5.5,5.5);
         \draw[blue,very thick] (4.5,1.5)--(4,1)--(2,1)--(1.5,1.5);
        \draw[blue,very thick] (4.5,3.5)--(4,3)--(4,4);
        \filldraw[blue,nearly transparent](2,1) rectangle (4,3);
        \filldraw [blue, nearly transparent](0,3) rectangle (4,4);
       \end{scope}
        \begin{scope}[shift={(13,0)}]
           \foreach \row in {0}
        {
            \foreach \column in {0,1,2}
            {
                \pgfmathtruncatemacro{\rw}{2*\row}
                \pgfmathtruncatemacro{\cl}{2*\column}
                \node at (2*\row,2*\column)[circle,fill,inner sep=1.5pt] (v\rw\cl){};
                \pgfmathtruncatemacro{\cl}{2*\column+1}
                \node at (2*\row,2*\column+1)[circle, draw,red] (v\rw\cl) {};
            };
        };

        \foreach \row in {0}
        {
            \foreach \column in {0,1,...,3}
            {
                \pgfmathtruncatemacro{\cl}{\column+1}
                \draw (v\row\column)--(v\row\cl);
            }
        }
        \filldraw [white] (-0.5,4.5) rectangle (0.5,5.5);
        \draw[blue,very thick] (0,3)--(0,4);
        \draw [blue,very thick] (0,3)--(0.5,3.5);
        \draw [blue,very thick] (0,1)--(0.5,1.5);
        \draw [blue,very thick] (0,1)--(-0.5,1.5);
       \end{scope}
    \end{tikzpicture}
    \caption{If we choose the rightmost column as the representative column for contraction, by flipping the X stabilizers in the qubit layer (marked in blue in the right figure), we can find a representative with its support on row only in the rightmost column. After flipping, we map the support of the rightmost column to the qubit line in $\cC'$, which is our new codeword $c'$.
    }
    \label{fig:line-contract}
\end{figure}

From $\cC'$ we then further contract each X check from a plane into a line.
The commutation relation between the new Z and X stabilizers is guaranteed by the structure in~\Cref{fig:defect-struct-contract}.
We denote the resulting code as $\cC''$ and its distance $d''$.

The relation between $\cC'$ and $\cC''$
  is that $\cC'$ can be obtained from the thickening and choosing heights procedure in~\cite{hastingsQuantumWeightReduction2023} on code $\cC''$.
By the known result of thickening and choosing heights in~\cite{hastingsQuantumWeightReduction2023}, we have that $d'\geq d''n$.

\begin{figure}[H]
    \centering
    \begin{tikzpicture}
            \node at (-2,-2) [circle, draw,blue] (z11){};
            \node at (0,-2) [circle, draw,blue] (z12){};
            \node at (2,-2) [circle, draw,blue] (z13) {};
            \node at (2,0) [circle,fill,inner sep=1.5pt] (qz3){};
            \node at (-0.5,2) [circle, draw,red] (x2){};
            \node at (-2.5,0)  [circle,fill,inner sep=1.5pt] (q1){};
           \node at (-1,-2) [circle,fill,inner sep=1.5pt] (qc1){};
           \node at (1,-2) [circle,fill,inner sep=1.5pt] (qc2){};
          \draw (x2)--(qc1);
           \draw (x2)--(q1);
           \draw (q1)--(z11);

           \draw (x2)--(qc2);
           \draw (z13)--(qz3);
           \draw (qz3)--(x2);
           \draw (z11)--(qc1);
           \draw (qc1)--(z12);
           \draw (z12)--(qc2);
           \draw (qc2)--(z13);

        \end{tikzpicture}
    \caption{The commutation gadget in code $\cC''$, which corresponds to the gadget in~\Cref{fig:defect-struct} after contraction.}
    \label{fig:defect-struct-contract}
\end{figure}

Below we analyze the distance $d''$ of the code $\cC''$.
The analysis relies on the observation that $\cC''$
  can be directly related to the original dense code $C$,
  where each X stabilizer in $C$ is split into a line of X stabilizers in $\cC''$.

\begin{lemma}\label{lem:plain-dist}
  With probability $1-\exp(-\Omega(n))$, the distance of the $\cC''$ code $d''\geq \Omega(n/\log n)$.
\end{lemma}

\begin{proof}
    The difference between $\cC''$ and the dense code $C$ is that the X stabilizer is split into a line, and each qubit is assigned to a different X check on the line.
    Therefore, for every codeword in $\cC''$, we only have to remove its support in the X stabilizer line to obtain a codeword in the dense code $C$.

  Note that the codeword $c$ of $C$ is also a valid codeword of the code $\cC''$ if we take the canonical map, which maps the qubits in $C$ to the corresponding qubits in $\cC''$.
  However, $c$ can possibly find a lower-weight representative $c''$ in $\cC''$.
  Consider one X stabilizer line. If the codeword assigns $p$ consecutive 1s on the qubits connected to the stabilizer, we are able to flip the X check on the line to obtain a codeword of weight $\leq|c|-p+2$. An example of the process is shown in~\Cref{fig:chain-flip}. We can use a string $w\in\FF_2^n$ to indicate the flipped qubits. If the qubits connected to the X check are indicated by $u$, $w$ must be a segment of $u$.

\begin{figure}
    \centering
    \includegraphics[width=0.8\linewidth]{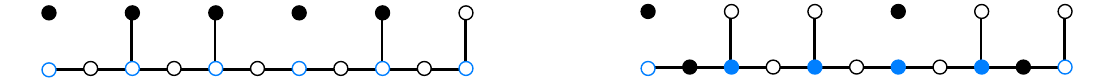}
    \caption{On the left-hand side, the solid black qubits are assigned to 1. By flipping the solid blue X checks on the right, we can see that the weight of the codeword is reduced by 1 (5 to 4 in the picture). Here the codeword $c=111110$ can be decomposed as $w_1=011010$ and $w_2=100100$, as $w_1$ indicates the flipped positions. Note that it is necessary to have consecutive 1s in the qubits connected to the check.
    }
    \label{fig:chain-flip}
\end{figure}

 In general, write $c''=(e,a)$, where $e$ is its support on the original qubits and $a$ is its ancillary support. If $u_{k'}$ denotes the $k'$-th row of $H_X$, the resulting codeword has the form
\[
c=e+\sum_{k=1}^{D}w_k,\qquad
w_k=u_{k'}([i_k,j_k]),\qquad |e|+D\le |c''|,
\]
where each $w_k$ is a row restricted to a consecutive segment, with other positions set to zero.

  For a fixed $c$, we use $wt''(c)$ to denote the minimum size of its segment decomposition in $\cC''$. Formally, $wt''(c)$ can be defined as follows.
\[
wt''(c)=\min\left\{|e|+D:
c=e+\sum_{k=1}^{D}w_k,\quad
w_k=u_{k'}([i_k,j_k])\text{ for some }k',i_k\le j_k
\right\}.
\]

The preceding argument shows that $d_X''\le T$ implies the existence of a nontrivial original X codeword $c$ with $wt''(c)\le T$.

We first fix $H_X$ and count all possible decompositions of size at most $T$. A segment chooses a row and two endpoints, giving at most $r_Xn(n+1)/2$ possibilities, while the residual support $e$ is a sum of singleton vectors, with $n$ possibilities for each singleton. Hence, with
$M=n+r_X\frac{n(n+1)}2=O(n^3)$,
there are at most $(M+1)^T$ resulting vectors: each decomposition is a sum of at most $T$ segments and singletons, and can be padded to length $T$ with zeros.

Now sample the $r_Z$ Z-check rows independently and uniformly from $\operatorname{row}(H_X)^\perp$. For each fixed candidate $c\notin\operatorname{row}(H_X)$, the probability that $H_Zc^T=0$ is at most $2^{-r_Z}$ . Via a union bound on all candidate vectors,
$$
\Pr[d_X''\le T\mid H_X]
\le\frac{(M+1)^T}{2^{r_Z}}=O\left( \frac{n^{3T}}{2^{cn}} \right).
$$
This bound is uniform in $H_X$, so it also holds without conditioning.
    setting $T=0.01n/\log n$, we can show that the term on the right-hand side is of order $\exp(-\Omega(n))$ as desired.
\end{proof}

Combining the above arguments, we have that with probability $1-\exp(-O(n))$, the distance of the layer code $d\geq d'\geq d''n\geq \Omega(n^2/\log n)$, proving the desired lower bound.

\section{Discussion}

\subsection{On optimality}
\label{sec:optimal}

We prove that several of our parameters are optimal.

\subsubsection{Qubit and check weight}
\label{sec:optimal-weight}
In our construction,
  we obtain qubit weight $6$ and check weight $5$.
(For comparison, Hastings' construction yields qubit weight $8$ and check weight $5$.)
It is shown in~\cite{wangCheckweightconstrainedQuantumCodes2026} that for codes with check weight $\leq 3$, the code cannot have linear distance and dimension simultaneously.
Under the restricted case of check weight 4 and every qubit connected to exactly 2 X-checks and 2 Z-checks,
they showed the code saturates a BPT-style bound $k d^2 \leq O(n)$.
In the following two theorems, we improve the no-go results for constructing codes with linear distance and dimension from low-weight codes.

In the first theorem, we show that any CSS code with each qubit incident to at most $2$ X/Z-checks cannot be a good LDPC code.

\begin{theorem}[Bounds on qubit weight]
  \label{thm:qubit-weight-bound}
  For an $[[n,k,d]]$ CSS code such that for every qubit $q$, $\min\{\deg_X(q),\deg_Z(q)\}\leq 2$, we have the following distance-dimension tradeoff:
  \[
                         k d=O(n\log n).
\]
\end{theorem}

In the second theorem, we show that any quantum LDPC code with check weight at most $4$ cannot be a good LDPC code.
\begin{theorem}
  \label{thm:check-weight-bound}
  For an $[[n,k,d]]$ CSS code with $w_X,w_Z\leq 4$, and with bounded qubit weight $\Delta$, we have the following distance-dimension tradeoff:
  \[
  k\log(d)\leq O(\Delta(n-k)).
  \]
\end{theorem}

We will defer the proofs of both theorems to Appendix~\ref{sec:lower-bound-proof}. Combining the two theorems,
in order to obtain a good LDPC code, we must have qubit weight at least 3 for both X and Z checks, 
and check weight at least 5 for either X or Z checks. These exact parameters are obtained by our construction, showing our
weight parameters are optimal.

\subsubsection{Qubit blowup}

In our construction,
  the resulting code uses $O(n w^2 \log w)$ qubits.
(For comparison, Hastings' construction requires a polynomial in $w$ of higher degree.)
We believe that the qubit blowup is optimal when $w = o(n)$,
  and we now explain the intuition behind this belief.
(When $w = \Theta(n)$, we can apply the dense code construction \Cref{sec:dense-code-reduce},
  where the blowup is $O(n^3)$.)

The reasoning again comes from a geometric perspective on the code.
Essentially, we believe $w^2 \log w$ is optimal
  for sparsifying a 2D cellular complex.
The factor $w^2$ is natural,
  since it comes from the fact that there are $O(w^2)$ squares
  incident to each qubit and check
  in the process of building the underlying square complex \Cref{sec:square-complex}.
The factor $\log w$ is also essential.
Ultimately, sparsifying a 2D cellular complex
  requires performing a process like sorting.
This is believed to be both necessary and sufficient.
The construction \cite{berdnikov2022parsimonious}
  relies on a sorting process
  which labels each vertex with an $O(\log w)$-bit string.

A more concrete fact is that the coning result \Cref{lemma:cone-simplicial-complex,lemma:cone-cell-complex}
  does incur an $O(\log w)$ factor.
This can be lower bounded by the fact that
  there are $w^{\Theta(w)}$ graphs with degree $3$ and $O(w)$ vertices.
In the graph-grammar viewpoint~\cite{sleator1992short},
each 2-cell can be thought of as a transformation.
The goal of coning is to derive the graph from a point.
By careful counting,
  it is shown in \cite{sleator1992short} that
  the number of graphs that can be derived from $k$ transformations
  is at most $2^{O(k)}$.
This means to reach all $w^{\Theta(w)}$ graphs,
  we need at least $k = \Omega(w \log w)$ transformations.
This suggests that the $\log w$ factor is necessary,
  whenever we apply a coning-like process.

\begin{remark}
  Note that similar log factors appear in other contexts.
  For example, the paper on building manifolds with codes
    contains a result on the decongestion lemma \cite[App A]{freedmanBuildingManifoldsQuantum2021}.
  The result incurs a $\log^2 w$ factor,
    which can likely be improved to $\log w$.

  The decongestion lemma was later used in several studies
  \cite{saboWeightReducedStabilizerCodes2024,cross2024improved,williamson2024low,ide2025fault,cowtanParallelLogicalMeasurements2025}.
  This implies that the overhead can be reduced from $\log^3 w$ to $\log w$
    in these studies.
\end{remark}

\subsubsection{Distance}

The distance of the weight-reduced code $\Omega(dw)$ is also believed to be optimal
  under the current construction framework.
The intuition is that the factor $w$ arises from
  the distance of the local code $\cC^q$,
  which imposes an inherent limitation.
One might hope to achieve a distance $O(w^2)$,
  since each local code $\cC^q$ has $O(w^2)$ qubits.
However, this seems difficult to achieve,
  since $\cC^q$ is not just an arbitrary code:
  it plays an important role in the sparsification process
  by allowing the same logical operator to be accessed by $w$ distinct X checks simultaneously.
This requirement suggests an inequality of the form
  $n^q \ge d^q w$,
  where $n^q$ is the number of qubits in $\cC^q$,
  and $d^q$ is its distance.
This inequality implies that $d^q \le O(w)$
  when $n^q = \Theta(w^2)$,
  which suggests that the distance of the weight-reduced code
  is at most $O(dw)$.

\subsection{On decoder}

We comment on how one can decode the weight-reduced code.
The strategy has been outlined in
  \cite{hastingsFiberBundleCodes2021}
  and worked out in a specific example
  in \cite{eggerickx2025almost}.
The key point is that
  whenever two codes are related by a (sparse) chain homotopy equivalence,
  the decoding problem for the new code can be reduced
  to the decoding problem for the original code.
Since our construction is based on homotopy operations,
  which naturally induce a chain homotopy equivalence,
  the decoding problem for our weight-reduced code
  can therefore be solved by reducing it to decoding the original code.

\subsection{On logical operator measurement}
\label{sec:logical-measurement}
There has been a series of recent works on measuring logical operators
  \cite{williamson2024low,ide2025fault,cowtanParallelLogicalMeasurements2025}.
We can improve these results
  by reducing the problem of measuring logical operators
  to a weight reduction problem.

In \cite{williamson2024low,ide2025fault},
  the task is to measure a single logical operator with support of size $W$
  in an arbitrary quantum LDPC code.
We may view the logical operator as introducing a new X-check or Z-check.
Applying our weight reduction routine
  reduces the check weights to constants,
  allowing them to be measured fault-tolerantly.
In this case,
  there is only one check with large weight $W$,
  so the procedure simplifies to applying coning, \Cref{lemma:cone-cell-complex},
  at that check.
This yields a code with constant qubit and check weight,
  and the resulting code has $n + O(W \log W)$ qubits,
  improving on the previous bound of $n + O(W \log^3 W)$.

In \cite{cowtanParallelLogicalMeasurements2025},
  the goal is to measure $t$ logical operators in parallel,
  each with support $\le W$.
Without loss of generality,
  we may assume the logical operators are all of X-type.
This introduces $t$ X-checks of weight $\le W$,
  and the qubits incur an additional $\le tW$ check incidences.
Our weight reduction procedure then produces a code with
  $n + O(t W \log W) + O(t W) + O(t W \log t)
  = n + O(t W (\log t + \log W))$ qubits.
\begin{itemize}
  \item $O(t W \log W)$ comes from coning at the X-checks,
        because there are $t$ checks of weight $\le W$.
  \item $O(t W)$ comes from product structure at the qubits,
        each check incidence increases some $\cS_{n_x}$ by $1$,
        which contributes $O(1)$ to the qubit count.
  \item $O(t W \log t)$ comes from coning at the Z-checks,
        because each Z-check is incident to $O(1)$ qubits,
        and each qubit is incident to $\le t$ X-checks.
        This means the graph $\partial \cX^z$ has degree $\le t$,
        so coning at the Z-checks incurs an $O(\log t)$ factor.
\end{itemize}
This improves on the earlier result $n + O(t W (\log t + \log^3 W))$.

Finally, to preserve distance,
  the link graph of the checks must satisfy the expansion condition.
As discussed in \Cref{sec:distance},
  this can be achieved by imposing an expander graph.

\section*{AI usage disclosure}
The paper's main weight reduction construction was developed and written by humans. The results on the lower bounds for qubit and check weight were generated
through two one-shot conversations with ChatGPT 5.6-Sol and 6-Astra respectively. A series of ChatGPT models (from 5.4 to 6.1-Sol) were used to improve the writing and figures in the paper.

\bibliographystyle{unsrt}
\bibliography{references.bib}

\appendix
\section{Omitted proofs of lower bounds}
\label{sec:lower-bound-proof}
\subsection{Proof of \Cref{thm:qubit-weight-bound}}
We restate the theorem here quantitatively, taking logarithms to base two in this subsection.

\begin{theorem}
For an $[[n,k,d]]$ CSS code such that for every qubit $q$, $\min\{\deg_X(q),\deg_Z(q)\}\leq 2$
we have the following distance-dimension tradeoff:
\begin{equation}
                 k d\le
 4n\left\lceil\log(2n)\right\rceil .                 
\end{equation}
\end{theorem}

The proof partitions the qubits into a set $U$ of low $Z$-check degree and a set $W$ of low $X$-check degree.
We choose a basis of $X$ logical operators supported in $U$ whose representatives have total Hamming weight
$O(|U|\log n)$, and similarly a basis of $Z$ logical operators supported in $W$ with total weight $O(|W|\log n)$.
Since the number of independent logical operators is $k$, the distance is bounded by $O(n \log n/k)$.

We introduce notation for logical operators supported on a subset of qubits.
We use the classical codes $C_x,C_z$ and the code dimension $k$ from~\Cref{sec:css-codes}.
For $U\subseteq [n]$, write $\F2^U$ for the coordinate subspace of $\F2^n$ supported in $U$.
For $\alpha\in\{x,z\}$, let $H_{\alpha,U}$ denote the submatrix of $H_\alpha$ formed by the columns indexed by $U$.

\begin{definition}
\label{def:supported-logical-dimensions}
For $U\subseteq [n]$, define the numbers of independent $X$ and $Z$ logical operators
that admit representatives supported entirely in $U$ by
\[
  \ell_x(U)=\dim\frac{C_x\cap\F2^U}{C_z^\perp\cap\F2^U},\qquad
  \ell_z(U)=\dim\frac{C_z\cap\F2^U}{C_x^\perp\cap\F2^U}.
\]
The subspaces are well defined by $C_z^\perp\subseteq C_x$ and $C_x^\perp\subseteq C_z$.
\end{definition}

The following lemma shows that the supported logical dimensions on complementary subsets of qubits sum to $k$.
\begin{lemma}
\label{lem:cleaning-identity}
For every partition $[n]=U\sqcup W$,
\[
  \ell_x(U)+\ell_z(W)=k.
\]
\end{lemma}

\begin{proof}
We order the columns of the check matrices according to the partition $[n]=U\sqcup W$, so that
\[
  H_x=\begin{bmatrix}H_{x,U}&H_{x,W}\end{bmatrix},\qquad
  H_z=\begin{bmatrix}H_{z,U}&H_{z,W}\end{bmatrix}.
\]
The vectors in $C_x$ supported in $U$ correspond to $\ker(H_{z,U})$, hence
\[
  \dim(C_x\cap\F2^U)=|U|-\rk(H_{z,U}).
\]
The $X$ stabilizer space $C_z^\perp$ is the row space of $H_x$ and has dimension $\rk H_x$.
Restriction to $W$ maps this space onto the row space of $H_{x,W}$,
and its kernel is $C_z^\perp\cap\F2^U$.
By rank-nullity,
\[
  \dim(C_z^\perp\cap\F2^U)=\rk H_x-\rk(H_{x,W}).
\]
Subtracting these dimensions as in~\Cref{def:supported-logical-dimensions} gives
\begin{align*}
  \ell_x(U)
  &=|U|-\rk(H_{z,U})-\rk H_x+\rk(H_{x,W}).
\end{align*}
Similarly, we obtain
\begin{align*}
  \ell_z(W)
  &=|W|-\rk(H_{x,W})-\rk H_z+\rk(H_{z,U}).
\end{align*}
Adding the two equations cancels the ranks of the restricted matrices and gives
\[
  \ell_x(U)+\ell_z(W)
  =n-\rk H_x-\rk H_z=k.
\]
\end{proof}

The following lemma shows that there is a short basis for the kernel of a binary matrix with column weight at most two.

\begin{lemma}
\label{lem:short-cycle-basis}
Let $M$ be a binary matrix with $m\geq 1$ columns, each of weight at most two.
Then $\ker M$ has a basis $\cB$ satisfying
\[
  \sum_{c\in\cB}|c|\leq 4m\left\lceil\log(2m)\right\rceil.
\]
For $m=0$, the kernel has the empty basis.
\end{lemma}

\begin{proof}
For an $n\times m$ binary matrix $M$ with column weight at most two, we can associate it with graph $G$ as follows:
We first construct a matrix $M'$ by appending one row equal to the sum of all rows in $M$. Note that after this operation, every column now contains
an even number of ones, and $\ker M'=\ker M$. We now remove the zero rows and columns in the matrix $M'$. Each removed zero column
corresponds to an indicator vector of the column in the kernel of $M'$. Removing the zero rows does not change the kernel of $M'$.

We can view the remaining matrix $M_G$ as an incidence matrix
of a graph $G$ (possibly with parallel edges), where each row corresponds to a vertex, and each column corresponds to an edge.
The kernel $\ker M_G$ is the cycle space of $G$: its vectors represent edge sets in which every vertex has even degree.
We use the following definition.

\begin{definition}[Cycle basis]
A \emph{cycle basis} of $G$ is a collection $\cB=\{C_1,\ldots,C_r\}$ of simple cycles whose edge-indicator vectors form a basis of $\ker M_G$ over $\F2$.
\end{definition}

Let $t$ be the number of removed zero columns, so $G$ has $m-t$ edges and at most $2(m-t)$ vertices.
The construction in~\cite[Theorem~4.4]{kavitha2009cycle}, with the constants estimated in Lemma~4.3,
gives a cycle basis $\cB_G$ with
\[
  \sum_{C\in\cB_G}|C|\leq 4(m-t)\left\lceil\log(2m)\right\rceil.
\]
Extending these vectors by zeros and adjoining the $t$ unit vectors for the removed columns gives a basis $\cB$ of $\ker M$ satisfying
\[
  \sum_{c\in\cB}|c|
  =t+\sum_{C\in\cB_G}|C|
  \leq t+4(m-t)\left\lceil\log(2m)\right\rceil
  \leq 4m\left\lceil\log(2m)\right\rceil.
\]
\end{proof}

\begin{proof}[Proof of~\Cref{thm:qubit-weight-bound}]
We partition the qubits into a set $U$ of low $Z$-check degree and a set $W$ of low $X$-check degree:
\[
  U=\{q\in[n]:\deg_Z(q)\leq 2\},\qquad W=[n]\setminus U.
\]
For every $q\in W$,  $\min\{\deg_X(q),\deg_Z(q)\}\leq 2$ implies $\deg_X(q)\leq 2$.

The submatrix $H_{z,U}$ has column weight at most two, and its kernel is $C_x\cap\F2^U$.
By~\Cref{lem:short-cycle-basis}, this kernel has a basis $\cB_x$ whose total Hamming weight is at most
$4|U|\lceil\log(2n)\rceil$. This gives us a set of logical $X$ operators supported in $U$.
 Since each selected vector has Hamming weight at least $d_x$, we have that
\[
  d_x\ell_x(U)\leq 4|U|\left\lceil\log(2n)\right\rceil.
\]

Similarly, the submatrix $H_{x,W}$ has column weight at most two and kernel $C_z\cap\F2^W$.
Applying~\Cref{lem:short-cycle-basis} and selecting a low-weight basis of
$(C_z\cap\F2^W)/(C_x^\perp\cap\F2^W)$ gives
\[
  d_z\ell_z(W)\leq 4|W|\left\lceil\log(2n)\right\rceil.
\]
Finally, $d=\min\{d_x,d_z\}$ and~\Cref{lem:cleaning-identity} imply
\begin{align*}
  kd
  &=d\bigl(\ell_x(U)+\ell_z(W)\bigr)\\
  &\leq d_x\ell_x(U)+d_z\ell_z(W)\\
  &\leq 4\bigl(|U|+|W|\bigr)\left\lceil\log(2n)\right\rceil\\
  &=4n\left\lceil\log(2n)\right\rceil.
\end{align*}
\end{proof}

\subsection{Proof of \Cref{thm:check-weight-bound}}
We restate the theorem here quantitatively.

\begin{theorem}
  \label{thm:check-weight-bound-quant}
  For any $[[n,k,d]]$ CSS code with $w_X,w_Z\leq 4$, and with bounded qubit weight $\Delta$, we have the following distance-dimension tradeoff:
  \[ 
k\left\lceil\log_2 d\right\rceil
\;\le\;
8\Delta(n-k).
  \]
\end{theorem}

To prove the theorem, we will use the following qubit reduction lemma.
\begin{lemma}
  \label{lem:qubit-reduction}
  For any $[[n,k,d]]$ CSS code with $w_X,w_Z\leq 4$, and with bounded qubit weight $\Delta$, one can obtain another $[[n',k',d']]$ CSS code with
  \begin{align*}
    n'\leq n-\frac{k}{8\Delta},\quad k'=k, \quad d'\geq d/2.
  \end{align*}
  and with the same bound on the qubit and check weights.
\end{lemma}

\begin{proof}
  If a qubit $q$ is incident to exactly one $X$-check $S_X$,
  \begin{align*}
    S_X=X_q\prod_{j\in S_X\setminus\{q\}}X_j,
  \end{align*}
  we introduce the following local unitary $U=\prod_{j\in S_{X}\setminus\{q\}}\CNOT_{ q\to j}$, and we have that $USU^{\dagger}=X_q$.
  For other stabilizers that are incident to $q$, we note that by our assumption, they should be $Z$-checks. For every $Z$ check $S_Z$ incident to $q$, 
  note that $S_Z$ commutes with $S_X$, thus $US_ZU^\dagger$ also commutes with $US_XU^{\dagger}=X_q$, implying that the support of $US_ZU^\dagger$ does not contain $q$, while
  every other qubit in the support of $S_Z$ is still in the support of $US_ZU^\dagger$ by the property of the CNOT gate.  
  Therefore, the new code space $U\cC U^\dagger=\ket{+}\bra{+}\otimes\cC'$. We can use a similar argument for the case when a qubit is incident to exactly one $Z$-check.
  As such, we can construct a new code $\cC'$ with one fewer qubit, and the same number of logical qubits.
  
  We can apply in parallel the above procedure to a set of qubits that are incident to exactly one $X$-check or one $Z$-check, provided the corresponding stabilizers are pairwise disjoint.

  We categorize the qubits of the code into low weight and high weight qubits. Each qubit that is nontrivially involved in the code is incident to
  at least one $X$-check and one $Z$-check. A qubit is low weight if it has total weight $\leq 3$, and high weight otherwise. Let $L$ be the number of low-weight qubits.
  We have the following bounds:
  \begin{align}
    \sum_{q}\deg(q)=\sum_{S}\wt(S)\leq 4(n-k),\label{eq:wt-up}\\
    \sum_{q}\deg(q)\geq 2L+4(n-L)=4n-2L.\label{eq:wt-low}
  \end{align}
  The first is by summing over the edges of the Tanner graph, and the second is that every low weight qubit contributes at least $2$ to the sum, while high weight qubit contributes at least $4$.
  Thus we can conclude that $L\geq 2k$.

  Now we count the pairwise disjoint stabilizers in the neighborhood of the low weight qubits, and we can apply in parallel the above procedure to remove the low weight qubits.
  Since each stabilizer is of weight $\leq 4$, we can note that there are at least $L/4$ distinct stabilizers in the neighborhood of the low weight qubits.
  For each of the stabilizers, as it is incident to at most $4$ qubits, and each qubit is incident to at most $\Delta$ stabilizers, each stabilizer can
  intersect with at most $4\Delta$ other stabilizers. Thus we can find a set of pairwise disjoint stabilizers of size at least $L/(4\cdot 4\Delta)=L/(16\Delta)\geq k/(8\Delta)$, 
  and we can apply in parallel the above procedure to remove the low weight qubits in the neighborhood of these disjoint stabilizers, and we can obtain a new code $\cC'$ with at most $n-k/(8\Delta)$ qubits.

  For the distance of the new code $\cC'$, we consider a minimal logical codeword $P'\in\cC'$, and conjugate it with the local unitaries $U_i$ of the previous process.
  For the $X$ removal process, as the stabilizers are pairwise disjoint, each unitary $U_i$ will introduce at most one additional $Z$ on the removed qubit, 
  thus if $P'$ is an $X$ logical operator, the conjugated operator $P=U^\dagger P'U$ will be a logical operator of the original code $\cC$, and $\wt(P)\leq 2\wt(P')$. A similar
  argument holds for the $Z$ removal process. Thus we have that the distance of the new code $\cC'$ is at least $d/2$.

\end{proof}

Now by iteratively applying~\Cref{lem:qubit-reduction} for $T=\lceil \log(d)\rceil$ times, we can obtain a new code with $n-\frac{Tk}{8\Delta}$ qubits and dimension $k$. Since we have that $n-\frac{Tk}{8\Delta}\geq k$, rearranging the terms gives
the desired bound in~\Cref{thm:check-weight-bound-quant}.

\begin{remark}
  The above process also works for larger check weight, as from~\Cref{eq:wt-up} and~\Cref{eq:wt-low} can directly obtain that $2L\geq wk-(w-4)n$. For $w=5$, the bound is trivial
  for $k\leq n/5$, which implies a limitation on the rate of our construction.
\end{remark}

\end{document}